%% file: paper.tex
\documentclass[envcountsect,envcountsame,runningheads,a4paper]{llncs}

\include{macros}

\usepackage{enumerate}

\newcommand{\mathsc}[1]{{\normalfont\textsc{#1}}}
\renewcommand{\problemname}[1]{{\mathsc{#1}}}
\newcommand{\algorithmname}[1]{{\mathsc{#1}}}

\mathchardef\mhyphen="2D

\newcommand{\EndRing}{\problemname{EndRing}} 
\newcommand{\OneEnd}{\problemname{OneEnd}} 
\newcommand{\IsogPath}{\ell\mhyphen\problemname{IsogenyPath}} 
\newcommand{\Isogeny}{\problemname{Isogeny}} 
 
\newcommand{\Rich}{\algorithmname{Rich}} 
\newcommand{\Reduce}{\algorithmname{Reduce}} 
\newcommand{\Divide}{\algorithmname{Divide}} 
\newcommand{\Saturate}{\algorithmname{Saturate}} 
\newcommand{\SaturateRam}{\algorithmname{SaturateRam}} 
\newcommand{\hatZ}{\widehat{\Z}} 
\newcommand{\hatQ}{\widehat{\Q}}
\newcommand{\order}{\mathcal{O}} 
\newcommand{\hatorder}{\widehat{\order}}
\newcommand{\hatB}{\widehat{B}}
\newcommand{\lquo}{\backslash}
\newcommand{\cC}{\mathcal{C}}
\newcommand{\cD}{\mathcal{D}}
\newcommand{\cF}{\mathcal{F}}
\newcommand{\cG}{\mathcal{G}}
\newcommand{\CGL}{\algorithmname{CGL}} 
\newcommand{\Oracle}{\mathscr{O}}
\newcommand{\modN}{(\mathrm{mod}\ N)} 
\newcommand{\Success}{\mathsf{Success}} 
\newcommand{\triv}{\mathbf{1}}
\newcommand{\nb}{\mathcal{N}}

\DeclareMathOperator{\nrd}{nrd}

\DeclareMathOperator{\trd}{trd}
\DeclareMathOperator{\End}{End}
\DeclareMathOperator{\Isom}{Isom}
\DeclareMathOperator{\level}{lev}
\DeclareMathOperator{\Mod}{Mod}
\DeclareMathOperator{\Cosets}{Cosets}
\DeclareMathOperator{\Sets}{Sets}
\DeclareMathOperator{\El}{El}
\DeclareMathOperator{\catSS}{SS}
\DeclareMathOperator{\gr}{Graph}
\DeclareMathOperator{\edg}{edg}
\DeclareMathOperator{\Cyc}{Cyc}
\DeclareMathOperator{\Deg}{Deg}
\DeclareMathOperator{\Ad}{Ad}

\DeclareFontFamily{U}{matha}{\hyphenchar\font45}
\DeclareFontShape{U}{matha}{m}{n}{
      <5> <6> <7> <8> <9> <10> gen * matha
      <10.95> matha10 <12> <14.4> <17.28> <20.74> <24.88> matha12
      }{}
\DeclareSymbolFont{matha}{U}{matha}{m}{n}
\DeclareFontFamily{U}{mathx}{\hyphenchar\font45}
\DeclareFontShape{U}{mathx}{m}{n}{
      <5> <6> <7> <8> <9> <10>
      <10.95> <12> <14.4> <17.28> <20.74> <24.88>
      mathx10
      }{}
\DeclareSymbolFont{mathx}{U}{mathx}{m}{n}
\DeclareMathSymbol{\obot}         {2}{matha}{"6B}
\DeclareMathSymbol{\bigobot}       {1}{mathx}{"CB}

\begin{document}

\title{The supersingular Endomorphism Ring and \\ One Endomorphism problems are equivalent}
\titlerunning{Endomorphism Ring and One Endomorphism are equivalent}
\author{Aurel Page\inst{1} \and Benjamin Wesolowski\inst{2}}
\institute{
Univ. Bordeaux, CNRS, INRIA, Bordeaux INP, IMB, UMR 5251, F-33400 Talence, France
\and
ENS de Lyon, CNRS, UMPA, UMR 5669, Lyon, France
}

\maketitle

\begin{abstract}
The supersingular Endomorphism Ring problem is the following: given a supersingular elliptic curve, compute all of its endomorphisms.
The presumed hardness of this problem is foundational for isogeny-based cryptography.
The One Endomorphism problem only asks to find a single non-scalar endomorphism. We prove that these two problems are equivalent, under probabilistic polynomial time reductions.

We prove a number of consequences. First, assuming the hardness of the endomorphism ring problem, the Charles--Goren--Lauter hash function is collision resistant, and the SQIsign identification protocol is sound.
Second, the endomorphism ring problem is equivalent to the problem of computing
  arbitrary isogenies between supersingular elliptic curves, a result previously
  known only for isogenies of smooth degree.
Third, there exists an unconditional probabilistic algorithm to solve the
  endomorphism ring problem in time $\tilde O(p^{1/2})$, a result that
  previously required to assume the generalized Riemann hypothesis.

To prove our main result, we introduce a flexible framework for the study of isogeny graphs with additional information. We prove a general and easy-to-use rapid mixing theorem.
\end{abstract}


\section{Introduction}\label{sec:intro}

\noindent
The endomorphism ring problem lies at the foundation of isogeny-based cryptography. On one hand, its presumed hardness is necessary for the security of all cryptosystems of this family (see for instance the reductions in~\cite{Wes22}). On the other hand, many cryptosystems of this family can be proven secure if this problem (or some variant) is hard (the earliest example being~\cite{CGL09}).
Isogeny-based cryptography takes its name from the \emph{isogeny problem}. An isogeny is a certain kind of map between two elliptic curves, and the isogeny problem consists in finding such a map, given the two curves. Formalising the meaning of ``finding an isogeny'' can lead to several versions of the isogeny problem, the most prominent being the \emph{$\ell$-isogeny path problem}. In isogeny-based cryptography, one typically restricts to supersingular elliptic curves, for which this problem is believed to be hard.

Fix a supersingular elliptic curve $E$. 
An endomorphism of $E$ is an isogeny 
from $E$ to itself (or the zero morphism). The collection of all endomorphisms of $E$ forms the endomorphism ring $\End(E)$. The \emph{supersingular endomorphism ring problem}, or $\EndRing$, consists in computing $\End(E)$, when given $E$. Assuming the generalised Riemann hypothesis, this problem is equivalent to the $\ell$-isogeny path problem (see~\cite{Wes21}, and the earlier heuristic equivalence~\cite{EHLMP18}), cementing its importance in the field.

The endomorphism ring contains scalars $\Z \subseteq \End(E)$, simple elements which are always easy to compute. 
While $\EndRing$ asks to find all endomorphisms, it has long been believed that finding even a single non-scalar endomorphism is hard. We call this the \emph{one endomorphism problem}, or $\OneEnd$.
Unfortunately, former heuristic arguments suggesting that $\OneEnd$ should be as hard as $\EndRing$ do not withstand close scrutiny, and actually fail in simple cases. Yet, the connection between these two problems bears important consequences on the hardness of $\EndRing$, on its connection with variants of the isogeny problem, and on the security of cryptosystems such as the CGL hash function~\cite{CGL09} or the SQIsign digital signature scheme~\cite{DFKLPW20}.

\subsection{Contributions}
In this article, we prove the following theorem.

\begin{theorem}\label{thm:main}
The $\EndRing$ and $\OneEnd$ problems are equivalent, under probabilistic polynomial time reductions.
\end{theorem}

Formal definitions are provided in Section~\ref{sec:prelim}, and the proof is the object of Section~\ref{sec:the-main-reduction}. The reduction from $\OneEnd$ to $\EndRing$ is obvious, and the other direction is stated more precisely in Theorem~\ref{thm:main-more-precise}.
As a consequence of the main theorem, we prove the following:
\begin{itemize}
\item If $\EndRing$ is hard, then the CGL hash function is collision resistant
  (Theorem~\ref{thm:CGLcollres}), and the SQIsign identification scheme is
    sound (Theorem~\ref{thm:SQIsignsound}). Previous security proofs relied on the hardness of $\OneEnd$ (see~\cite[Theorem~1]{DFKLPW20}), or on
    flawed heuristic reductions (see~\cite[Algorithm~8]{EHLMP18}, and the flaws discussed Section~\ref{subsec:techoverview}). This is the object of
    Section~\ref{subsec:CGL} and Section~\ref{subsec:SQIsign}.
\item $\EndRing$ reduces to the isogeny problem
  (Theorem~\ref{thm:EndRing-Isogeny}). Here, the isogeny problem refers to the
    problem of finding \emph{any} isogeny between two elliptic curves. Previous
    results~\cite{EHLMP18,Wes21} only applied to isogenies of smooth degree (like the $\ell$-isogeny
    path problem), and were conditional on the generalised Riemann hypothesis. This is the
    object of Section~\ref{subsec:Isogeny-EndRing}.
\item There is an algorithm solving $\EndRing$ in expected time $\tilde
  O(p^{1/2})$ (Theorem~\ref{thm:solvingEndRing}), where $p > 0$ is the
    characteristic. Previous algorithms were conditional on the
    generalised Riemann hypothesis (via the conditional equivalence with the $\ell$-isogeny
    path problem~\cite{Wes21}; see also~\cite[Theorem~5.7]{EndRingGRH} for a more direct approach). Previous unconditional algorithms ran in time $\tilde O(p)$ and only returned a full-rank subring~\cite[Theorem~75]{Kohel96}. This is the object of
    Section~\ref{subsec:solving-EndRing}.
\end{itemize}

\bigskip

Our main technical tool is an equidistribution result for isogeny walks in the
graph of supersingular elliptic curves equipped with an endomorphism modulo~$N$.
In fact, we prove a more general equidistribution result generalising the
classical one (see~\cite{Mestre86,Pizer90} and
Proposition~\ref{prop:rand-walk-standard}), which we think is of
independent interest.
We state this result informally here, refering the reader to the body of the
paper for a formal statement.

\begin{definition}
  Equipping the set of supersingular elliptic curves
  with extra data consists in defining for each such curve~$E$ a finite set~$\cF(E)$,
  and for every isogeny~$\varphi\colon E\to E'$ a
  map~$\cF(\varphi)\colon\cF(E)\to\cF(E')$, compatible under composition of
  isogenies
  (see Definitions~\ref{def:extradata} and~\ref{def:catSSSigma}).
  We obtain the isogeny graph~$\cG_\cF$ of pairs~$(E,x)$ where~$x\in\cF(E)$ (see
  Definition~\ref{def:graphF}).

  Let~$N\ge 1$ be an integer.
  The extra data \emph{satisfies the $\modN$-congruence property} if for
  every curve~$E$, pairs of endomorphisms of~$E$ that are congruent modulo~$N$ act
  identically on~$\cF(E)$ (see Definition~\ref{defi:congrpropE}).
\end{definition}

Our equidistribution result, stated informally, reads as follows.
\begin{theorem}\label{thm:equidistrinformal}
  Let~$N\ge 1$ be an integer.
  Random walks in the isogeny graph of supersingular elliptic curves equipped
  with extra data satisfying the $\modN$-congruence property equidistribute
  optimally.
\end{theorem}
We refer to Theorem~\ref{thm:equidistrE} for a formal statement.
The optimality refers to the fact that the graphs can be disconnected or
multipartite, resulting in the adjacency matrix having several forced
eigenvalues (see Proposition~\ref{prop:companionequidistr} and
Remark~\ref{rem:equidistr}), but all the remaining eigenvalues are as small as possible.
A similar general result was recently proved by Codogni and
Lido~\cite{otherequidistr}, so we point out some similarities and differences.
In~\cite{otherequidistr}, the extra
data needs to be expressed in terms of $N$-torsion points (a \emph{level structure}), whereas 
we allow for extra data of arbitrary nature,
only requiring it to satisfy a simple property (the ${\modN}$-congruence property). We hope
that this makes our theorem flexible, and easy to use in a variety of situations.
In particular, the extra data used in our main application trivially fits within our framework; in contrast, this data is not a level structure, so does not directly fit the framework of~\cite{otherequidistr}.
Moreover, we allow $p$ to divide~$N$, contrary to the results
in~\cite{otherequidistr}.
Both proofs use Deligne's bounds, but the proof in~\cite{otherequidistr} is
purely algebro-geometric, whereas ours proceeds via the Deuring correspondence
and the Jacquet--Langlands correspondence; as a result, the two proofs could
have different interesting generalisations. 

\subsection{Technical overview}\label{subsec:techoverview} 

The ideas behind our reduction are as follows. Assume we have an
oracle~$\Oracle$ for~$\OneEnd$ and we want to compute~$\End(E)$ for a given~$E$.

The ring $\End(E)$ is a lattice of dimension $4$ and volume $p/4$. Computing $\End(E)$ consists in finding a basis: four endomorphisms that generate all the others. 
Given a collection of endomorphisms, one can compute the volume of the lattice they generate, and easily check whether they generate $\End(E)$.

\subsubsection*{A first flawed attempt.}
We thus need a way to generate several endomorphisms of $E$. Naively, one could repeatedly call $\Oracle(E)$, hoping to eventually obtain a generating set. This can fail, for instance if the oracle is deterministic and~$\Oracle(E)$ always returns the same endomorphism.


To circumvent this issue, it was
proposed in~\cite{EHLMP18} to randomise the curve.
More precisely, one constructs a richer, randomised oracle $\Rich^\Oracle$ from~$\Oracle$ as follows.
On input $E$, walk randomly on the $2$-isogeny graph, resulting in an isogeny
$\varphi \colon E \to E'$. This graph has rapid mixing properties, so $E'$ is close to uniformly distributed among supersingular curves. Now, call the oracle $\mathscr O$ on $E'$, to get an endomorphism $\beta \in \End(E')$. The composition $\alpha = \hat\varphi \circ \beta \circ \varphi$ is an endomorphism of $E$, the output of $\Rich^\Oracle$.

With this randomisation, there is hope that calling $\Rich^\Oracle$ repeatedly on~$E$ could yield several independent endomorphisms that would eventually generate $\End(E)$.
This method is essentially~\cite[Algorithm~8]{EHLMP18}.
In that article, it is heuristically assumed that endomorphisms produced by $\Rich^\Oracle$ are very nicely distributed, and they deduce that a generating set for $\End(E)$ is rapidly obtained.
This heuristic has a critical flaw: one can construct oracles that contradict it. Consider an integer $M > 1$, and suppose that for any input $E$, the oracle $\Oracle$ returns an endomorphism from the strict subring $\Z + M\End(E)$. Then, the above algorithm would fail, because the randomisation $\Rich^\Oracle$ would still be stuck within the subring $\Z + M\End(E)$.
Worse, juggling with several related integers $M$, we will see that there are oracles for which this algorithm only stabilises after an exponential time. 

\subsubsection*{Identifying and resolving obstructions.}
The core of our method rests on the idea that this issue is, in essence, the only possible obstruction. The key is \emph{invariance by conjugation}. 
If $\varphi,\varphi' : E \to E'$ are two random walks of the same length, and $\beta$ is an endomorphism of $\End(E')$, the elements $\alpha = \hat\varphi \circ \beta \circ \varphi$ and $\alpha' = \hat\varphi' \circ \beta \circ \varphi'$ are equally likely outputs of $\Rich^\Oracle$.
These two elements are conjugates of each other in $\End(E)/N\End(E)$ for any odd integer $N$, as
\[\alpha = \frac{\hat\varphi \circ \hat\varphi'}{[\deg(\varphi')]} \circ \alpha' \circ \frac{\varphi' \circ \varphi}{[\deg(\varphi')]} \bmod N.
\]
From there, one can prove that the output of $\Rich^\Oracle$ follows a distribution that is invariant by conjugation: each output is as likely as any of its conjugates, modulo odd integers $N$ (up to some bound).
Intuitively, for the outputs of $\Rich^\Oracle$ to be ``stuck'' in a subring
(such as $\Z + M\End(E)$ above), that subring must itself be stable by
conjugation (modulo odd integers $N$). There comes the next key: every subring of
$\End(E)$ (of finite index not divisible by~$p$) stable by conjugation
modulo all integers is of the form $\Z + M\End(E)$.
From a basis of $\Z + M\End(E)$, it is easy to recover a basis of $\End(E)$ essentially by dividing by $M$ (using a method due to
Robert~\cite{RobertApplications} that stems from the attacks on SIDH).

This intuition does not immediately translate into an algorithm, as an oracle could be ``bad'' without really being stuck in a subring. Imagine an oracle that outputs an element of $\Z + 2^e\End(E)$ (and not in $\Z + 2^{e+1}\End(E)$) with probability $2^{e-n}$ for each $e \in [0,\dots,n-1]$. A sequence of samples $(\alpha_i)_i$ could eventually generate $\End(E)$, but only after an amount of time exponential in $n$. This particular case could be resolved as follows: for each sample $\alpha$, identify the largest $e$ such that $\beta = (2\alpha  - \Tr(\alpha))/2^{e}$ is an endomorphism. A sequence of samples $(\beta_i)_i$ could rapidly generate $\Z + 2\End(E)$, from which one easily recovers $\End(E)$.
This resolution first identifies the prime $2$ as the source of the obstruction, then ``reduces'' each sample ``at $2$''. In general, such obstructive primes would appear as factors of $\disc(\alpha)$. Identifying these primes, and ensuring that each sample is ``reduced'' at each of them, one gets, in principle, a complete algorithm.
However, factoring $\disc(\alpha)$ could be hard. Instead, we implement an optimistic approach: we identify obstructive pseudo-primes using a polynomial time partial-factoring algorithm. The factors may still be composite, but it is fine: the algorithm will either behave as if they were prime, or reveal a new factor.

\bigskip

\subsubsection*{Equidistribution in isogeny graphs.} The technical core of our result is the proof that the distribution of $\Rich^\Oracle$ is indeed invariant by conjugation. It is a consequence of Theorem~\ref{thm:equidistrinformal}, our general equidistribution result.
The proof of Theorem~\ref{thm:equidistrinformal} proceeds as follows. We use a
categorical version of the Deuring correspondence to bring everything to the
quaternion world. We then use a technical result (Theorem~\ref{thm:congrprop})
to show that extra data satisfying the congruence property yield graphs
isomorphic to special ones constructed from quaternionic groups.
Finally, these special graphs are directly related to automorphic forms,
so we can apply the Jacquet--Langlands correspondence and Deligne's bounds on
coefficients of modular forms. The resulting bounds on the eigenvalues of the
adjacency operators give the desired fast mixing result.

\subsection*{Acknowledgements}
The authors would like to thank Damien Robert for discussions about this project.
The authors were supported by the Agence Nationale de la Recherche under grant ANR MELODIA (ANR-20-CE40-0013), ANR CIAO (ANR-19-CE48-0008), and the France 2030 program under grant agreement No. ANR-22-PETQ-0008 PQ-TLS.

\section{Preliminaries}\label{sec:prelim}

\subsection{Notation} 
We write $\Z$, $\Q$, $\R$ and $\C$ for the ring of integers, the fields of rational, real, and complex numbers.
For any prime $\ell$, we write $\Z_\ell$ and $\Q_\ell$ for the ring of $\ell$-adic integers and the field of $\ell$-adic numbers.
For any prime power $q$, we write $\F_q$ for the finite field with $q$ elements. For any field $K$, we write $\overline K$ for its algebraic closure.
For any set $S$, we write $\#S$ for its cardinality.
We write $f = O(g)$ for the classic big O notation, and equivalently $g = \Omega(f)$ for the classic $\Omega$ notation. We also write $f = \Theta(g)$ if we have both $f = O(g)$ and $f = \Omega(g)$. We use the soft O notation $\tilde O(g) = \log(g)^{O(1)} \cdot O(g)$. We also write $\poly(f_1,\dots,f_n) = (f_1 + \dots + f_n)^{O(1)}$.
The logarithm function 
$\log$ is in base $2$. For any ring $R$, we write $R^{\times}$ the multiplicative group of invertible elements, and $M_2(R)$ the ring of $2\times 2$ matrices with coefficients in $R$.

\subsection{Quaternion algebras}
A general reference for this section is~\cite{VoightBook}.
A \emph{quaternion algebra over~$\Q$} is a ring~$B$ having a $\Q$-basis~$1,i,j,k$
satisfying the multiplication rules~$i^2 = a$, $j^2 = b$ and~$k=ij=-ji$, for
some~$a,b\in\Q^\times$.
Let~$w = x+yi+zj+tk\in B$.
The \emph{reduced trace} of~$w$ is~$\trd(w) = 2x$.
The \emph{reduced norm} of~$w$ is~$\nrd(w) =
x^2-ay^2-bz^2-abt^2$. The reduced norm map is multiplicative.

A \emph{lattice} in a $\Q$-vector space~$V$ of finite dimension~$d$ is a
subgroup~$L\subset V$ of rank~$d$ over~$\Z$ and such that~$V = L\Q$.
The \emph{discriminant} of a lattice~$L$ in~$B$ is~$\disc(L) =
\det(\trd(b_ib_j))\neq 0$ where~$(b_i)$ is a~$\Z$-basis of~$L$.
When~$L'\subset L$ is a sublattice, we have~$\disc(L') = [L:L']^2\disc(L)$.

An \emph{order} in~$B$ is a subring~$\order\subset B$ that is also a lattice. A
\emph{maximal order} is an order that is not properly contained in another
order.

The algebra~$B$ is \emph{ramified at~$\infty$} if~$B\otimes\R \not\cong M_2(\R)$.
Let~$\ell$ be a prime number. The algebra~$B$ is \emph{ramified at~$\ell$}
if~$B_\ell := B\otimes\Q_\ell \not\cong M_2(\Q_\ell)$. If~$\ell$ is unramified
and~$\order$ a maximal order, then~$\order_\ell := \order\otimes \Z_\ell\cong
M_2(\Z_\ell)$. The discriminant of a maximal order in~$B$ is the square of the
product of the ramified primes of~$B$.
When~$B$ is ramified at~$\infty$, the quadratic form~$\nrd$ is positive
definite, and for every lattice~$L$ in~$B$, the volume~$\Vol(L)$
satisfies~$\disc(L) = 16\Vol(B)^2$.


\subsection{Elliptic curves}
A general reference for this section is~\cite{Silverman-Arithmetic}.
An \emph{elliptic curve} over a field~$K$ is a genus~$1$ projective curve with a
specified base point~$O$. An elliptic curve has a unique group law (defined algebraically) with neutral
element~$O$.
An algebraic morphism between elliptic curves (preserving the base point)
is automatically a group morphism, and is either a constant map or has a finite
kernel. In the latter case, we say that the morphism is an \emph{isogeny}.
The \emph{degree} of an isogeny~$\varphi$, written $\deg(\varphi)$, is its degree as a rational map.
If~$\varphi$ is separable, we have $\deg(\varphi) = \#\ker(\varphi)$.
An isogeny of degree $d$ is also called a \emph{$d$-isogeny}.
For every integer~$n\neq 0$, the multiplication-by-$n$ map~$[n]\colon E\to E$ is
an isogeny of degree~$n^2$.
Every isogeny~$\varphi\colon E\to E'$ has a \emph{dual
isogeny}~$\hat\varphi\colon E'\to E$ such
that~$\varphi\hat\varphi=[\deg\varphi]$ and~$\hat\varphi\varphi=[\deg\varphi]$.
An \emph{endomorphism} of an elliptic curve~$E$ is a morphism from~$E$ to~$E$.
We denote~$\End(E)$ the ring of endomorphisms of~$E$ defined over~$\overline K$.
The degree map is a positive definite quadratic form on~$\End(E)$.
For~$\alpha\in\End(E)$, the endomorphism~$\alpha+\hat\alpha$ equals the
multiplication map by an integer, the \emph{trace}~$\Tr(\alpha)$ of~$\alpha$, and we have~$\Tr(\alpha)^2\le
4\deg(\alpha)$; we also define the \emph{discriminant}~$\disc(\alpha) =
\Tr(\alpha)^2-4\deg(\alpha)$, which satisfies~$|\disc(\alpha)|\le 4\deg(\alpha)$
and~$\disc(\alpha+[n])=\disc(\alpha)$ for all~$n\in\Z$.
If the characteristic of~$K$ is not~$2$ or~$3$, we have $\Aut(E) = \{\pm 1\}$
for all~$E$, except two isomorphism classes
over~$\overline K$ having respectively~$\#\Aut(E) = 6$ and~$\#\Aut(E) = 4$.

Assume that~$K$ has positive characteristic~$p$ and let~$E$ be an elliptic curve
over~$K$. We say that~$E$ is \emph{supersingular} if~$\End(E)$ is an order in a
quaternion algebra. In this case, $B = \End(E)\otimes \Q$ is a quaternion
algebra over~$\Q$ with ramification set~$\{p,\infty\}$, the ring~$\End(E)$ is a
maximal order in~$B$, and~$E$ is defined over~$\F_{p^2}$.
When we see a nonzero
endomorphism~$\alpha\in\End(E)$ as a quaternion~$a\in B$, we have~$\deg(\alpha)
= \nrd(a)$ and~$\Tr(\alpha) = \trd(a)$.

\subsection{Computing with isogenies}

Let us formalise how one can computationally encode isogenies. All we need is a notion of \emph{efficient representation}: some data efficiently represents an isogeny if it allows to evaluate it efficiently on arbitrary inputs.

\begin{definition}[Efficient representation]\label{def:efficient-representation}
Let $\mathscr A$ be an algorithm, and let~$\varphi : E \to E'$ be an isogeny over a finite field $\F_q$.
An \emph{efficient representation} of $\varphi$ (with respect to $\mathscr A$)
is some data $D_\varphi \in \{0,1\}^*$ such that
\begin{itemize}
\item $D_\varphi \in \{0,1\}^*$ has size polynomial in $\log(\deg(\varphi))$ and $\log q$, and 
\item on input $D_\varphi$ and $P \in E(\F_{q^k})$, the algorithm $\mathscr A$ returns $\varphi(P)$, and runs in polynomial time in $\log(\deg(\varphi))$, $\log q$, and $k$.
\end{itemize}
\end{definition}

\begin{remark}
When we say that an isogeny is in efficient representation, the algorithm $\mathscr A$ is often left implicit. There are only a handful of known algorithms to evaluate isogenies, so one can think of $\mathscr A$ as an algorithm that implements each of these, and $D_\varphi$ would start with an indicator of which algorithm to use.
\end{remark}

We will use the following proposition.

\begin{proposition}\label{prop:division}
There is an algorithm $\Divide$ which takes as input  
\begin{itemize}
\item a supersingular elliptic curve $E/\F_{p^2}$, 
\item an endomorphism $\alpha$ of $E$ in efficient representation, and 
\item an integer $N$,
\end{itemize}
and returns an efficient representation of $\alpha/N$ if $\alpha \in N \End(E)$, and $\bot$ otherwise, and runs in time polynomial in the length of the input.
\end{proposition}

\begin{proof}
This is the division algorithm introduced by
Robert~\cite{RobertApplications} that was derived from the attacks on
SIDH~\cite{CastryckDecruAttack,MMPPW23,RobertAttack}.
Note that in~\cite{RobertApplications}, the algorithm is only presented for particular endomorphisms (translates of the Frobenius), but it works, mostly unchanged, in all generality. The general statement and detailed proof can be found in~\cite{HW23}.
\qed\end{proof}

\subsection{Computational problems}

The endomorphism ring problem is the following.
\begin{prob}[$\EndRing$]
Given a prime $p$ and a supersingular elliptic curve~$E$ over $\F_{p^2}$, find four endomorphisms in efficient representation that form a basis of $\End(E)$ as a lattice.
\end{prob}




As the endomorphism ring problem asks to find, in a sense, all the endomorphisms, it is  natural to study the problem of finding even a single one.
Scalar multiplications~$[m]$ for~$m\in\Z$ are trivial to find, so we exclude them.

\begin{prob}[$\OneEnd$]
Given a prime $p$ and a supersingular elliptic curve $E$ over $\F_{p^2}$, find an endomorphism in $\End(E)\setminus \Z$ in efficient representation.
\end{prob}

There exists arbitrarily large endomorphisms, so it is convenient to introduce a bounded version of this problem.
Given a function $\lambda \colon \Z_{>0} \to \Z_{>0}$, the $\OneEnd_\lambda$ problem denotes the $\OneEnd$ problem where the solution $\alpha$ is required to satisfy $\log (\deg \alpha) \leq \lambda(\log p)$ (in other words, the length of the output is bounded by a function of the length of the input).

The $\ell$-isogeny path problem is a standard problem in isogeny-based cryptography. Fix a prime $\ell$. An $\ell$-isogeny path is a sequence of isogenies of degree $\ell$ such that the target of each isogeny is the source of the next. 
\begin{prob}[$\IsogPath$]
Given a prime $p$ and two supersingular elliptic curves $E$ and $E'$ over $\F_{p^2}$, find an $\ell$-isogeny path from $E$ to $E'$.
\end{prob}

\subsection{Probabilities} 
Given a random variable $X$ with values in a discrete set $\mathscr X$, we say
it has distribution $f$ if $f(x) = \Pr[X = x]$ for every $x \in \mathscr X$. We also write $f(A) = \sum_{x\in A}f(x)$ for any $A \subseteq \mathscr X$.
For two distributions $f_1$ and $f_2$ over the same set~$\mathscr X$, their \emph{statistical distance} (or \emph{total variation distance}) is
\[
\frac{1}{2}\|f_1-f_2\|_1 = \frac{1}{2}\sum_{x \in \mathscr X}|f_1(x) - f_2(x)| = \sup_{A\subseteq \mathscr X}|f_1(A)-f_2(A)|.
\]

Random walks play a key role in isogeny-based cryptography. 
Fix a field $\F_{p^2}$ and a prime number $\ell \neq p$. The supersingular $\ell$-isogeny graph has vertices the (finitely many) isomorphism classes of supersingular elliptic curves over $\F_{p^2}$, and edges are the $\ell$-isogenies between them (up to isomorphism of the target).
At the heart of the Charles--Goren--Lauter hash function~\cite{CGL09}, one of the first isogeny-based constructions, lies the fact that random walks in supersingular $\ell$-isogeny graphs have rapid-mixing properties: they are Ramanujan graphs.
This is the following well-known proposition. It is a particular case of our more general Theorem~\ref{thm:equidistrE}.


\begin{proposition}\label{prop:rand-walk-standard}
  Let $E$ be a supersingular elliptic curve over $\F_{p^2}$, and $\ell \neq p$ a prime
  number. Let $\varepsilon > 0$.
  There is a bound $n = O(\log_\ell(p) - \log_\ell(\varepsilon))$ such that the endpoint of
  a uniform random walk  of length at least~$n$ from $E$ in the $\ell$-isogeny
  graph is at statistical distance at most $\varepsilon$ from the stationary
  distribution~$f$, which satisfies $f(E) = \frac{24}{(p-1)\#\Aut(E)}$. 
\end{proposition}

\begin{proof}
This is a standard consequence of Pizer's proof that the supersingular $\ell$-isogeny graph is Ramanujan~\cite{Pizer90}. 
Details can be found, for instance, in~\cite[Theorem~11]{SECUER} for the length of the walk, and in~\cite[Theorem~7, Item~2]{SECUER} for the description of the stationary distribution.
\qed\end{proof}

 The stationary distribution is at statistical distance $O(1/p)$ of the uniform distribution. Note that rejection sampling allows to efficiently transform a sampler for the stationary distribution into a sampler for the uniform distribution.

\subsection{Categories}
A general reference for this section is~\cite{MacLane}.
A \emph{category}~$\cC$ consists of objects, for every objects~$x,y\in\cC$,
a set of morphisms~$\Hom_\cC(x,y)$, sometimes denoted~$f\colon x\to y$, an
associative composition law for morphisms with compatible source and target, and
an identity morphism~$\id_x\in\Hom_\cC(x,x)$ for every object~$x\in\cC$.
An \emph{isomorphism} is a morphism that admits a two-sided inverse. For~$x,y\in\cC$, we
define the set~$\End_\cC(x) = \Hom_\cC(x,x)$ of \emph{endomorphisms} of~$x$, the
set~$\Isom_\cC(x,y)$ of isomorphisms from~$x$ to~$y$, the group~$\Aut_\cC(x) = \Isom_\cC(x,x)$ of
\emph{automorphisms} of~$x$.
Let~$\cC,\cD$ be categories. A \emph{functor}~$\cF\colon \cC \to \cD$ is an
association of an object~$\cF(x)\in\cD$ for every object~$x\in\cC$, and
of a morphism~$\cF(f)\colon \cF(x)\to \cF(y)$ for every morphism~$f\colon x \to
y$ in~$\cC$, that respects composition\footnote{All our functors are covariant.}
and identities. Functors can be composed in the obvious way.

Let~$\Sets$ be the category of sets.
The following is a standard construction.
\begin{definition}\label{def:extradata}
  Let~$\cC$ be a category and~$\cF\colon \cC \to \Sets$ be a
  functor.
  The \emph{category of elements~$\El(\cF)$} is the category with
  \begin{itemize}
    \item objects: pairs~$(c,x)$ where~$c\in\cC$ and~$x\in\cF(c)$;
    \item morphisms~$(c,x) \to (c',x')$:
      morphisms~$f\in\Hom_{\cC}(c,c')$ s.t.~$\cF(f)(x) = x'$.
  \end{itemize}
  This category is equipped with the natural forgetful functor~$\El(\cF) \to \cC$.
\end{definition}

\begin{remark}
  One could also use the contravariant version of this
  definition. All our results would hold in this setting, as one can
  compose~$\cF$ with the isogeny duality to reverse the direction of all morphisms.
\end{remark}

Let~$\cF,\cF'\colon \cC \to \cD$ be functors. A \emph{morphism of
functors}~$\psi\colon \cF \to \cF'$ is a collection~$\psi = (\psi_x)_{x\in\cC}$
of morphisms~$\psi_x \colon \cF(x) \to \cF'(x)$ in~$\cD$ such that for every morphism~$f\colon
x\to y$ in~$\cC$, we have~$\cF'(f)\circ \psi_x = \psi_y \circ \cF(f)$.
A morphism of functors is an isomorphism if and only if every~$\psi_x$ is an
isomorphism in~$\cD$.
A functor~$\cF\colon\cC\to\cD$ is an \emph{equivalence of categories} if there exists a
functor~$\cG\colon\cD\to\cC$ and isomorphisms of functors~$\cG\circ\cF\cong
\id_{\cF}$ and~$\cF\circ\cG\cong \id_{\cD}$.
The functor~$\cF$ is \emph{full} if all the corresponding maps~$\Hom_\cC(x,y) \to
\Hom_{\cD}(\cF(x),\cF(y))$ are surjective, and \emph{faithful} if they are all
injective. The functor~$\cF$ is \emph{essentially surjective} if every object
in~$\cD$ is isomorphic to the image of some object in~$\cC$ under~$\cF$.
A functor is an equivalence if and only if it is full, faithful and
essentially surjective.
If~$\cF,\cF'\colon \cC \to \Sets$ are functors, every isomorphism of
functors~$\cF\cong\cF'$ induces an equivalence of categories~$\El(\cF) \cong
\El(\cF')$.

The categorical formulation of the Deuring correspondence~\cite{Deuring41} provides the most
versatile way of transfering problems from supersingular elliptic curves to
quaternions.
Let~$\catSS(p)$ denote the category with

\begin{itemize}
  \item objects: supersingular elliptic curves over~$\overline{\F}_{p}$;
  \item morphisms: algebraic group morphisms.
\end{itemize}

We fix a base curve~$E_0\in \catSS(p)$. Let~$\order = \End(E_0)$ and~$B =
\order \otimes \Q$.

Let~$\Mod(\order)$ denote the category with
\begin{itemize}
  \item objects: invertible right $\order$-modules;
  \item morphisms: right $\order$-module homomorphisms.
\end{itemize}


Then we have the classical Deuring correspondence
(\cite[Theorem~42.3.2]{VoightBook}, 
except we are using the covariant version;
see also~\cite[Theorem~45]{Kohel96}).

\begin{theorem}\label{thm:deuring}
  The association~$E \mapsto \Hom(E_0,E)$, $(\varphi \colon E \to E') \mapsto
  (\psi \mapsto \varphi\psi)$ defines a equivalence of categories
  \[
    \catSS(p) \longrightarrow \Mod(\order).
  \]
\end{theorem}

\subsection{Quaternionic automorphic forms}

The following preliminaries concern the proof of our result on the equidistribution of isogeny random walks. The reader willing to admit Theorem~\ref{thm:equidistrE} without proof does not need background on automorphic forms.
A general reference for this section is~\cite{JL}; a more gentle one, for
Borel-type level, is~\cite[Section~3]{DembeleVoight}.

A \emph{Hilbert space}~$V$ is a complex vector space equipped with a Hermitian
inner product~$\langle\cdot,\cdot\rangle$, and complete for the induced
norm~$\|\cdot\|$, which is automatic if~$V$ has finite dimension.
Let~$V$ be a finite-dimensional Hilbert space.
The \emph{adjoint} of a linear operator~$T\colon V \to V$ is the unique
operator~$T^*$ satisfying~$\langle Tv,w\rangle = \langle v,T^*w\rangle$ for
all~$v,w\in V$. A \emph{normal operator} is an operator that commutes with its
adjoint. The \emph{operator norm} of an operator~$T$ is
\[
  \max_{v\neq 0}\frac{\|Tv\|}{\|v\|}.
\]
Every normal operator stabilises the orthogonal complement of every stable
subspace, is diagonalisable in an orthogonal basis, and has operator norm equal
to the maximum absolute value of its eigenvalues.

Let~$\hatZ = \prod_{\ell} \Z_{\ell}$ be the profinite completion of~$\Z$. For every abelian
group~$A$, we write~$\widehat{A} = A \otimes_{\Z} \hatZ$.
In particular~$\widehat{\Q} = \prod_{\ell}' \Q_{\ell}$ is the ring of finite ad\`eles
of~$\Q$.

Let~$N\ge 1$ be an integer, and let~$U(N) = (1+N\hatZ)\cap \hatZ^\times$.
Let~$H\subset (\Z/N\Z)^\times$ be a subgroup, and
let~$U\subset\hatZ^\times$ be the preimage of~$H$ under the quotient
map~$\hatZ^\times \to \hatZ^\times/U(N) = (\Z/N\Z)^\times$.
Then we have an isomorphism
\[
  \Q_{>0}^\times \lquo \hatQ^\times / U \cong \hatZ^\times/U \cong (\Z/N\Z)^\times / H.
\]
Indeed, since every ideal of~$\Z$ has a positive generator we have~$\hatQ^\times
= \Q_{>0}^\times \hatZ^\times$.

Let~$B$ be a quaternion algebra over~$\Q$.
The group~$\hatB^\times$ admits a measure~$\mu$ that is bi-invariant under group
translations, finite on compact subsets and nontrivial on open subsets, called
its~\emph{Haar measure}.

Let~$N\ge 1$ be an integer, let~$U(N) = (1+N\hatorder) \cap
\hatorder^\times$, which is a compact open subgroup of~$\hatB^\times$,
and let~$U$ be a subgroup satisfying~$U(N)\subseteq U
\subseteq \hatorder^\times$.
The set~$B^\times \lquo \hatB^\times / U$ is finite. The \emph{space of
automorphic forms of level~$U$} is the Hilbert space~$L^2(B^\times\lquo
\hatB^\times / U)$, equipped with the inner product induced by the projection of
the Haar measure:
\[
  \langle F, G\rangle = \int_{b\in B^\times \lquo \hatB^\times / U}
  F(b)\overline{G(b)}\mu(bU).
\]

Fix, for every~$\ell$ that is unramified in~$B$, an
isomorphism~$\order_\ell\cong M_2(\Z_\ell)$, and
let~$\delta_\ell\in \hatB^\times$ have component~$1$ at
every~$\ell'\neq \ell$ and that corresponds to~$\begin{pmatrix}\ell & 0 \\ 0 &
1\end{pmatrix}$ at~$\ell$ via the chosen isomorphism.
The Hecke operator
\[
  T_\ell \colon L^2(B^\times\lquo \hatB^\times / U)
  \longrightarrow L^2(B^\times\lquo \hatB^\times / U)
\]
is defined by
\[
  T_\ell F(B^\times xU) = \sum_{uU\in U\delta_\ell U/U} F(B^\times xuU),
\]
and its adjoint admits the expression
\[
  T_\ell^* F(B^\times xU) = \sum_{uU\in U\delta_\ell^{-1} U/U} F(B^\times xuU).
\]
For~$\ell$ that do not divide~$N$, the Hecke operators~$T_\ell$ are normal
operators that pairwise commute.

\section{Equidistribution of elliptic curves with extra data}\label{sec:equidistr}

The goal of this section is to prove Theorem~\ref{thm:equidistrE}. We state our
results in Subsection~\ref{subsec:statements}. In
Subsection~\ref{subsec:deuring}, we set up a suitable version of the Deuring
correspondence. The goal of Subsection~\ref{subsec:congrtype} is to prove a
technical result classifying extra data satisfying a simple property.
In Subsection~\ref{subsec:equidistr} we
apply automorphic methods to prove the equidistribution theorem.

\subsection{Statement of the equidistribution theorem} \label{subsec:statements}

In order to avoid bad primes, we will need to restrict the possible degrees of
isogenies under consideration.
Let~$\Sigma$ be a set of primes, and let~$N\ge 1$ be an integer not divisible by
any prime in~$\Sigma$.

\begin{definition}\label{def:catSSSigma}
  Let~$\catSS_\Sigma(p)$ denote the category with
  \begin{itemize}
    \item objects: supersingular elliptic curves over~$\overline{\F}_{p}$;
    \item morphisms $\Hom_\Sigma(E,E')$: isogenies with degree a product of the primes in~$\Sigma$.
  \end{itemize}
\end{definition}

Our results are expressed in terms of categories of elements of various functors, as in
Definition~\ref{def:extradata}.
For us, this is going to play the role of ``equipping with extra structure":
when~$\cF\colon \cC \to \Sets$ is a functor,
$\El(\cF)$ is the category of ``objects~$c\in\cC$ with extra structure taken
from~$\cF(c)$".
A related definition can be found in~\cite[Definition~2.1]{LeiMueller},
formulated at the level of graphs.
For us, an advantage of the category-theoretic formulation is that we
can forget about $\catSS(p)$ and work in
a quaternionic category, thanks to the Deuring correspondence.

\begin{example}\label{ex:cycN}
  Assume~$p\nmid N$. Let~$\Sigma$ be the set of primes not dividing~$N$.
  Define the functor~$\Cyc_N\colon \catSS_\Sigma(p) \to \Sets$ by:
  \begin{itemize}
    \item $\Cyc_N(E)$ is the set of cyclic subgroups of order~$N$ of~$E$;
    \item for every isogeny~$\varphi\in\Hom_\Sigma(E, E')$, the map~$\Cyc_N(\varphi)$
      is~$C \mapsto \varphi(C)$.
  \end{itemize}
  Then~$\El(\Cyc_N)$ is the category of supersingular elliptic curves equipped
  with a cyclic subgroup of order~$N$.
\end{example}

\begin{example}\label{ex:endmodN}
  Let~$\Sigma$ be the set of primes not dividing~$N$.
  Let~$\End/N$ denote the functor~$\catSS_\Sigma(p) \to \Sets$ defined by
  \begin{itemize}
    \item $(\End/N)(E) = \End(E)/N\End(E)$;
    \item for~$\varphi\colon E \to E'$, the map~$(\End/N)(\varphi)$ is~$\alpha \mapsto
      \varphi\alpha\hat\varphi$.
  \end{itemize}
  Then~$\El(\End/N)$ is is the category of supersingular elliptic curves equipped with an
  endomorphism modulo~$N$, which will play an important role in
  Section~\ref{sec:enriching}.
\end{example}

We now introduce the graphs of interest (more generally see Definition~\ref{def:graphC}).

\begin{definition}\label{def:graphF}
  Let~$\cF\colon \catSS_\Sigma(p) \to \Sets$ be a functor with~$\cF(E)$ finite
  for all~$E$.
  We define the graph~$\cG_\cF$ with:
  \begin{itemize}
    \item vertices: isomorphism classes of objects in~$\El(\cF)$;
    \item edges: let~$(E,x)\in \El(\cF)$; edges from~$(E,x)$
      are isogenies~$\varphi\in\Hom_\Sigma(E,E')$ modulo
      automorphisms of~$(E',\cF(\varphi)(x))$.
  \end{itemize}
  Let~$L^2(\cG_\cF)$ be the space of complex functions on
  vertices of~$\cG_\cF$, and define
  \[
    \langle F,G\rangle =
    \sum_{(E,x)\in\cG_\cF}\frac{F(E,x)\overline{G(E,x)}}{\#\Aut(E,x)}
    \text{ for }F,G\in L^2(\cG_\cF).
  \]
  For every prime~$\ell$, we define the adjacency operator~$A_\ell$
  on~$L^2(\cG_\cF)$ by
  \[
    A_\ell F(E,x) = \sum_{(E,x)\to (E',x')} F(E',x'),
  \]
  where the sum runs over edges of degree~$\ell$ leaving~$(E,x)$.
\end{definition}

\begin{remark}
  The graphs~$\cG_\cF$ have finitely many vertices, but infinitely
  many~edges.
\end{remark}

\begin{example}
  Assume~$p\nmid N$, and let~$\ell$ a prime not dividing~$Np$.
  The graph obtained from~$\cG_{\Cyc_N}$ by keeping only the edges of
  degree~$\ell$ is the $\ell$-isogeny graph of supersingular elliptic curves
  with Borel structure studied in~\cite{Arpin} and~\cite{SECUER}.
  When~$N=1$ this is the classical supersingular $\ell$-isogeny graph.
\end{example}

We are now in position to state our equidistribution theorem.

\begin{definition}\label{defi:congrpropE}
  Let~$\cF\colon \catSS_\Sigma(p) \to \Sets$ be a functor and~$N\ge 1$ an
  integer. We say that $\cF$
  \emph{satisfies the $\modN$-congruence property} if for
  every~$E\in\catSS(p)$ and every~$\varphi,\psi\in\End_\Sigma(E)$ such that~$\varphi-\psi\in N\End(E)$, we
  have~$\cF(\varphi)=\cF(\psi)$.
\end{definition}

\begin{example}
  Assume that~$p$ does not divide~$N$.
  The functor~$\Cyc_N$ from Example~\ref{ex:cycN} satisfies the $\modN$-congruence
  property: indeed, endomorphisms divisible by~$N$ act as~$0$ on
  $N$-torsion points.
\end{example}

\begin{example}
  The functor~$\End/N$ from Example~\ref{ex:endmodN} has the
  $\modN$-congruence property: if~$\varphi,\psi\in\End_{\Sigma}(E)$
  and~$\alpha,\beta\in\End(E)$
  satisfy~$\psi = \varphi + N\beta$, then~$\psi\alpha\hat\psi =
  (\varphi+N\beta)\alpha(\hat\varphi+N\hat\beta) \in \varphi\alpha\hat\varphi + N\End(E)$,
  so that~$(\End/N)(\varphi) = (\End/N)(\psi)$.
\end{example}


\begin{theorem}\label{thm:equidistrE}
  Let~$p$ be a prime and~$N\ge 1$ an integer.
  Let~$\Sigma$ be a set of primes that do not divide~$N$, such that~$\Sigma$
  generates~$(\Z/N\Z)^\times$.
  Let~$\cF\colon \catSS_\Sigma(p) \to \Sets$ be a functor satisfying the
  $\modN$-congruence property and such that all sets~$\cF(E)$ are finite.

  Then, for every~$\ell\in\Sigma$ different from~$p$, the adjacency operator~$A_\ell$ is a normal operator on $L^2(\cG_\cF)$ which stabilises the following subspaces:
  \begin{itemize}
    \item $L^2_{\deg}(\cG_\cF)$, the subspace of functions that are constant on
      every connected component of the graph~$\cG_\cF^1$ obtained from~$\cG_\cF$ by keeping only the
      edges of degree~$1\bmod N$.
      The operator norm of $A_\ell$ on $L^2_{\deg}(\cG_\cF)$ is $\ell+1$.
    \item $L^2_0(\cG_\cF)$, the orthogonal complement of~$L^2_{\deg}(\cG_\cF)$. The operator norm of $A_\ell$ on $L^2_0(\cG_\cF)$ is at most~$2\sqrt{\ell}$.
  \end{itemize}
  Moreover, the~$A_\ell$ for~$\ell\in\Sigma$ pairwise commute.
\end{theorem}

In other words, the normalised operator $A'_\ell = \frac{1}{\ell+1}A_\ell$ makes
functions rapidly converge to the subspace $L^2_{\deg}(\cG_\cF)$.
This operator $A'_\ell$ preserves the subset of probability distributions, and closely relates to
the effect of a random walk of $\ell$-isogenies (see Appendix~\ref{appen:dist-func}).
In simple cases (such as $N = 1$), the space~$L^2_{\deg}(\cG_\cF)$ has dimension
$1$, is generated by the constant function~$1$
and the theorem says that random walks in $\ell$-isogeny graphs rapidly converge
to the unique stationary distribution~$f$ with~$f(E,x)$ proportional
to~$\frac{1}{\#\Aut(E,x)}$.
One thus sees that the classical rapid-mixing property for isogeny graphs (Proposition~\ref{prop:rand-walk-standard}) is a particular case of Theorem~\ref{thm:equidistrE}.
More details and other illustrations of Theorem~\ref{thm:equidistrE} are available in Appendix~\ref{appen:illustrations}.

In general $L^2_{\deg}(\cG_\cF)$ could have higher dimension. This reflects the fact that the graph may be disconnected or multipartite, two obstructions for random walks to converge to a unique limit.
To ease the application of Theorem~\ref{thm:equidistrE} in such cases, we provide the
following companion proposition that gives extra information on the graph~$\cG_\cF$
and an explicit description of the space~$L^2_{\deg}(\cG_\cF)$.

\begin{proposition}\label{prop:companionequidistr}
  With the same hypotheses and notations as in Theorem~\ref{thm:equidistrE}:
  \begin{enumerate}[(1)]
    \item\label{item:bij} for every isogeny~$\varphi$ in~$\catSS_\Sigma(p)$, the
      map~$\cF(\varphi)$ is a bijection;
    \item\label{item:connect} for every~$E,E'\in\catSS(p)$, there exists~$\varphi\in\Hom_\Sigma(E,E')$ of degree~$1\bmod N$;
    \item\label{item:action} the morphism~$\End_\Sigma(E_0) \to (\End(E_0)/N\End(E_0))^\times$ is
      surjective, inducing an action of the group~$G=(\End(E_0)/N\End(E_0))^\times$
      on~$\cF(E_0)$.
  \end{enumerate}
  Let~$x_1,\dots,x_n$ denote representatives of the orbits of the action of~$G$
  on~$\cF(E_0)$ and for each~$i$, let~$H_i$ denote the stabiliser of~$x_i$
  in~$G$.
  Let~$\cG_{\deg}$ denote the graph
  with edges labelled by elements of~$(\Z/N\Z)^\times$ and
  with
  \begin{itemize}
    \item vertex set~$\bigsqcup_i (\Z/N\Z)^\times/\deg(H_i)$;
    \item for every~$i$, every~$a\in (\Z/N\Z)^\times/\deg(H_i)$ and
      every~$d\in(\Z/N\Z)^\times$, an edge~$a\to b$ labelled by~$d$, where~$b =
      ad\in (\Z/N\Z)^\times/\deg(H_i)$.
  \end{itemize}
  Then:
  \begin{enumerate}[(1)] \setcounter{enumi}{3}
    \item\label{item:Degmap} there exists a unique morphism of graphs
      \[
        \Deg\colon \cG_\cF \longrightarrow \cG_{\deg}
      \]
      such that for all~$i$ we have~$\Deg(E_0,x_i) =
      1\in(\Z/N\Z)^\times/\deg(H_i)$
      and for every edge~$\varphi$ of~$\cG_\cF$, the
      edge~$\Deg(\varphi)$ is labelled by~$\deg(\varphi) \bmod N$;
    \item\label{item:surj} the map~$\Deg$ is surjective; and
    \item\label{item:L2deg} $L^2_{\deg}(\cG_\cF)$ is the space of
      functions that factor through~$\Deg$.
  \end{enumerate}
\end{proposition}

\begin{remark}\label{rem:equidistr}\hfill
  \begin{itemize}
    \item Properties~(\ref{item:bij}) and~(\ref{item:connect}) allow us to
      transfer what happens at~$E_0$ to any other curve.

    \item Property~(\ref{item:action}) allows us to define the~$H_i$.
      When~$p\nmid N$, this can be
      used to relate~$\cF$ to the setup of~\cite{otherequidistr}, using an
      isomorphism~$G\cong \GL_2(\Z/N\Z)$.
      Note that when~$p\mid N$, the group~$(\order/N\order)^\times$ is not
      isomorphic to~$\GL_2(\Z/N\Z)$.
    \item The graph~$\cG_{\deg}$ is the Cayley graph of the set~$\bigsqcup_i
      (\Z/N\Z)^\times/\deg(H_i)$ equipped with its natural action
      of~$(\Z/N\Z)^\times$.
    \item Property~(\ref{item:Degmap}) amounts to stating the existence
      of a disconnectedness and a multipartition of~$\cG_\cF$.
    \item Property~(\ref{item:surj}) ensures that the space of functions
      on~$\cG_{\deg}$ injects into~$L^2(\cG_\cF)$ via the map~$\Deg$.
    \item Using Properties~(\ref{item:surj}) and~(\ref{item:L2deg}), one easily
      obtains the spectra of the adjacency operators~$A_\ell$
      on~$L^2_{\deg}(\cG_\cF)$: for every complex character~$\chi$
      of~$(\Z/N\Z)^\times$, one obtains the eigenvalue~$\chi(\ell)(\ell+1)$ with
      multiplicity equal to the number of~$i$ such that~$\chi(\deg(H_i)) = 1$.
    \item From Proposition~\ref{prop:companionequidistr} and
      Theorem~\ref{thm:equidistrE}, since~$2\sqrt{\ell} < \ell+1$,
      one can simply deduce connectedness and
      multipartition properties of~$\cG_\cF$, its degree $\ell$ subgraphs, etc.
      For instance, the graph~$\cG_\cF$ has exactly~$n$ connected components: the
      preimages of the~$(\Z/N\Z)^\times/\deg(H_i)$ via the map~$\Deg$.
  \end{itemize}
\end{remark}

\begin{example}
  Assume~$p\nmid N$, let~$\Sigma$ denote the set of all
  primes that do not divide~$pN$, and apply Theorem~\ref{thm:equidistrE}
  and Proposition~\ref{prop:companionequidistr} to~$\cF = \Cyc_N$.
  Then we have an isomorphim~$G\cong\GL_2(\Z/N\Z)$ and a compatible
  bijection~$\cF(E_0)\cong \{\Z/N\Z\text{-lines in } (\Z/N\Z)^2\}$.
  In particular, there is a single orbit ($n=1$) and, choosing~$x_1$
  corresponding to the line generated by~$\binom{1}{0}$, the stabiliser~$H =
  H_1$ corresponds to the subgroup of upper-triangular matrices, so
  that~$\deg(H) = (\Z/N\Z)^\times$. The space~$L^2_{\deg}(\cG_{\Cyc_N})$ is
  therefore one-dimensional, generated by the constant function~$1$.
  Hence Theorem~\ref{thm:equidistrE} recovers~\cite[Theorem 8]{SECUER}.
\end{example}

\subsection{Ad\'elic Deuring correspondence}\label{subsec:deuring}

Since automorphic forms are usually defined using ad\'elic language, we will
reformulate the Deuring correspondence using ad\`eles
(Corollary~\ref{cor:adelicdeuring}).

The following terminology will be convenient.

\begin{definition}
  A \emph{degree map} on a category~$\cC$ is the data, for every morphism~$f$
  of~$\cC$, of an integer~$\deg(f)\in\Z_{\ge 0}$ such
  that~$\deg(fg) = \deg(f)\deg(g)$ for every morphisms~$f,g$ that can be
  composed, and such that~$\deg(\id_x) = 1$ for all~$x\in\cC$.

  A functor~$\cF$ between categories equipped with degree maps is
  \emph{degree-preserving} if~$\deg(\cF(f)) = \deg(f)$ for every morphism~$f$.

  When~$\cC$ is a category with a degree map and~$\cF\colon \cC\to \Sets$ is a
  functor, we equip the category of elements~$\El(\cF)$ with its inherited
  degree map.
\end{definition}

\begin{remark}
  In a category with a degree map, every isomorphism has degree~$1$.
\end{remark}

\begin{example}
  The category~$\catSS(p)$ is equipped with a degree map: the usual degree for
  isogenies, and~$0$ for the zero morphism.
\end{example}

The following is the basic object underlying the ad\'elic Deuring
correspondence.

\begin{definition}
  We define the category~$\Cosets(\hatorder^\times)$ with
  \begin{itemize}
    \item objects: cosets~$[x] := x\hatorder^\times \in  \hatB^\times /
      \hatorder^\times$ for~$x\in \hatB^\times$;
    \item morphisms: $\Hom([x],[y]) = B\cap y\hatorder x^{-1}$, using
      multiplication in~$B$ as composition.
  \end{itemize}
  We equip the category~$\Cosets(\hatorder^\times)$ with a degree map as
  follows:
  for every~$b\in \Hom([x],[y])$, we define the \emph{degree of~$b$} to be the
  positive integer $\deg(b)$ such that~$\deg(b)\hatZ = \nrd(u)\hatZ$ where~$b=
  yux^{-1}$.
\end{definition}

\begin{remark}
  We warn the reader that a single element~$b\in B$ can represent
  different morphisms, depending on the source~$[x]$ and target~$[y]$. Moreover,
  the degree of the morphism is in general not the reduced norm of~$b$.
\end{remark}

We reformulate the Deuring correspondence ad\'elically as follows.

\begin{proposition}\label{prop:adelise}
  The association~$[x] \mapsto B\cap x\hatorder$, $g\in \Hom([x],[y]) \mapsto
  (b \mapsto gb)$ defines a equivalence of categories
  \[
    \Cosets(\hatorder^\times) \longrightarrow \Mod(\order).
  \]

  Its composition with the equivalence of Theorem~\ref{thm:deuring} is a
  degree-preserving equivalence of categories
  \[
    \Cosets(\hatorder^\times) \longrightarrow \catSS(p).
  \]
\end{proposition}
\begin{proof}
  The association described clearly defines a faithful functor.

  We claim that the functor is full. Indeed, let~$f\in \Hom(B\cap x\hatorder,
  B\cap y\hatorder)$. Since every right $B$-module endomorphism of~$B$ is a left
  multiplication by an element of~$B$, there exists~$g\in B$ such
  that~$f(b) = gb$ for all~$b\in B\cap x\hatorder$.
  Moreover, by weak approximation the closure of~$B\cap x\hatorder$ in~$\hatB$
  is~$x\hatorder$, so we must have~$gx\hatorder \subset y\hatorder$ and
  therefore~$g\in y\hatorder x^{-1}$, so that~$g\in \Hom([x],[y])$ as claimed.

  Finally, the functor is essentially surjective since every right invertible
  $\order$-module is isomorphic to a right invertible $\order$-ideal~$I$, and
  such an ideal is locally principal and therefore of the form~$I = B\cap
  x\hatorder$.

  By examining the determinant of a morphism on the modules, we see that the
  equivalence preserves the degree.

\qed\end{proof}

Fix~$\Sigma$ be a set of primes. Let~$\hatorder_{\Sigma}$ denote the
ring obtained from~$\hatorder$ by inverting all primes in~$\Sigma$.
Then~$\hatorder\cap \hatorder_{\Sigma}^\times$ is the set of
elements~$u\in \hatorder$ such that~$\nrd(u)\hatZ$ is generated by a product
of the primes in~$\Sigma$.

\begin{definition}
  Let~$\Cosets_\Sigma(\hatorder^\times)$ be the category with
  \begin{itemize}
    \item objects: cosets~$[x] = x\hatorder^\times \in  \hatB^\times /
      \hatorder^\times$ for~$x\in \hatB^\times$;
    \item morphisms: $\Hom_\Sigma([x],[y]) = B^\times\cap y(\hatorder \cap
      \hatorder_{\Sigma}^\times) x^{-1} = \Hom([x],[y])\cap y \hatorder_{\Sigma}^\times x^{-1}$,
      using multiplication in~$B$ as composition.
  \end{itemize}
\end{definition}
Equivalently, the morphisms in~$\Cosets_\Sigma(\hatorder^\times)$ are the morphisms
in~$\Cosets(\hatorder^\times)$ whose degree is a product of the primes
in~$\Sigma$.
We obtain the following corollary.

\begin{corollary}[Ad\'elic Deuring correspondence]\label{cor:adelicdeuring}
  The second equivalence from Proposition~\ref{prop:adelise} induces a degree-preserving equivalence of categories
\[
  \Cosets_\Sigma(\hatorder^\times) \longrightarrow \catSS_\Sigma(p).
\]
\end{corollary}

\subsection{Extra data of congruence type}\label{subsec:congrtype}
Fix~$N\ge 1$ an integer not divisible by any prime in~$\Sigma$.

In this subsection, we will study categories of elements of various functors, as in
Definition~\ref{def:extradata}.
We will use functors coming from quaternionic constructions, which will allow us
to apply automorphic methods.
The main result of this subsection is Theorem~\ref{thm:congrprop}, which
classifies functors satisfying a simple property in terms of ad\'elic groups.

Let~$U(N) = (1+N\hatorder) \cap
\hatorder^\times$, which is a finite index subgroup of~$\hatorder^\times$,
and similarly~$U_\Sigma(N) = (1+N\hatorder_{\Sigma}) \cap \hatorder_{\Sigma}^\times$.
Let~$U$ be a subgroup of~$\hatorder^\times$ containing~$U(N)$, and
let~$U_\Sigma = U \cdot U_\Sigma(N)$, so that~$U = U_\Sigma \cap
\hatorder^\times$. Note that the natural map
$\hatorder^\times/U \to \hatorder_{\Sigma}^\times / U_\Sigma$ is a
bijection, i.e. we have~$\hatorder_\Sigma^\times = \hatorder^\times U_\Sigma$.

It is helpful to think about these definition in terms of the product
decomposition~$\hatB^\times = \prod_\ell' B_\ell^\times$ as follows: we have
\[
  U = U' \times \prod_{\ell \nmid N} \order_\ell^\times
  \text{ and }
  U_\Sigma = U' \times \hspace{-7pt}
    \prod_{\ell \nmid N, \ell \notin \Sigma} \hspace{-8pt} \order_\ell^\times
    \times \prod_{\ell\in\Sigma}\!{}' B_\ell^\times
\]
where~$U'$ is the image of~$U$ in~$\prod_{\ell\mid N} \order_\ell^\times$.

\begin{definition}\label{def:FU}
  Let~$\cF_U$ be the functor
  \[
    \cF_U \colon \Cosets_\Sigma(\hatorder^\times) \to \Sets
  \]
  defined by
  \begin{itemize}
    \item $\cF_U([x]) = x\hatorder_{\Sigma}^\times / U_\Sigma$,
    \item for~$b\in\Hom_\Sigma([x],[y])$, the map~$\cF_U(b)\colon \cF_U([x]) \to \cF_U([y])$ is left multiplication by~$b$.
  \end{itemize}

  Let~$\Cosets_\Sigma(U)$ be the category with
  \begin{itemize}
    \item objects: cosets~$xU \in  \hatB^\times / U$
      for~$x\in \hatB^\times$;
    \item morphisms: $\Hom_U(xU,yU) = B^\times\cap y(\hatorder \cap U_\Sigma)
      x^{-1} = \Hom([x],[y]) \cap yU_\Sigma x^{-1}$, using
      multiplication in~$B$ as composition.
  \end{itemize}
\end{definition}
In other words, morphism are required to respect~$U$ at the primes
dividing~$N$, and to have degree a product of the primes in~$\Sigma$.

\begin{example}\label{ex:eichler}
  Assume that~$p$ does not divide~$N$.
  Let~$\order_0(N)\subset \order$ be an Eichler
  order of level~$N$, and
  let~$U = \hatorder_0(N)^\times$.
  Then for all~$x\in \hatB^\times$, the set~$\cF_U([x])$ is in bijection
  with~$\P^1(\Z/N\Z)$, or more naturally with the set of Eichler orders
  contained in the maximal order~$x\hatorder x^{-1} \cap B$.
\end{example}

The following proposition is the bridge between functors on the quaternionic
side of the Deuring correspondence and automorphic forms.

\begin{proposition}\label{prop:FUcosetsU}
  The association
  \begin{itemize}
    \item $([x],xgU_\Sigma) \mapsto xg'U$ where~$xgU_\Sigma = xg'U_\Sigma$
      and $g'\in\hatorder^\times$;
    \item $b\in\Hom_{\El(\cF_U)}(([x],xgU_\Sigma),([y],yhU_\Sigma)) \mapsto b
      \in \Hom_U(xg'U,yh'U)$
  \end{itemize}
  defines a degree-preserving equivalence of categories
  \[
    \El(\cF_U) \longrightarrow \Cosets_\Sigma(U).
  \]
\end{proposition}
\begin{proof}
  First, the association is well-defined on objects:
  if two elements~$g',g''\in\hatorder^\times$ satisfy~$xg'U_\Sigma =
  xg''U_\Sigma$, then~$(g')^{-1}g''\in
  U_\Sigma\cap \hatorder^\times = U$ so~$xg'U = xg''U$.
  Next, it is well-defined on morphisms:
  let~$b\in\Hom_{\El(\cF_U)}(([x],xgU_\Sigma),([y],yhU_\Sigma))$,
  and~$xgU_\Sigma=xg'U_\Sigma$ and~$yhU_\Sigma=yh'U_\Sigma$
  with~$g',h'\in\hatorder^\times$;
  then~$bxg'U_\Sigma = yh'U_\Sigma$ so~$b\in yh' U_\Sigma (xg')^{-1}$
  and therefore~$b\in\Hom_{U}(xg'U,yh'U)$.
  Since the association is clearly multiplicative on morphisms, it defines a
  functor.
  Moreover, the functor~$\Cosets_\Sigma(U) \longrightarrow \El(\cF_U)$ defined
  by
  \begin{itemize}
    \item $xU \mapsto ([x],xU_\Sigma)$;
    \item $b\in\Hom_U(xU,yU) \mapsto b$
  \end{itemize}
  is clearly an inverse, so we obtain an equivalence as claimed.
\qed\end{proof}

We will need the following consequence of the strong approximation theorem.

\begin{lemma}\label{lem:strongapprox}
  Assume that~$\Sigma$ contains at least one prime different from~$p$ and
  that it generates~$(\Z/N\Z)^\times$.
  Then for every~$g,x\in\hatB^\times$, we have
  \[
    \Hom_\Sigma([x],[gx]) \cap gx U_\Sigma(N) x^{-1} \neq \emptyset.
  \]
\end{lemma}
\begin{proof}
  Let~$x\in \hatB^\times$, and
  let~$H = x U_\Sigma(N) x^{-1}$.
  Since~$\Sigma$ contains a prime different from~$p$, strong approximation
  holds (\cite[Theorem 28.5.3, see also 28.5.5]{VoightBook}), so the reduced 
  norm induces a bijection
  \[
    B^\times \lquo \hatB^\times / H \longrightarrow \Q_{>0}^\times \lquo
    \hatQ^\times / \nrd(H).
  \]
  On the other hand, the group~$\Q_{>0}^\times \lquo \hatQ^\times / \nrd(H)$ is isomorphic
  to a quotient of~$(\Z/N\Z)^\times/\langle\Sigma\rangle$. Since~$\Sigma$
  generates~$(\Z/N\Z)^\times$ the latter quotient is trivial, so~$\hatB^\times =
  B^\times H$.
  Now let~$g\in\hatB^\times$. Write~$g = b_0xu_0^{-1}x$ with~$b_0\in B^\times$
  and~$u_0\in U_\Sigma(N)$, so that~$b_0 = gxu_0x^{-1}$.
  Let~$\lambda$ be a product of the primes in~$\Sigma$ such that~$\lambda
  u_0\in\hatorder$ and~$\lambda \equiv 1 \bmod N$.
  Let~$b = \lambda b_0$ and~$u = \lambda u_0$, which satisfy~$gxux^{-1} = b\in
  B^\times$ and~$u\in \hatorder\cap U_\Sigma(N)$. Then
  $b\in \Hom_\Sigma([x],[gx]) \cap gx U_\Sigma(N) x^{-1}$, which is therefore
  not empty.
\qed\end{proof}

We are now in position to state and prove the main result of this subsection.

\begin{definition}\label{def:congrprop}
  Let~$\cF\colon \Cosets_\Sigma(\hatorder^\times) \to \Sets$ be a functor. We say
  that~$\cF$ is \emph{of $N$-congruence type} if
  $\cF$ is isomorphic to a disjoint union of functors~$\cF_U$.
  We say that~$\cF$ \emph{satisfies the $\modN$-congruence property}
  if for every~$x\in\hatB^\times$ and for every~$a,b\in \End_\Sigma([x])$ such
  that~$a - b \in N\cdot \End([x])$, we have~$\cF(a) = \cF(b)$.
\end{definition}

\begin{theorem}\label{thm:congrprop}
  Assume that~$\Sigma$ contains at least one prime different from~$p$ and
  that it generates~$(\Z/N\Z)^\times$.
  Let~$\cF\colon \Cosets_\Sigma(\hatorder^\times) \to \Sets$ be a functor.
  Then~$\cF$ is of $N$-congruence type if and only if~$\cF$ satisfies the
  $\modN$-congruence property.
  More precisely, assume that~$\cF$ satisfies the $\modN$-congruence
  property. Then:
  \begin{enumerate}[(1)]
    \item\label{item:bijU} for every morphism~$f$ in~$\Cosets_\Sigma(\hatorder^\times)$, the
      map~$\cF(f)$ is a bijection; and
    \item\label{item:actionU} the morphism~$\End_\Sigma([1]) \to
      (\order/N\order)^\times$ is surjective, inducing an action of the group~$G
      = (\order/N\order)^\times$ on~$\cF([1])$.
  \end{enumerate}
  Choose a $G$-equivariant bijection~$\cF([1]) \cong \bigsqcup_i G/H_i$, and
  for all~$i$, let~$U_i$ be the preimage of~$H_i$ under the
  quotient map~$\hatorder^\times \to (\order/N\order)^\times$. Then:
  \begin{enumerate}[(1)]\setcounter{enumi}{2}
    \item\label{item:congrprop} there exists an isomorphism of functors
      \[
        \cF \cong \bigsqcup_i \cF_{U_i}.
      \]
  \end{enumerate}
\end{theorem}
\begin{proof}
  First, every functor of $N$-congruence type clearly satisfies the $(\bmod
  N)$-congruence property.

  Assume that~$\cF$ satisfies the $\modN$-congruence property.

  \textbf{Step 1}. We claim that all~$\cF(f)$ are bijections.
  Indeed, let~$f\in\End_\Sigma([x])$ be an endomorphism
  in~$\Cosets_\Sigma(\hatorder^\times)$. Then its reduction modulo~$N$ has finite
  order, say~$k$, so that~$f^k-1 \in Nx\hatorder x^{-1}$. By the $\modN$-congruence
  property this implies~$\cF(f)^k = \cF(1) = \id$, so that~$\cF(f)$ is
  invertible.
  Now let~$f\colon [x]\to [y]$ be an arbitrary morphism
  in~$\Cosets_\Sigma(\hatorder^\times)$. Then there exists a
  morphism~$g\colon [y]\to [x]$ such that~$fg = \deg(f)\in\End_\Sigma([y])$ and
  $gf = \deg(f)\in \End_\Sigma([x])$. By the endomorphism case, this proves
  that~$\cF(fg)$ and~$\cF(gf)$ are invertible, hence that~$\cF(f)$ is.
  This proves~(\ref{item:bijU}).


  \textbf{Step 2}. We claim that for all~$x,y\in \hatB^\times$ and
  all~$a,b\in\Hom_\Sigma([x],[y])$, if~$a-b\in Ny\hatorder x^{-1}$ then~$\cF(a)
  = \cF(b)$. Indeed under these conditions there exists~$c\in
  \Hom_\Sigma([y],[x])$. Then~$ca-cb\in \End_\Sigma([x]) \cap Nx\hatorder
  x^{-1}$. By the $\modN$-congruence property we have~$\cF(ca) = \cF(cb)$.
  Since~$\cF(c)$ is invertible, this proves~$\cF(a) = \cF(b)$.

  We also claim that for all~$x,y\in \hatB^\times$ and
  all~$a,b\in\Hom_\Sigma([x],[y])$, if~$b\in axU_\Sigma(N) x^{-1}$ then~$\cF(a)
  = \cF(b)$. Indeed the condition implies that
  \[
    b\in a(1+Nx\hatorder_\Sigma x^{-1})
    = a + Nax\hatorder_\Sigma x^{-1}
    = a + Ny\hatorder_\Sigma x^{-1}.
  \]
  Since~$\hatorder \cap N\hatorder_\Sigma = N\hatorder$, the previous claim
  applies, and therefore~$\cF(a)=\cF(b)$.

  \textbf{Step 3}. Inspired by Proposition~\ref{prop:FUcosetsU}, we are going to define
  an action of~$\hatB^\times$ on~$\bigsqcup_{[x]}\cF([x])$.
  Let~$g,x\in\hatB^\times$. By Lemma~\ref{lem:strongapprox}, there
  exists~$b\in\Hom_\Sigma([x],[gx])\cap gx U_\Sigma(N) x^{-1}$.
  For~$A\in\cF([x])$, we define~$g\cdot A = \cF(b)(A)\in\cF([gx])$.
  To see that this is well-defined, let~$b'\in\Hom_\Sigma([x],[gx])\cap gx
  U_\Sigma(N) x^{-1}$ be another element. We have~$b'\in bxU_\Sigma(N)x^{-1}$,
  so that~$\cF(b) = \cF(b')$.

  The defined action is multiplicative, because
  when~$d\in\Hom_\Sigma([x],[hx])\cap hxU_\Sigma(N)x^{-1}$
  and~$c\in\Hom_\Sigma([hx],[ghx]) \cap ghx U_\Sigma(N)(hx)^{-1}$
  we have~$cd\in\Hom_\Sigma([x],[ghx]) \cap ghx U_\Sigma(N) x^{-1}$: the action
  of~$gh$ is given by~$\cF(cd)=\cF(c)\cF(d)$.

  We therefore get an action of~$\hatB^\times$ on~$\bigsqcup_{[x]}\cF([x])$ with
  the following properties for~$g,x\in\hatB^\times$:
  \begin{itemize}
    \item the action of~$g$ induces a bijection~$\cF([x]) \to \cF([gx])$;
    \item $x\hatorder^\times x^{-1}$ stabilises~$\cF([x])$;
    \item $x U(N) x^{-1}$ acts trivially on~$\cF([x])$.
  \end{itemize}
  In particular, we obtain an action of~$\hatorder^\times$ on~$\cF([1])$.
  By decomposing this action into orbits, we obtain
  an~$\hatorder^\times$-equivariant bijection
  \[
    \psi_{[1]} \colon \bigsqcup_{i\in I} \hatorder^\times/U_i \longrightarrow
    \cF([1]),
  \]
  where the~$U_i$ are subgroups of~$\hatorder^\times$ containing~$U(N)$.
  Recalling the definition of the functors~$\cF_{U_i}$ (Definition~\ref{def:FU}) we see that this is the
  same as 
  an~$\hatorder^\times$-equivariant bijection
  \[
    \psi_{[1]} \colon \bigsqcup_{i\in I} \cF_{U_i}([1]) \longrightarrow
    \cF([1]).
  \]

  In fact, the action of~$\hatorder^\times$ factors through~$\hatorder^\times
  \to \hatorder^\times/U(N) \cong (\order/N\order)^\times$, and comes from the
  application of Lemma~\ref{lem:strongapprox} to~$x=1$
  and~$g\in\hatorder^\times$, from which we see that~(\ref{item:actionU}) holds
  and one can choose the~$U_i$ compatibly with the~$H_i$ from the statement of
  the theorem.

  \textbf{Step 4}. We extend~$\psi_{[1]}$ to an isomorphism of functors~$\psi\colon \bigsqcup_{i\in
  I} \cF_{U_i} \to \cF$.
  Let~$x\in\hatB^\times$. We define
  \[
    \psi_{[x]} \colon \bigsqcup_{i\in I} \cF_{U_i}([x]) \longrightarrow
    \cF([x])
  \]
  by setting for every~$U=U_i$ and every~$xgU_\Sigma\in\cF_{U}([x])$ with~$g\in
  \hatorder_\Sigma^\times$,
  \[
    \psi_{[x]}(xgU_\Sigma) = x\cdot\psi_{[1]}(gU_\Sigma).
  \]
  The map~$\psi_{[x]}$ is well-defined since~$xgU_\Sigma = xg'U_\Sigma$
  implies~$gU_\Sigma = g'U_\Sigma$.
  In addition, $\psi_{[x]}$ depends only on~$[x]$: for all~$u\in\hatorder^\times$ we
  have~$(xu)\cdot\psi_{[1]}(u^{-1}gU_\Sigma) = x\cdot\psi_{[1]}(gU_\Sigma)$
  by~$\hatorder^\times$-equivariance.
  Since the multiplication by~$x^{-1}$ from~$\cF_U([x])$ to~$\cF_U([1])$, the
  map~$\psi_{[1]}$ and the action of~$x$ from~$\cF([1])$ to~$\cF([x])$ are all
  bijections, the map~$\psi_{[x]}$ is a bijection.
  We now prove that~$\psi = (\psi_{[x]})_{[x]}$ is a morphism of functors.

  The proof will follow the following diagram:
  \[
  \xymatrix{
    xgU_\Sigma \ar@{|->}[rrr] \ar@{|->}[ddd] 
    &
    &
    & \cF(b)\psi_{[1]}(gU_\Sigma) \ar@{|->}[dd]
    \\
    & \cF_U([x]) \ar[r]^{\psi_{[x]}} \ar[d]_{\cF_U(f)}
    & \cF([x]) \ar[d]^{\cF(f)}
    &
    \\
    & \cF_U([y]) \ar[r]_{\psi_{[y]}}
    & \cF([y])
    & \cF(fb)\psi_{[1]}(gU_\Sigma)
    \\
    yugU_\Sigma \ar@{|->}[rr]
    &
    & \cF(cd)\psi_{[1]}(gU_\Sigma) \ar@{=}[ur]
    &
  }
  \]

  Let~$x,y\in\hatB^\times$, $f\in\Hom_\Sigma([x],[y])$ and~$U$ be one of
  the~$U_i$; we will prove
  that~$\cF(f)\circ \psi_{[x]} = \psi_{[y]} \circ \cF_U(f)$ holds on~$\cF_U([x])$.
  Let~$xgU_\Sigma\in \cF_U([x])$.
  We have
  \[
    \cF(f)\circ \psi_{[x]} (xgU_\Sigma)
    = \cF(f)\bigl( x\cdot \psi_{[1]}(gU_\Sigma) \bigr)
    = \cF(fb)\psi_{[1]}(gU_\Sigma),
  \]
  where~$b\in\Hom_\Sigma([1],[x])\cap xU_\Sigma(N)$.
  Write~$f = yux^{-1}$ with~$u\in\hatorder\cap \order_\Sigma^\times$, and note
  that~$fxgU_\Sigma = yugU_\Sigma \in \cF_U([y])$.
  We therefore have
  \[
    \psi_{[y]} \circ \cF_U(f) (xgU_\Sigma)
    = \psi_{[y]} (fxgU_\Sigma)
    = y\cdot \psi_{[1]}(ugU_\Sigma)
    = \cF(c)\psi_{[1]}(ugU_\Sigma)
  \]
  where~$c\in \Hom_\Sigma([1],[y]) \cap y U_\Sigma(N)$.
  Now~$u\in v U_\Sigma(N)$ for some~$v\in \hatorder^\times$, so that~$ug
  U_\Sigma = vg U_\Sigma$ since~$\hatorder_\Sigma^\times$
  normalises~$U_\Sigma(N)$, and by equivariance of~$\psi_{[1]}$ we have
  \[
    \psi_{[1]}(ugU_\Sigma)
     = \psi_{[1]}(vgU_\Sigma)
     = v\cdot \psi_{[1]}(gU_\Sigma)
     = \cF(d) \psi_{[1]}(gU_\Sigma)
  \]
  where~$d\in\End_\Sigma([1]) \cap vU_\Sigma(N)$.
  We get
  \[
    \psi_{[y]} \circ \cF_U(f) (xgU_\Sigma)
    = \cF(cd)\psi_{[1]}(gU_\Sigma).
  \]
  We finally compare~$fb$ and~$cd$.
  We have
  \[
    fb \in (yux^{-1})xU_\Sigma(N) = yuU_\Sigma(N),
  \]
  and
  \[
    cd \in y U_\Sigma(N) v U_\Sigma(N) = yv U_\Sigma(N) = yu U_\Sigma(N).
  \]
  Since~$fb$ and~$cd$ both belong to~$\Hom_\Sigma([1],[y]) =
  \Hom_\Sigma([1],[yu])$, this proves that~$\cF(fb) = \cF(cd)$.
  This proves that~$\psi$ is an isomorphism of functors, so that~$\cF$ is of
  $N$-congruence type, proving~(\ref{item:congrprop}) and concluding the proof.
\qed\end{proof}

\subsection{Associated graphs and equidistribution}\label{subsec:equidistr}

In this subsection, we study the graphs of interest and prove our main
equidistribution theorem: Theorem~\ref{thm:equidistrE} and its companion
Proposition~\ref{prop:companionequidistr}.

We first introduce a categorical construction of graphs generalising
Definition~\ref{def:graphF}.

\begin{definition}\label{def:graphC}
  Let~$\cC$ be category with finitely many isomorphism classes of objects, finite
  automorphism groups, and equipped with a degree map. We define the graph~$\gr(\cC)$ with:
  \begin{itemize}
    \item vertices: isomorphism classes of objects in~$\cC$;
    \item edges: let~$x\in \cC$; the set of edges from the vertex corresponding
      to~$x$ is the set of classes of morphisms from~$x$ modulo the
      relation~$(f\colon x\to y) \sim (g\colon x\to z)$ if and only there
      exists~$u\in\Isom_{\cC}(y,z)$ such that~$g=uf$;
      the endpoint of the edge corresponding to~$f\colon x \to y$ is the
      isomorphism class of~$y$.
      In other words, the set of edges between the classes
      of~$x,y\in\cC$ is~$\Aut(y)\lquo \Hom(x,y)$.
  \end{itemize}

  The \emph{degree} of an edge is the degree of the corresponding morphism.

  We define a measure on the set of vertices of~$\gr(\cC)$ by giving each
  vertex~$v$ measure~$\frac{1}{\#\Aut(x)}$
  where~$v$ corresponds to~$x\in\cC$, and we
  write~$L^2(\gr(\cC))$ the Hilbert space of complex functions on the set of
  vertices of~$\gr(\cC)$.

  For every prime~$\ell$, we define an adjacency operator~$A_\ell$
  on~$L^2(\gr(\cC))$ given by
  \[
    A_\ell F(x) = \sum_{x\to y} F(y),
  \]
  where the sum runs over edges of degree~$\ell$ leaving~$x$.
\end{definition}

\begin{remark}
  For every functor~$\cF$ as in Definition~\ref{def:graphF}, we have~$\cG_\cF =
  \gr(\El(\cF))$.
  Every degree-preserving equivalence of categories~$\cC\cong\cD$ induces an
  isomorphism of graphs~$\gr(\cC)\cong\gr(\cD)$ compatible with all the
  structure from Definition~\ref{def:graphC}.
\end{remark}

The following lemma relates the graphs obtained from our quaternionic categories
to automorphic forms.

\begin{lemma}\label{lem:graphhecke}
  The category~$\Cosets_\Sigma(U)$ and its associated graph have the following
  properties:
  \begin{enumerate}[(1)]
    \item\label{item:graphdoublequo} Two objects~$x,y\in \hatB^\times / U$ are isomorphic if and only if they have
      the same image in the quotient~$B^\times \lquo \hatB^\times / U$.

    \item The projection to~$B^\times \lquo \hatB^\times / U$ of a Haar measure on~$\hatB^\times$ coincides with the
      measure on the set of vertices of~$\gr(\Cosets_\Sigma(U))$.

    \item\label{item:graphedges} For every~$x\in\hatB^\times$, the map
      \[
        \edg\colon u\in \hatorder\cap U_\Sigma
        \longmapsto
        1\in\Hom_U(xU,xu^{-1}U)
      \]
      induces a bijection between~$U\lquo (\hatorder\cap U_\Sigma)$ and the set of
      edges leaving the vertex~$B^\times x U$ in~$\gr(\Cosets_\Sigma(U))$.

    \item\label{item:graphhecke} For every prime~$\ell\in\Sigma$ different from~$p$, the adjacency
      operator~$A_\ell$ coincides with the adjoint of the Hecke
      operator~$T_\ell$ on~$L^2(B^\times \lquo \hatB^\times / U)$.
  \end{enumerate}
\end{lemma}
\begin{proof}\hfill
  \begin{enumerate}[(1)]
    \item Isomorphisms in~$\Cosets_\Sigma(U)$ are exactly morphisms of
      degree~$1$, so that the set of isomorphisms between two
      cosets~$xU,yU$ is~$B^\times \cap yUx^{-1}$, i.e. the
      set of elements~$b\in B^\times$ such that~$bxU = yU$. This proves the claim.
    \item Since cosets of~$U$ are open and form a disjoint union, by translation
      invariance every element of~$\hatB^\times/U$ has the same nonzero measure.
      We normalise the Haar measure so that each coset of~$U$ has measure~$1$.
      For every~$x\in\hatB^\times$, every fiber of the projection map~$xU
      \mapsto B^\times x U$ has cardinality~$\# (B^\times\cap xUx^{-1}) =
      \#\Aut_U(xU)$, so the projected measure of~$B^\times x U$ is the inverse
      of this cardinality, as claimed.
    \item Let~$x\in \hatB^\times$. For every~$u\in \hatorder\cap U_\Sigma$, we
      have~$1 = (xu^{-1})ux \in xu^{-1}(\hatorder\cap U_\Sigma)x \cap B^\times =
      \Hom_U(xU,xu^{-1}U)$, so the map~$\edg$ is well-defined.
      Let~$f\in\Hom_U(xU,yU)$ represent an edge~$x\to y$
      in~$\gr(\Cosets_\Sigma(U))$. Then~$1\in\Hom_U(xU,f^{-1}yU)$ represents the
      same edge, which is therefore~$\edg(y^{-1}fx)$.
      Moreover, two morphisms~$f\in \Hom_U(xU,yU)$ and~$g\in \Hom_U(xU,zU)$
      represent the same edge if and only if there exists~$b\in B^\times$ such
      that~$g = bf$. For morphisms in the image of~$\edg$, this can only happen
      with~$b=1$, so the edge~$\edg(u)$ is completely determined by its
      endpoint~$xu^{-1}U$, i.e. by the coset~$Uu$.
    \item Consider the cosets~$Uu\in U\lquo (\hatorder\cap U_\Sigma)$ such that~$\nrd(u)\hatZ
      = \ell\hatZ$. Then for every~$\ell'\neq \ell$, the $\ell'$-component
      of~$u$ is in the~$\ell'$-component of~$U$, so we may replace it by~$1$.
      Choosing an
      isomorphism~$\order\otimes\Z_\ell \cong M_2(\Z_\ell)$, the possible
      cosets correspond to the cosets in~$\GL_2(\Z_\ell)\lquo \GL_2(\Q_\ell)$
      whose determinant has valuation~$1$: these are exactly the cosets
      of~$\GL_2(\Z_\ell)$ that belong to the double coset
      \[
        \GL_2(\Z_\ell)
        \begin{pmatrix}\ell & 0 \\ 0 & 1\end{pmatrix}
        \GL_2(\Z_\ell).
      \]
      Let~$F\in L^2(\gr(\Cosets_U(\Sigma)))$.
      From the above we have
      \[
        A_\ell F(B^\times xU) = \sum_{Uu\in U\lquo U\delta_\ell U} F(B^\times xu^{-1}U).
      \]
      This is the adjoint of the Hecke operator~$T_\ell$, as claimed.
  \end{enumerate}
\qed\end{proof}

The following proposition is the quaternionic version of our equidistribution
result.

\begin{proposition}\label{prop:equidistr}
  Let~$L^2_{\nrd}(\gr(\Cosets_\Sigma(U))) \subset
  L^2(\gr(\Cosets_\Sigma(U)))$ denote the subspace of functions that factor
  through the reduced norm map
  \[
    B^\times \lquo \hatB^\times / U \longrightarrow \Q^\times_{>0} \lquo
    \hatQ^\times / \nrd(U),
  \]
  and let~$L^2_0(\gr(\Cosets_\Sigma(U)))$ denote the orthogonal complement
  of the subspace~$L^2_{\nrd}(\gr(\Cosets_\Sigma(U)))$.
  Then, for every~$\ell\in\Sigma$ different from~$p$, the adjacency operator~$A_\ell$
  is a normal operator,
  stabilises~$L^2_{\nrd}(\gr(\Cosets_\Sigma(U)))$
  and~$L^2_0(\gr(\Cosets_\Sigma(U)))$, and its operator norm
  on~$L^2_0(\gr(\Cosets_\Sigma(U)))$ is at most~$2\sqrt{\ell}$.
  Moreover, the~$A_\ell$ for~$\ell\in\Sigma$ pairwise commute.
\end{proposition}
\begin{proof}
  It is clear that~$L^2_{\nrd}(\gr(\Cosets_\Sigma(U)))$ is stable
  under~$A_\ell$.
  The operators~$A_\ell$ are normal and pairwise commute
  by Lemma~\ref{lem:graphhecke}~(\ref{item:graphhecke}),
  and therefore
  leave~$L^2_0(\gr(\Cosets_\Sigma(U)))$ stable and are diagonalisable.
  We bound the operator norm of~$A_\ell$ by bounding its eigenvalues,
  equivalently by bounding the eigenvalues of the Hecke operator~$T_\ell$.
  The space~$L^2_{\nrd}(\gr(\Cosets_\Sigma(U)))$ is exactly the subspace
  of~$L^2(B^\times\lquo \hatB^\times / U)$ of automorphic forms that generate a
  one-dimensional automorphic representation (i.e. of the form~$g\mapsto
  \chi(\nrd g)$ for some Dirichlet character~$\chi$). Therefore, by the
  Jacquet--Langlands
  correspondence~\cite[Theorem 14.4]{JL}, every system of Hecke eigenvalues appearing
  in~$L^2_0(\gr(\Cosets_\Sigma(U)))$ is also the one attached to a cuspidal
  modular newform of weight~$2$ ramified only at primes dividing~$pN$.
  Therefore, by Deligne's theorem~\cite[Theorem 8.2]{Delignebounds},
  the absolute values of the eigenvalues of~$T_\ell$ are bounded
  by~$2\sqrt{\ell}$. This proves the proposition.
\qed\end{proof}

\begin{remark}
  Using the full statement of the Jacquet--Langlands correspondence, one could
  obtain the exact eigenvalues in terms of classical modular forms. This is not
  needed in our applications.
\end{remark}

We can finally prove Theorem~\ref{thm:equidistrE} and
Proposition~\ref{prop:companionequidistr}.

\begin{proof}[Theorem~\ref{thm:equidistrE} and Proposition~\ref{prop:companionequidistr}]
  First, we use Corollary~\ref{cor:adelicdeuring} to
  transfer the entire situation to the quaternionic category~$\Cosets_\Sigma(\hatorder^\times)$; in
  particular $\cF$ induces a functor~$\cF'\colon \Cosets_\Sigma(\hatorder^\times)
  \to \Sets$. Since the equivalence of Proposition~\ref{prop:adelise} is
  additive, the functor~$\cF'$ satisfies the
  $\modN$-congruence property in the sense of
  Definition~\ref{def:congrprop}, so that we can apply
  Theorem~\ref{thm:congrprop}. From~(\ref{item:bijU}) and~(\ref{item:actionU}) of
  Theorem~\ref{thm:congrprop}, we obtain~(\ref{item:bij})
  and~(\ref{item:action})
  respectively. Moreover, we can choose the~$H_i$ of Theorem~\ref{thm:congrprop}
  to coincide with those of Proposition~\ref{prop:companionequidistr}.
  
  Let~$E\in\catSS(p)$. It is standard that there
  exists~$\psi\in\Hom_\Sigma(E_0,E)$. By~(\ref{item:action}), there
  exists~$\alpha\in \End_\Sigma(E_0)$ whose degree is the
  inverse of~$\deg(\psi)\bmod N$, so that~$\varphi=\psi\alpha\in\Hom_\Sigma(E_0,E)$
  has degree~$1\bmod N$. This proves~(\ref{item:connect}) when one of the curves
  is~$E_0$, and therefore in general by going via~$E_0$.

  By~(\ref{item:connect}), every vertex of~$\cG_\cF$ is connected to one
  above~$E_0$. Moreover, two vertices above~$E_0$ are connected if and only if
  they are related by an element of~$\End_\Sigma(E_0)$, if and only if they are
  in the same orbit under~$G$. In particular there is exactly one vertex of the
  form~$(E_0,x_i)$ in each connected component of~$\cG_\cF$.
  Since every vertex of~$\cG_{\deg}$ has exactly one outgoing edge labelled by
  each element of~$(\Z/N\Z)^\times$, this proves that there is at most one
  morphism of graphs satisfying the properties of~(\ref{item:Degmap}).

  We now prove the existence of~$\Deg$.
  Let~$U_i$ be as in Theorem~\ref{thm:congrprop}.
  Applying~(\ref{item:congrprop})
  of that theorem and Proposition~\ref{prop:FUcosetsU} we obtain a
  degree-preserving equivalence of categories
  \[
    \El(\cF) \cong \bigsqcup_i \Cosets_\Sigma(U_i),
  \]
  inducing an isomorphism of graphs
  \[
    \cG_\cF \cong \bigsqcup_i \gr(\Cosets_\Sigma(U_i)).
  \]
  By Lemma~\ref{lem:graphhecke}~(\ref{item:graphdoublequo})
  and~(\ref{item:graphedges}), the reduced norm map
  \[
    \nrd\colon B^\times \lquo \hatB^\times / U_i
    \to \Q_{>0}^\times \lquo \hatQ^\times / \nrd(U_i)
  \]
  combined with the isomorphism
  \[
    \Q_{>0}^\times \lquo \hatQ^\times / \nrd(U_i)
    \cong (\Z/N\Z)^\times / \nrd(H_i)
  \]
  translates into a graph morphism
  \[
    \Deg\colon \cG_\cF \longrightarrow \cG_{\deg}
  \]
  satisfying the properties of~(\ref{item:Degmap}).
  Since the ad\'elic reduced norm map~$\hatB^\times \to \hatQ^\times$ is surjective, so is~$\Deg$,
  proving~(\ref{item:surj}).

  Let~$L^2_{\Deg}(\cG_\cF)\subset L^2(\cG_\cF)$ be the subspace of functions
  that factor through~$\Deg$.
  Since vertices connected by an edge of degree~$1\bmod N$ clearly have the same
  image under~$\Deg$, we have~$L^2_{\Deg}(\cG_\cF)\subseteq
  L^2_{\deg}(\cG_\cF)$. Moreover, if two vertices have the same image
  under~$\Deg$, then they are in the same connected component of~$\cG_\cF$ from the above
  analysis, and are therefore connected by a single edge by composing the morphisms
  corresponding to a path between them; the degree of this edge must therefore
  be~$1\bmod N$ by the properties of~$\Deg$. So we have the reverse inclusion,
  and~$L^2_{\Deg}(\cG_\cF) = L^2_{\deg}(\cG_\cF)$.
  This proves~(\ref{item:L2deg}) and concludes the proof of
  Proposition~\ref{prop:companionequidistr}.

  Finally, applying Proposition~\ref{prop:equidistr} to each~$U_i$, and the
  isomorphisms above, yields Theorem~\ref{thm:equidistrE}.
\qed\end{proof}

\section{Enriching a $\OneEnd$ oracle}\label{sec:enriching}

In this section, we show how to turn an oracle for the~$\OneEnd$ problem into a
richer oracle with better distributed output. The quality of this enrichment is
quantified in Theorem~\ref{thm:rich-is-conj-invariant}. The proof is an
application of the equidistribution results of Section~\ref{sec:equidistr}.

The following lemma relates conjugation-invariance of distributions to the
abstract setup of Section~\ref{sec:equidistr}.

\begin{lemma}\label{lem:cg}
  Let~$p>3$ be a prime, let~$N\ge 1$ and let~$E\in\catSS(p)$.
  Let~$g\in \left(\End(E) / N\End(E)\right)^\times$ be an element of
  degree~$1\in(\Z/N\Z)^\times$.
  Define the linear operator~$c_g\colon L^2(\cG_{\End/N}) \to L^2(\cG_{\End/N})$ by
  \[
    c_gF(E,\alpha) = F(E,g\alpha g^{-1})
    \text{ and }
    c_gF(E',\alpha') = F(E',\alpha') \text{ for all }E'\neq E.
  \]
  Then:
  \begin{enumerate}[(1)]
    \item\label{item:cgnorm} for all~$F\in L^2(\cG_{\End/N})$, we
      have~$\|c_gF\|^2 \le 3\|F\|^2$;
    \item\label{item:cginvar} for all~$G\in L^2_{\deg}(\cG_{\End/N})$, we
      have~$c_g G = G$; and
    \item\label{item:cgbound} for every~$F=F_0+F_1\in L^2(\cG_{\End/N})$ with~$F_0\in L^2_0(\cG_{\End/N})$ and~$F_1\in
      L^2_{\deg}(\cG_{\End/N})$, we have
      \[
        \|F - c_g F\| \le (1+\sqrt{3})\|F_0\|.
      \]
  \end{enumerate}
\end{lemma}
\begin{proof}\hfill
  \begin{enumerate}[(1)]
    \item Let~$F\in L^2(\cG_{\End/N})$.
      We have, where~$E'$ ranges over the set of supersingular curves up to isomorphism
      except~$E$,
      the elements~$\alpha$ and~$\beta$ range over~$\End(E)/N\End(E)$ and~$\alpha'$ over~$\End(E')/N\End(E')$,
      \[
        \|F\|^2
        = \sum_{(E,\alpha)}\frac{1}{\#\Aut(E,\alpha)}|F(E,\alpha)|^2 +
        \sum_{(E',\alpha')}\frac{1}{\#\Aut(E',\alpha')}|F(E',\alpha')|^2,
      \]
      and
      \begin{eqnarray*}
        \|c_g F\|^2
        &=& \sum_{(E,\alpha)}\frac{|F(E,g\alpha g^{-1})|^2}{\#\Aut(E,\alpha)} +
        \sum_{(E',\alpha')}\frac{|F(E',\alpha')|^2}{\#\Aut(E',\alpha')} \\
        &=& \sum_{(E,\beta)}\frac{|F(E,\beta)|^2}{\#\Aut(E,g^{-1}\beta g)} +
        \sum_{(E',\alpha')}\frac{|F(E',\alpha')|^2}{\#\Aut(E',\alpha')} \\
        &=& \sum_{(E,\beta)}\frac{\#\Aut(E,\beta)}{\#\Aut(E,g^{-1}\beta g)}\frac{|F(E,\beta)|^2}{\#\Aut(E,\beta)} +
        \sum_{(E',\alpha')}\frac{|F(E',\alpha')|^2}{\#\Aut(E',\alpha')} \\
        &\le& 3\sum_{(E,\beta)}\frac{|F(E,\beta)|^2}{\#\Aut(E,\beta)} +
        \sum_{(E',\alpha')}\frac{|F(E',\alpha')|^2}{\#\Aut(E',\alpha')} 
          \le 3\|F\|^2,
      \end{eqnarray*}
      where the inequality comes from~$\#\Aut(E,\beta) \le 6$
      and~$\#\Aut(E,g^{-1}\beta g) \ge 2$.
      
    \item Let~$h\in \End_\Sigma(E)$ be a lift of~$g$, which exists by
      Proposition~\ref{prop:companionequidistr}~(\ref{item:action}).
      Let~$G\in L^2_{\deg}(\cG_{\End/N})$. For every~$E'\neq E$ we have~$c_gG(E',\alpha) =
      G(E',\alpha)$. Moreover, $h$ defines an edge~$(E,\alpha) \to (E,h\alpha
      \hat h) = (E,h\alpha h^{-1})$ in~$\cG_{\End/N}$ of
      degree~$1\bmod N$, so~$G(E,\alpha) = G(E,h\alpha h^{-1})$.
      Since~$c_gG(E,\alpha) = G(E,g\alpha g^{-1}) = G(E,h\alpha h^{-1})$, this
      proves that~$c_g G = G$.

    \item Let~$F=F_0+F_1\in L^2(\cG_{\End/N})$ with~$F_0\in L^2_0(\cG_{\End/N})$ and~$F_1\in L^2_{\deg}(\cG_{\End/N})$.
      Then 
      \begin{eqnarray*}
        \|F-c_g F\| &\le& \|F_0-c_gF_0\| + \|F_1-c_gF_1\| \\
        &=& \|F_0-c_gF_0\| \text{ since }c_gF_1 = F_1 \text{ by (\ref{item:cginvar})} \\
        &\le& (1+\sqrt{3})\|F_0\| \text{ by (\ref{item:cgnorm})}.
      \end{eqnarray*}
  \end{enumerate}
\qed\end{proof}

 \begin{algorithm}[h]\label{algo:enrich}
 \caption{$\Rich^\Oracle_k$: turning an oracle $\Oracle$ for $\OneEnd$ into a
   `richer' oracle $\Rich^\Oracle_k$, with guarantees on the distribution of the output.}\label{alg:Enrich}
 \begin{algorithmic}[1]
 \REQUIRE {A supersingular elliptic curve $E/\F_{p^2}$, and a parameter
   $k\in\Z_{>0}$. We suppose access to an oracle $\Oracle$ that solves the $\OneEnd$ problem.}
 \ENSURE {An endomorphism $\alpha \in \End(E)$.}
\STATE $\varphi \gets$ a $2$-isogenies random walk of length~$k$ from $E$
   \STATE $E'\gets $ endpoint of~$\varphi$
\STATE $\alpha \gets \Oracle(E')$, a non-scalar endomorphism of $E'$
\RETURN $\hat \varphi \circ \alpha \circ \varphi$
\end{algorithmic}
 \end{algorithm}

 We can now prove the main result of this section.

\begin{theorem}\label{thm:rich-is-conj-invariant}
  Let~$p>3$ be a prime and~$N$ an odd integer. Let~$\Oracle$ be an oracle for $\OneEnd$.
  Let~$E$ be a supersingular elliptic curve defined over~$\F_{p^2}$ and
  let $\alpha$ be the random endomorphism produced by $\Rich^\Oracle_k(E)$.
  Then for every element $g\in \left(\End(E) / N\End(E)\right)^\times$ of
  degree~$1\in(\Z/N\Z)^\times$,
  the statistical distance between the distribution of  $\alpha \bmod N$
  and the distribution of $g^{-1}(\alpha  \bmod N)g$
  is at most
  \[
    \frac{1+\sqrt{3}}{4}\lambda^kN^2\sqrt{p+13}
    = O(\lambda^kN^2\sqrt{p}),
  \]
  where~$\lambda = \frac{2\sqrt{2}}{3}\approx 0.94$.
\end{theorem}
\begin{proof}
  Define~$F\in L^2(\cG_{\End/N})$ by the following formula for every vertex~$(E',\beta)$:
  \[
    F(E',\beta) = \Pr[\Oracle(E') \bmod N = \beta],
  \]
  so that
  \[
    \Bigl(\frac{A_2}{3}\Bigr)^k F(E,\beta) = \Pr[\Rich^\Oracle_k(E) \bmod N = 4^k \beta].
  \]
  Indeed, $\bigl(\frac{A_2}{3}\bigr)^k F(E,\beta)$ is the average, over all
  random walks~$\varphi\colon E \to E'$ that Algorithm~\ref{algo:enrich} could
  follow from~$E$,
  of~$\Pr[\Oracle(E') \bmod N = \varphi \beta \hat\varphi]$, and the equality~$\Oracle(E') \bmod N =
  \varphi \beta \hat\varphi$ is equivalent to~$\hat\varphi\Oracle(E')\varphi
  \bmod{N} =
  \deg(\varphi) \beta \deg(\hat\varphi) = 4^k \beta$ since~$2$ is invertible
  mod~$N$.

  We have, where~$E'$ ranges over isomorphism classes in~$\catSS(p)$ and~$\beta$
  over the set~$\End(E')/N\End(E')$,
  \begin{eqnarray*}
    \|F\|^2
    &=& \sum_{(E',\beta)}\frac{1}{\#\Aut(E',\beta)}\Pr[\Oracle(E') \bmod N = \beta]^2 \\
    &\le& \frac{1}{2}\sum_{(E',\beta)}\Pr[\Oracle(E') \bmod N = \beta] \\
    &=& \frac{1}{2}\sum_{E'}1
      \le \frac{p+13}{24} \text{ by \cite[Theorem 4.1~(c)]{Silverman-Arithmetic}}.
  \end{eqnarray*}

  Write~$F = F_0 + F_1$ with~$F_0\in L^2_0(\cG_{\End/N})$ and~$F_1\in L^2_{\deg}(\cG_{\End/N})$.
  Since~$A_2$ preserves the orthogonal decomposition~$L^2_0(\cG_{\End/N}) \oplus L^2_{\deg}(\cG_{\End/N})$,
  we may apply Lemma~\ref{lem:cg}~(\ref{item:cgbound}) to~$A_2^k F = A_2^k F_0 +
  A_2^k F_1$, giving
  \[
    \|A_2^kF - c_gA_2^k F\| \le (1+\sqrt{3})\|A_2^k F_0\|.
  \]
  On the other hand, by Theorem~\ref{thm:equidistrE} we have
  \[
    \left\|\Bigl(\frac{A_2}{3}\Bigr)^k F_0\right\| \le \lambda^k\|F_0\| \le \lambda^k\|F\|.
  \]
  Finally, with~$\beta$ ranging over~$\End(E)/N\End(E)$, the statistical
  distance in the statement of the theorem is
  \begin{eqnarray*}
    && \frac 12\sum_{\beta}\bigl|\Pr[\Rich^\Oracle_k(E) \bmod N = \beta] - \Pr[\Rich^\Oracle_k(E) \bmod N = g\beta g^{-1}]\bigr| \\
    &=& \frac 12\sum_{\beta}\left|\Bigl(\frac{A_2}{3}\Bigr)^kF(E,4^{-k}\beta) - c_g\Bigl(\frac{A_2}{3}\Bigr)^kF(E,4^{-k}\beta)\right| \\
    &=& \frac 12\sum_{\beta}\left|\Bigl(\frac{A_2}{3}\Bigr)^kF(E,\beta) - c_g\Bigl(\frac{A_2}{3}\Bigr)^kF(E,\beta)\right|
    \text{ since }\beta \mapsto 4^k\beta \text{ is a bijection} \\
    &\le& \frac 12\left(N^4\sum_{\beta}\Bigl|\Bigl(\frac{A_2}{3}\Bigr)^kF(E,\beta) -
    c_g\Bigl(\frac{A_2}{3}\Bigr)^kF(E,\beta)\Bigr|^2 \right)^{\frac{1}{2}}
      \text{ by the Cauchy--Schwarz inequality}\\
    &\le& \frac 12N^2\sqrt{6}\left\|\Bigl(\frac{A_2}{3}\Bigr)^kF -
      c_g\Bigl(\frac{A_2}{3}\Bigr)^kF\right\|
      \text{ since }\#\Aut(E,\beta) \le 6\\
    &\le& \frac 12(1+\sqrt{3})N^2\sqrt{6}\left\|\Bigl(\frac{A_2}{3}\Bigr)^kF_0\right\|
      \le \frac 12(1+\sqrt{3})\lambda^kN^2\sqrt{6}\|F\| \\
    &\le& \frac 12(1+\sqrt{3})\lambda^kN^2\sqrt{6\cdot\frac{p+13}{24}}
      = \frac 14(1+\sqrt{3})\lambda^kN^2\sqrt{p+13}
      \text{, as claimed.}
  \end{eqnarray*}
\qed\end{proof}


\section{On conjugacy-invariant distributions}\label{sec:conj-invar}

\noindent
Theorem~\ref{thm:rich-is-conj-invariant} proves that given a $\OneEnd$ oracle, the randomization method allows one to sample endomorphisms from a distribution which is (locally) invariant under conjugation by $\left(\End(E) / N\End(E)\right)^\times$. In this section, we study such conjugacy-invariant distributions, and show that with good probability,  such endomorphisms generate interesting suborders. 
In the whole section, fix~$B$ a quaternion algebra over $\Q$ and~$\order\subset B$ a maximal order. 

\subsection{The local case}
We start by studying the local case. Let  $\ell$ be a prime unramified in~$B$.
In this subsection, we study 
distributions on $M_2(\F_\ell) \cong \cO / \ell\cO$ and $M_2(\Z_\ell) \cong \cO_\ell$.

\begin{definition}
The distribution of a random $\alpha \in M_2(\F_\ell)/\F_\ell$ is \emph{$\varepsilon$-close to $\SL_2(\F_\ell)$-invariant} if,
for every $g \in \SL_2(\F_\ell)$, the statistical distance between the
  distributions of  $\alpha$ and of $g^{-1}\alpha g$ is at most~$\varepsilon$.
  When the distributions are the same (i.e., $\varepsilon = 0$), we say that the
  distribution of~$\alpha$ is \emph{$\SL_2(\F_\ell)$-invariant}.
 \end{definition}

A key observation is that a conjugacy class cannot be stuck in a subspace.
\begin{lemma}\label{lemma:init-lemma-F-ell-focused}
  Suppose $\ell>2$.
  Let $\alpha \in M_2(\F_\ell)\setminus\F_\ell$. Let $V \subsetneq
  M_2(\F_\ell)/\F_\ell$ be an $\F_\ell$-linear subspace.
  Let $\beta \in M_2(\F_\ell)$ be a random element uniformly distributed in the $\SL_2(\F_\ell)$-conjugacy class of $\alpha$.
  Then, $\beta \in V$ with probability at most $1/2$.
\end{lemma}

\begin{proof}

  The size of the orbit~$X$ of~$\alpha$ is~$\#\SL_2(\F_\ell) / \#C$, where~$C$ is
  the centraliser of~$\alpha$ in~$\SL_2(\F_\ell)$. The size of this centraliser
  can be~$\ell+1, \ell-1$ or~$2\ell$, so~$\#X \ge \frac{\ell^2-1}{2}$.

  We now bound~$\#(X\cap V)$ by noting that every element~$v$ of this intersection
  satisfies the quadratic equation~$\disc(v) = \disc(\alpha)$.
  The discriminant quadratic form on~$M_2(\F_\ell)/\F_\ell$ is isomorphic
  to~$x^2-yz$, so the maximal dimension of a totally isotropic subspace is~$1$.
  If~$\dim V = 1$, the number of solutions is at most~$\ell$.
  If~$\dim V = 2$, either the equation is degenerate and has at most~$2\ell$ solutions,
  or it represents a  conic and has at most~$\ell+1$ solutions.

  So the probability of~$\beta\in V$ is at most~$2\ell/\frac{\ell^2-1}{2}  =
  \frac{4\ell}{\ell^2-1}$, which is less than~$1/2$ for~$\ell\ge 11$.
  We check the bound by bruteforce enumeration for~$\ell\in\{3,5,7\}$.
\qed\end{proof}


%
%

\begin{lemma}\label{lem:probability-new-samples-local}
Suppose $\ell > 2$.
Let $\alpha_1,\alpha_2,\alpha_3 \in M_2(\F_\ell) / \F_\ell$
be independent non-zero $\SL_2(\F_\ell)$-invariant elements.
Then, $(\alpha_1,\alpha_2,\alpha_3)$ is a basis of $M_2(\F_\ell)/\F_\ell$ with probability at least $1/8$.
\end{lemma}

\begin{proof}
Let $V_1 = \{0\}$ and $V_{i} = V_{i-1} + \F_\ell\cdot\alpha_i$.
By dimensionality, we have $V_i \neq M_2(\F_\ell)/\F_\ell$ for every $i<3$.
Lemma~\ref{lemma:init-lemma-F-ell-focused} implies that with probability at least~$1/8$, we have $\alpha_i \not\in V_{i-1}$ for each $i$. When this occurs, each $V_i$ is an $\F_\ell$-vector space of dimension $i$, hence, $V_3 = M_2(\F_\ell)/\F_\ell$.
\qed\end{proof}

In our application, we will only approach $\SL_2(\F_\ell)$-invariance, so we now derive the corresponding result for distributions that are close to $\SL_2(\F_\ell)$-invariant.

\begin{proposition}\label{prop:probability-new-samples-local-2}
Suppose $\ell > 2$.
Let $\alpha_1,\alpha_2,\alpha_3 \in M_2(\F_\ell) / \F_\ell$
be independent non-zero random elements which are $\varepsilon$-close to $\SL_2(\F_\ell)$-invariant.
Then, $(\alpha_1,\alpha_2,\alpha_3)$ is a basis of $M_2(\F_\ell)/\F_\ell$ with probability at least $1/8 -3\varepsilon$.
\end{proposition}

\begin{proof}
Let $g_i \in \SL_2(\F_\ell)$ be uniformly distributed and independent. Let $\beta_i = g_i^{-1}\alpha_ig_i$, three independent variables.
For each~$i$, the statistical distance between $\alpha_i$ and $\beta_i$ is at most $\varepsilon$.
By the triangle inequality, the statistical distance between $(\alpha_1,\alpha_2,\alpha_3)$ and $(\beta_1,\beta_2,\beta_3)$ is at most $3\varepsilon$.
From Lemma~\ref{lem:probability-new-samples-local}, $(\beta_1,\beta_2,\beta_3)$ is a basis of $M_2(\F_\ell)/\F_\ell$ with probability at least $1/8$. Therefore, $(\alpha_1,\alpha_2,\alpha_3)$ is a basis of $M_2(\F_\ell)/\F_\ell$ with probability at least $1/8 - 3\varepsilon$.
\qed\end{proof}

%
%

We now show that these results about $M_2(\F_\ell)$ have consequences in~$M_2(\Z_\ell)$.

\begin{definition}\label{def:level}
  The \emph{level} of $\alpha \in M_2(\Z_\ell) \setminus \Z_\ell$ at $\ell$ is the largest integer $\level_\ell(\alpha)$ such that $\alpha \in \Z_\ell + \ell^{\level_\ell(\alpha)}M_2(\Z_\ell)$.
 \end{definition}

 
\begin{proposition}\label{prop:special-Nakayama}
Suppose $\ell > 2$.
Let $\alpha_1,\alpha_2,\alpha_3 \in M_2(\Z_\ell) \setminus \Z_\ell$
be three elements of level $a$.
Then $(1,\alpha_1,\alpha_2,\alpha_3)$ is a $\Z_\ell$-basis of $\Z_\ell + \ell^{a} M_2(\Z_\ell)$ if and only if $(\alpha_1,\alpha_2,\alpha_3)$ is an $\F_\ell$-basis of
\((\Z_\ell + \ell^{a} M_2(\Z_\ell))/(\Z_\ell + \ell^{a+1}M_2(\Z_\ell)) \cong M_2(\F_\ell)/\F_\ell.\)
\end{proposition}

\begin{proof}
The forward implication is clear. The converse is Nakayama's lemma.
\qed\end{proof}

\subsection{Dealing with hard-to-factor numbers}
In the previous section, we have studied the properties of conjugacy-invariant distributions locally at a prime $\ell$.
However, in our application, we may be confronted to local obstructions at an integer $N$ which is hard to factor; it is then not possible to isolate the primes $\ell$ to apply the results of the previous section.

In this section, fix a positive integer $N$. We imagine that $N$ is hard to factor, and rework the previous results ``\emph{locally at~$N$}''.
We suppose that $B$ does not ramify at any prime factor of $N$. Recall that $\order\subset B$ is a maximal order.


\begin{definition}\label{def:reduced}
An element $\alpha \in \cO$ is $N$-reduced if $\alpha \not\in \Z + N \cO$.
\end{definition}

\begin{lemma}\label{lem:reduced-small-level}
Let $\alpha \in \cO$ be a random variable supported on $N$-reduced elements.
Then, there exist a prime factor $\ell$ of $N$ and an integer $a$ such that $\ell^{a+1}$ divides~$N$ and $\Pr[\level_\ell(\alpha) = a] \geq (\log N)^{-1}$.
\end{lemma}

\begin{proof}
Write the prime factorisation $N = \prod_{i=1}^t\ell_i^{e_i}$.
Let $i$ and $a < e_i$ which maximise the probability $q = \Pr[\level_{\ell_i}(\alpha) = a]$.
We have
\begin{align*}
    \sum_{j = 1}^t\Pr[\level_{\ell_j}(\alpha) < e_j]
& = \sum_{\beta}\Pr[\alpha = \beta] \cdot \#\{j \mid
  \level_{\ell_j}(\beta) < e_j\} \geq 1,
\end{align*}
where the last inequality follows from the fact that the distribution is supported on
  $N$-reduced elements, so for every $\beta$, there exists $j$ such that
  $\level_{\ell_j}(\beta) < e_j$.
  We get
\begin{align*}
1 \leq \sum_{j = 1}^t\Pr[\level_{\ell_j}(\alpha) < e_j] = \sum_{j = 1}^t\sum_{x < e_j} \Pr[\level_{\ell_j}(\alpha) = x]  \leq q \sum_{j = 1}^t e_j \leq q\log(N).
\end{align*}
We deduce $q \geq (\log N)^{-1}$.
\qed
\end{proof}

\begin{definition}
Let $M$ be a ring with an isomorphism $\iota \colon M_2(\Z/N\Z)\to M$.
The distribution of a random $\alpha \in M/\iota(\Z/N\Z)$ is \emph{$\varepsilon$-close to $\SL_2(\Z/N\Z)$-invariant} if,
for every $g \in \iota(\SL_2(\Z/N\Z))$, the statistical distance between the
  distributions of  $\alpha$ and of $g^{-1}\alpha g$ is at most~$\varepsilon$.
 \end{definition}

\begin{lemma}\label{lemma:good-distribution-leads-to-success-mod-n}
Let $R = \Z/N\Z$,  $M = \cO/N\cO \cong M_2(R)$ and $\overline M = M/R$. Let~$\ell$ be a prime factor of $N$, and $a$ an integer such that $\ell^{a+1} \mid N$.
Consider a distribution~$\nu$ on $\overline M$ that is $\varepsilon$-close to $\SL_2(R)$-invariant.
For $\alpha$ sampled from $\nu$, let $q$ be the probability that
$\alpha \neq 0$ and that $a$ is the largest integer such that $\alpha \in \ell^{a}\overline M$.

\begin{enumerate}
\item
Let $\alpha_1,\alpha_2,\alpha_3 \in \overline M$ independent random elements with distribution $\nu$.
Let $\Lambda$ be the subgroup generated by $(\alpha_1,\alpha_2,\alpha_3)$.
We have $\Lambda / \ell^{a+1}\overline M = \ell^{a}\overline M /\ell^{a+1}\overline M$ with probability at least $q^3/8 - 3\varepsilon$.
\item 
Let $\alpha_1,\alpha_2,\alpha_3 \in \cO$
be independent random elements such that $\alpha_i \bmod \Z + N\cO$ follows the distribution $\nu$.
Let $\Lambda$ be the lattice generated by $(1,\alpha_1,\alpha_2,\alpha_3)$.
Then $\Lambda \otimes \Z_{\ell} = (\Z+ \ell^{a}\cO) \otimes \Z_{\ell}$ with probability at least $q^3/8 - 3\varepsilon$.
\end{enumerate}
\end{lemma}

\begin{proof}
\noindent\textbf{Item 1, with $\varepsilon = 0$.}
For any $\alpha \in M$, let $\level_\ell(\alpha)$ be the largest integer such that $\alpha \in R + \ell^{a}M$ when it exists, and $\level_\ell(\alpha) = \infty$ otherwise.
Let $L$ be the event that $\level_\ell(\alpha_i) = a$ for all $i\in\{1,2,3\}$. 
Note that the level is constant over any $\SL_2(R)$-conjugacy class, so conditional on $L$, the variables $\alpha_i$ are still $\SL_2(R)$-invariant.
If $L$ occurs, the random variables $\alpha_i \bmod \ell^{a+1}\overline M$ are non-zero and $\SL_2(R)$-invariant in $\ell^{a}\overline M /\ell^{a+1}\overline M \cong M_2(\F_\ell)/\F_\ell$.
The result follows from Lemma~\ref{lem:probability-new-samples-local} and the fact that $\Pr[L] = q^3$.\\

\noindent\textbf{Item 1, with $\varepsilon > 0$.}
By the triangular inequality, the triple $(\alpha_1,\alpha_2,\alpha_3)$ is $3\varepsilon$-close to a triple of $\SL_2(\Z/N\Z)$-invariant elements.
The result thus follows from the case $\varepsilon = 0$ and the defining property of the statistical distance.\\

\noindent\textbf{Item 2.} This is the combination of Item 1 with Proposition~\ref{prop:special-Nakayama}.
\qed
\end{proof}

\begin{proposition}\label{prop:good-distribution-leads-to-success}
Assume that~$N$ is not a cube.
Let $\alpha_1,\alpha_2,\alpha_3 \in \cO$
be three independent random elements from a distribution~$\alpha$ that satisfies the following properties:
\begin{enumerate}[(1)]
\item \label{prop:good-distribution-leads-to-success-item-1} $\alpha$ is supported on $N$-reduced elements;
\item \label{prop:good-distribution-leads-to-success-item-2} $\alpha \bmod \Z + N\cO$ is $\varepsilon$-close to $\SL_2(\Z/N\Z)$-invariant for $\varepsilon < \frac{1}{6000000\cdot (\log N)^{12}}$.
\end{enumerate}
Let $\Lambda$ be the lattice generated by $(1,\alpha_1,\alpha_2,\alpha_3)$.
With probability $\Omega\left((\log N)^{-12}\right)$, either $\gcd(N,[\cO:\Lambda]) = 1$, or $[\cO:\Lambda] = N^nK$ with $\gcd(N,K) \not\in \{1,N\}$.
\end{proposition}

\begin{remark}
The exhibited event either produces a lattice $\Lambda$ that is saturated at every prime factor of $N$ (when $\gcd(N,[\cO:\Lambda]) = 1$), or reveals a non-trivial factor of $N$.
\end{remark}

\begin{proof}
Let $\Success$ be the event that either $\gcd(N,[\cO:\Lambda]) = 1$, or $[\cO:\Lambda] = N^nK$ where $\gcd(N,K) \not\in \{1,N\}$.
Write the prime factorisation $N = \prod_{i=1}^t\ell_i^{e_i}$.
Since $N$ is not a cube, we may assume without loss of generality that $\gcd(e_1,3) = 1$.
Write $\cO_i = \cO \otimes \Z_{\ell_i}$ and $\Lambda_i = \Lambda \otimes \Z_{\ell_i}$.

We now split the proof in two cases, depending on the value of \[q_+ = \Pr[\level_{\ell_1}(\alpha) \geq  e_1].\]

\noindent
\textbf{Case 1:} suppose $q_+^3 > 1 - \frac 1 2\left(\frac{1}{8\cdot (\log N)^3} - 3 \varepsilon\right)$. 
Let $\ell_i$ and $a_i$ be the $\ell$ and $a$ from Lemma~\ref{lem:reduced-small-level}.
Let $q = \Pr[\level_{\ell_i}(\alpha) = a_i] > (\log N)^{-1}$.
Let $E$ be the event that 
$\Lambda_i = \Z_{\ell_i} + \ell_i^{a_i}\cO_i,$
and let $F$ be the event that 
$\Lambda_1 \subseteq \Z_{\ell_1} + \ell_1^{e_1}\cO_1$.
Suppose $E$ and $F$ both happen. In that situation, $[\cO_i:\Lambda_i] = \ell_i^{3a_i} < \ell_i^{3e_i}$, 
and $[\cO_1:\Lambda_1] \geq [\cO_1: \Z_{\ell_1} + \ell^{e_1}\cO_1] \geq \ell_1^{3e_1}$, hence if $[\cO:\Lambda] = N^nK$, then $\gcd(K,N) \not \in \{1,N\}$.
So if $E$ and $F$ both happen, then $\Success$ happens.
We have 
\begin{align*}
\Pr[F] &= 
\Pr\left[
\bigwedge_{j=1}^3(\level_{\ell_1}(\alpha_j) \geq e_1)\right]
= 
\prod_{j=1}^3\Pr\left[\level_{\ell_1}(\alpha_i) \geq e_1\right] = q_+^3.
\end{align*}
From Lemma~\ref{lemma:good-distribution-leads-to-success-mod-n}, $\Pr[E] = q^3/8 - 3\varepsilon$.
We deduce 
\begin{align*}
\Pr[\Success] &\geq \Pr[E \wedge F] \geq \Pr[E] + \Pr[F] - 1 \\
&= \frac{q^3}{8} - 3\varepsilon + q_+^3 - 1
\geq \frac 3 2\left(\frac{1}{24 \cdot(\log N)^3} -  \varepsilon\right) \geq \frac{1}{32 \cdot(\log N)^3}.
\end{align*}

\noindent
\textbf{Case 2:} suppose $q_+^3 \leq 1 - \frac 1 2\left(\frac{1}{8 \cdot(\log N)^3} - 3 \varepsilon\right)$. 
Let $a_1 < e_1$ which maximises the probability
\[q = \Pr[\level_{\ell_1}(\alpha) = a_1].\]
We have
\[1 - q_+ = \sum_{x < e_1} \Pr[\level_{\ell_1}(\alpha) = x] \leq e_1q.\]
Then 
\[q \geq \frac{1 - q_+}{\log(N)} \geq \frac{1 - q_+^3}{3\log(N)} \geq  \frac{1}{48\cdot (\log N)^4} - \frac{\varepsilon}{2\log(N)}.\]
Let $G$ be the event that $\Lambda_1 = \Z_{\ell_1} + \ell_1^{a_1}M_1$.
If $G$ happens and $[\cO  : \Lambda]$ is of the form $N^nK$ with $\gcd(N,K) = 1$, then $3a_1 = ne_1$. In that situation, $\gcd(e_1,3) = 1$ implies that $3$ divides $n$, and $a_1 < e_1$ implies that $n < 3$; together, these imply $n = 0$, so $\gcd(N,[\cO:\Lambda]) = 1$. This proves that when $G$ happens, then $\Success$ happens.
We deduce
\begin{align*}
\Pr[\Success] &\geq \Pr[G] \geq  \frac{q^3}{8} - 3\varepsilon
= \left(\frac{1}{96 \cdot(\log N)^4} - \frac{\varepsilon}{4\log(N)}\right)^3 - 3\varepsilon\\
&\geq 3\left(\frac{1}{3\cdot 100^3 \cdot (\log N)^{12}} - \varepsilon\right)\\
&\geq \frac{1}{2000000 \cdot (\log N)^{12}},
\end{align*}
which concludes the proof.
\qed\end{proof}

\section{Saturation and reduction}

In this section, we present three algorithms to saturate a known order of endomorphisms
of a supersingular curve, and to reduce an endomorphism (in the sense of Definition~\ref{def:reduced}). The overall strategy is folklore, but only
a crucial new ingredient allows it to work in polynomial time:
the division algorithm due to Robert~\cite{RobertApplications} (see
Proposition~\ref{prop:division}).

Let us start with saturation, which is used to deal with problematic primes in
the main reduction. 
The running time of $\Saturate_\ell$ is polynomial in $\ell$, not in~$\log
\ell$, so we only use it for small $\ell$.

\begin{algorithm}[h]
  \caption{$\Saturate_\ell(R_0)$: turns an order into a super-order which is maximal at $\ell$.}\label{alg:Saturate}
  \begin{algorithmic}[1]
    \REQUIRE {A supersingular elliptic curve $E/\F_{p^2}$, a prime $\ell\neq p$, and an order $R_0 \subset \End(E)$.}
    \ENSURE {An order $R$ such that $R_0 \subseteq R \subseteq \End(E)$ and $R \otimes \Z_\ell$ is maximal.}
    \STATE $L \gets R_0$
    \WHILE {$\gcd(\disc(L),\ell) \neq 1$} 
    \FOR {lattices $L'$ such that $[L':L] = \ell$}
    \STATE $\alpha \gets$ an element of $L$ such that $\alpha/\ell \in L' \setminus L$
    \STATE $\beta \gets \Divide(\alpha,\ell)$ an efficient representation of $\alpha / \ell$
    \COMMENT{Proposition~\ref{prop:division}}
    \IF {$\beta\neq\bot$}
    \STATE $L \gets L + \Z\beta$
    \ENDIF
    \ENDFOR
    \ENDWHILE
    \RETURN $L$
  \end{algorithmic}
\end{algorithm}

\begin{algorithm}[h]
  \caption{$\SaturateRam(R_0)$: turns an order into a super-order which is maximal
   at $p$.}\label{alg:Saturatep}
  \begin{algorithmic}[1]
    \REQUIRE {A supersingular elliptic curve $E/\F_{p^2}$, and an order $R_0 \subset \End(E)$.}
    \ENSURE {An order $R$ such that $R_0 \subseteq R \subseteq \End(E)$ and $R \otimes \Z_p$ is maximal.}
    \STATE $(\order_0,\iota) \gets$ an order~$\order_0$ given by
      multiplication table, and an isomorphism~$\iota\colon \order_0 \to R_0$
    \STATE $\order \gets$ an order containing~$\order_0$, maximal at~$p$
    with~$[\order:\order_0]$ a power of~$p$
    \COMMENT{\cite{maxord1} or~\cite{maxord2}}
    \STATE $(1,b_1,b_2,b_3) \gets$ a basis of~$\order$
    \FOR {$i\in\{1,2,3\}$}
      \STATE Write~$b_i = a_i/p^{k_i}$ with~$a_i\in \order_0$
      \STATE $\alpha_i \gets \iota(a_i)$
          \STATE $\beta_i \gets \Divide(\alpha_i,p^{k_i})$ an efficient representation of $\alpha_i/p^{k_i}$
    \COMMENT{Proposition~\ref{prop:division}}
    \ENDFOR
    \RETURN $\Z+\Z\beta_1+\Z\beta_2+\Z\beta_3$
  \end{algorithmic}
\end{algorithm}

\begin{proposition}
  Algorithm~\ref{alg:Saturate} ($\Saturate_\ell$) is correct and runs in
  time polynomial in~$\ell$ and the size of the input.
\end{proposition}
\begin{proof}
  Since~$\ell\neq p$, the discriminant of the maximal order~$\End(E)$ is coprime
  to~$\ell$. Therefore, at each iteration,
  there is at least one~$L'$ that is contained in~$\End(E)$, so that the division succeeds.
  Every iteration divides the discriminant by~$\ell^2$, so the number of iterations is half the
  valuation of~$\disc(R_0)$ at~$\ell$, which is polynomial.
  At every iteration, there are~$O(\ell^3)$ lattices~$L'$ since they
  correspond exactly to lines~$\ell L'/\ell L \subseteq L/\ell L$.
  Every operation performed in the loops takes polynomial time.
  This proves that the algorithm terminates within the claimed running time.
  Consider the following properties for a lattice~$M$ in~$\End(E)$:
  \begin{enumerate}[(1)]
    \item\label{item:firstprop} $R_0\subseteq M \subseteq \End(E)$;
    \item $[M:R_0]$ is a power of~$\ell$;
    \item\label{item:lastprop}$ [\End(E):M]$ is coprime to~$\ell$.
  \end{enumerate}
  There exists at most one~$M$
  satsifying~(\ref{item:firstprop})--(\ref{item:lastprop}).
  When the algorithm terminates, the lattice~$L$
  satisfies~(\ref{item:firstprop})--(\ref{item:lastprop}).
  On the other hand, there exists an order~$R$
  satisfying~(\ref{item:firstprop})--(\ref{item:lastprop}).
  Therefore~$R=L$, as claimed.
\qed\end{proof}

 \begin{proposition}
   Algorithm~\ref{alg:Saturatep} ($\SaturateRam$) is correct and runs in
   polynomial time.
 \end{proposition}
 \begin{proof}
   Since~$R_0\otimes\Q$ is ramified
   at~$p$, there is a unique order~$R$ containing~$R_0$, maximal at~$p$
   with~$[R:R_0]$ a power of~$p$, and this order is contained in~$\End(E)$.
   This imples~$(\iota\otimes\Q)(\order) = R \subset \End(E)$, so all the
   divisions succeed, the family~$(1,\beta_1,\beta_2,\beta_3)$ is a basis of~$R$, and
   the algorithm is correct.
   All the operations take polynomial time.
 \qed\end{proof}

We now present an algorithm to reduce endomorphisms at odd integers.
 

 \begin{algorithm}[h]
  \caption{$\Reduce_N(\alpha)$: reduces an endomorphism $\alpha$ at $N$.}\label{alg:reduce}
  \begin{algorithmic}[1]
    \REQUIRE {An endomorphism $\alpha \in \End(E)\setminus \Z$ in efficient representation, and an odd integer $N$.}
    \ENSURE {An $N$-reduced endomorphism (Definition~\ref{def:reduced})  $\beta = \frac{\alpha - t}{N^e}$ with $t,e \in \Z$.}
    \STATE $\gamma \gets 2\alpha - \Tr(\alpha)$\label{step:make-trace-zero}
    \REPEAT
        \STATE $\beta \gets \gamma$
        \STATE $\gamma \gets \Divide(\beta,N)$ an efficient representation of $\beta/N$
          \label{alg:condition-reducible}
          \COMMENT{Proposition~\ref{prop:division}}
    \UNTIL{$\gamma = \bot$}
    \IF {$\Tr(\beta) \equiv 0 \mod 4$}
    \RETURN $\Divide(\beta,2)$
    \ELSE
    \RETURN $\Divide(\beta+1,2)$
    \ENDIF
  \end{algorithmic}
\end{algorithm}

\begin{proposition}
   Algorithm~\ref{alg:reduce} ($\Reduce_N$) is correct and runs in
   polynomial time.
\end{proposition}

 \begin{proof}
Let $e$ be the largest integer such that $\alpha \in \Z + N^e \End(E)$.
At Step~\ref{step:make-trace-zero},
we have that $\gamma \in N^e \End(E)$ and $\gamma \not\in \Z + N^{e+1} \End(E)$.
Therefore, at the end of the loop, $\beta \in \End(E)$ and $\beta \not\in \Z + N \End(E)$, i.e., $\beta$ is $N$-reduced. The last division removes the extra factor $2$ introduced in Step~\ref{step:make-trace-zero}, to ensure the result is of the form $\beta = \frac{\alpha - t}{N^e}$ with $t \in \Z$.

Let us prove that it runs in polynomial time.
We have~$N^{2e}\mid \disc(\alpha)$, and at each iteration of the loop,
$\disc(\beta)$ gets divided by~$N^2$. So the number of iterations is bounded
by~$e \le \log(\disc(\alpha)) = O(\log\deg(\alpha))$, which concludes the proof.
\qed\end{proof}
 
 \section{The reduction}\label{sec:the-main-reduction}

 In this section, we prove the main result of the paper (Theorem~\ref{thm:main}).
 We start with a lemma putting together results from the previous
 sections.

 \begin{algorithm}[h!]
 \caption{Turning an oracle $\Oracle$ for $\OneEnd$ into an $\EndRing$
   algorithm}\label{alg:OneEndtoEndRing}
 \begin{algorithmic}[1]
 \REQUIRE {A supersingular elliptic curve $E/\F_{p^2}$, and a parameter $k > 0$. We suppose access to an oracle $\Oracle$ that solves the $\OneEnd$ problem.}
 \ENSURE {The endomorphism ring $\End(E)$.}
 \STATE{$k_1 \gets \left\lceil\frac{\log\left( 12\cdot 9\cdot(1+\sqrt{3})\cdot\sqrt{p+13} \right)}{\log\left(\frac{3}{2\sqrt 2}\right)}\right\rceil$}\label{step:defk1}
 \STATE $R \gets \Z$
 \WHILE {$\rank_{\Z}(R) \neq 4$} \label{step:OneEndtoEndRing-first-loop}
\STATE $\alpha \gets \Rich^\Oracle_{k_1}(E)$, a random endomorphism of $E$ \COMMENT{Algorithm~\ref{alg:Enrich}}
\STATE $R \gets $ the ring generated by $R$ and $\alpha$
\ENDWHILE\label{step:OneEndtoEndRing-first-loop-end}
\STATE $R \gets \Saturate_2(R)$ \COMMENT{Algorithm~\ref{alg:Saturate}}
\STATE $R \gets \SaturateRam(R)$ \COMMENT{Algorithm~\ref{alg:Saturatep}}
\STATE $[\End(E) : R] \gets \sqrt{\disc(R)}/p$
\STATE Factor $[\End(E) : R] = \prod_{i = 1}^tN_i^{e_i}$ where no $N_i$ is a cube
   \COMMENT{a complete prime factorisation is not required; the somewhat trivial
   factorisation $[\End(E) : R] = N_1^{3^n}$ where $N_1^{1/3} \not\in \Z$
   and~$n\ge 0$ is sufficient as a starting point, and the subsequent steps of
   the algorithm may refine it}
\WHILE {$[\End(E) : R] \neq 1$} \label{step:OneEndtoEndRing-second-loop}
\STATE $N \gets N_t$
 \STATE{$k_2 \gets \left\lceil12 \cdot \log\left(4100000\cdot(\log N)^{12}N^{2}\sqrt{p+13}\right) \right\rceil$}\label{step:defk2}
 \STATE Let $\Oracle_N$ the oracle which given $E$, runs $\alpha \gets \Oracle(E)$ and returns $\Reduce_N(\alpha)$\;
\STATE $\alpha_i \gets  \Rich^{\Oracle_N}_{k_2}(E)$ for $i \in \{1,2,3\}$, random endomorphisms of $E$ \COMMENT{Algorithm~\ref{alg:Enrich}}
\STATE $\Lambda \gets $ the lattice generated by $(1,\alpha_1,\alpha_2,\alpha_3)$
\IF {$\rank_{\Z}(\Lambda) = 4$}
\STATE $n \gets$ the largest integer such that $N^n$ divides $[\End(E) : \Lambda]$
\STATE $d \gets \gcd([\End(E) : \Lambda] / N^n, N)$
\IF {$d \neq 1$}
\STATE Update the factorisation of $[\End(E) : R]$ with $N = d \cdot (N/d)$ \label{step:OneEndtoEndRing-first-success}
\ENDIF
\IF {$\Lambda \not \subset R$}
\STATE $R \gets$ the order generated by $R$ and $\Lambda$ \label{step:OneEndtoEndRing-second-success}
\STATE Recompute $[\End(E) : R] = \sqrt{\disc(R)}/p$, and update its factorisation
\ENDIF
\ENDIF
\ENDWHILE\label{step:OneEndtoEndRing-second-loop-end}
\RETURN $R$
\end{algorithmic}
 \end{algorithm}

%
%
%

\begin{lemma}\label{lem:enrich}
Let $\Oracle$ be an oracle for $\OneEnd$, and $N$ an odd integer. Let $\Oracle_N$ be the oracle which on input $E$, samples $\alpha \gets \Oracle(E)$, and returns $\Reduce_N(\alpha)$. For any \[
    k \geq 12 \cdot \log\left(4100000\cdot(\log N)^{12}N^{2}\sqrt{p+13}\right),
  \] the output of $\Rich_k^{\Oracle_N}$ satisfies the conditions of Proposition~\ref{prop:good-distribution-leads-to-success}.
\end{lemma}

\begin{proof}
Let $\varphi \colon E \to E'$ of degree a power of $2$. For any endomorphism $\beta \in \End(E')$, since $N$ is odd, we have that $\beta$ is $N$-reduced if and only if $\hat \varphi \circ \beta \circ \varphi$ is $N$-reduced.
The output of $\Rich_k^{\Oracle_N}$ is of the form $\hat \varphi \circ \Reduce_N(\alpha) \circ \varphi$, so is $N$-reduced.
So the distribution of $\Rich_k^{\Oracle_N}$ satisfies Item~\ref{prop:good-distribution-leads-to-success-item-1} of Proposition~\ref{prop:good-distribution-leads-to-success}.

From Theorem~\ref{thm:rich-is-conj-invariant}, $\Rich_k^{\Oracle_N} \bmod N$ is $\varepsilon$-close to $\SL_2(\Z/N\Z)$-invariant for
  \[
    \varepsilon = \frac{1+\sqrt{3}}{4}\left(\frac{2\sqrt{2}}{3}\right)^kN^{2}\sqrt{p+13}.
  \]
With
    \[
    k \geq \frac{\log\left(6000000\cdot(\log N)^{12}\cdot\frac{1+\sqrt{3}}{4}N^{2}\sqrt{p+13}\right)}{\log\left(\frac{3}{2\sqrt{2}}\right)},
  \]
 we have $\varepsilon \leq (6000000\cdot(\log N)^{12})^{-1}$, satisfying Item~\ref{prop:good-distribution-leads-to-success-item-2} of Proposition~\ref{prop:good-distribution-leads-to-success}.
\qed\end{proof}

We now have all the ingredients to prove our main result.

\begin{theorem}[$\EndRing$ reduces to $\OneEnd$]\label{thm:main-more-precise}
  Algorithm~\ref{alg:OneEndtoEndRing} is a
  reduction from $\EndRing$ to~$\OneEnd_\lambda$ of expected polynomial time in $\log(p)$ and $\lambda(\log p)$.
\end{theorem}
\begin{proof}
The correctness is clear as at any time, $R$ is a subring of $\End(E)$, and the success condition $[\End(E) : R] = 1$ implies $R = \End(E)$.

We now analyse the expected running time.

\noindent\textbf{First loop (Step~\ref{step:OneEndtoEndRing-first-loop} to Step~\ref{step:OneEndtoEndRing-first-loop-end}).}
First, let us analyse the expected number of iterations of the first loop. 
From Theorem~\ref{thm:rich-is-conj-invariant}, each $\alpha$ generated during this loop is $\varepsilon$-close to $\SL_2(\F_3)$-invariant with
\[\varepsilon = \frac{1+\sqrt{3}}{4}\left(\frac{2\sqrt 2}{3}\right)^{k_1}3^2\sqrt{p+13}.\]
Choosing $ k_1 = O(\log p)$ as in Step~\ref{step:defk1},
we have $\varepsilon \leq 1/48$.

Consider any three consecutively generated elements $\alpha_1, \alpha_2, \alpha_3$. Let $t = \max_i \level_3(\alpha_i)$, and $\beta_i = 3^{t - \level_3(\alpha_i)}\alpha_i$, so all $\beta_i$ are at the same level $t$. Like the variables $\alpha_i$, the variables $\beta_i$ are $\varepsilon$-close to $\SL_2(\F_3)$-invariant. Combining Proposition~\ref{prop:probability-new-samples-local-2} and Proposition~\ref{prop:special-Nakayama}, the tuple $(1,\beta_1, \beta_2, \beta_3)$ generates a full-rank lattice with probability at least $1/8 - 3\varepsilon$, and so does $(1,\alpha_1, \alpha_2, \alpha_3)$.
Choosing $k_1$ as above, this probability is at least $1/16$.
We deduce that the loop terminates after an expected $O(1)$ number of iterations.

Let us now analyse the output of this loop. Let $R_1$ be the order $R$ obtained
  at the end of the first loop. Let $\alpha_i$ be any three elements generated
  during the loop such that $(1,\alpha_1,\alpha_2,\alpha_3)$ are independent.
  Combining the bound $\deg(\alpha_i) \leq 2^{2k_1\lambda(\log p)}$ and Hadamard's inequality, we get
\[\disc(R_1) = 16\cdot \Vol(R_1)^2 \leq 16\cdot \prod_{i=1}^3\sqrt{\deg(\alpha_i)}
\leq 16 \cdot  2^{6k_1\lambda(\log p)}.\]
We deduce that 
\begin{equation}\label{eq:bound-disc-1}
[\End(E) : R_1] \leq 2^{3k_1\lambda(\log p)+2} /p =  2^{O(\log (p)\cdot \lambda(\log p))}.
\end{equation}

\noindent\textbf{Second loop (Step~\ref{step:OneEndtoEndRing-second-loop} to Step~\ref{step:OneEndtoEndRing-second-loop-end}).}
It remains to analyse the second loop.
An iteration of this loop is a \emph{success} if either Step~\ref{step:OneEndtoEndRing-first-success} or Step~\ref{step:OneEndtoEndRing-second-success} is reached. In case of success, either a new factor of $[\End(E) : R]$ is found (Step~\ref{step:OneEndtoEndRing-first-success}), or $[\End(E) : R]$ gets divided by an integer at least $2$ (Step~\ref{step:OneEndtoEndRing-second-success}). The number of successes is thus polynomially bounded in $\log([\End(E) : R_1])$, hence in $\poly(\log p, \lambda(\log p))$ (thanks to Equation~\eqref{eq:bound-disc-1}).
Therefore, we only have to prove that as long as $R \neq \End(E)$, each iteration has a good probability of success.

The event analysed in Proposition~\ref{prop:good-distribution-leads-to-success} corresponds precisely to a success.
By Lemma~\ref{lem:enrich}, the distribution of $\alpha_i$ satisfies the conditions of Proposition~\ref{prop:good-distribution-leads-to-success}. 
Therefore, Proposition~\ref{prop:good-distribution-leads-to-success} implies that each iteration has a probability of success $\Omega((\log N)^{-12})$, which concludes the proof.
\qed\end{proof}

 
%
%

\section{Applications}

In this section we describe four applications of our main result.

\subsection{Collision resistance of the Charles--Goren--Lauter hash function}\label{subsec:CGL}

The first cryptographic construction based on the supersingular isogeny problem is a hash function proposed by Charles, Goren and Lauter~\cite{CGL09}, the \emph{CGL hash function}.
Fix a (small) prime number $\ell$, typically $\ell = 2$. For any elliptic curve $E$, there are $\ell+1$ outgoing $\ell$-isogenies $E \to E'$ (up to isomorphism of the target), so given a curve and an incoming $E''\to E$, there remain $\ell$ non-backtracking $\ell$-isogenies from $E$, which can be arbitrarily labelled by the set $\{0,\dots,\ell-1\}$. Then, fixing an initial curve $E_0$ and an arbitrary isogeny $E_{-1} \to E_0$, the set $\{0,\dots,\ell-1\}^*$ encodes non-backtracking paths from $E_0$ in the $\ell$-isogeny graph. 
The CGL hash function
\[\CGL_{E_0} : \{0,\dots,\ell-1\}^* \longrightarrow \F_{p^2}\]
associates to any sequence $(x_i)_i$ the $j$-invariant of the endpoint of the walk from $E_0$ it encodes.
Clearly, this function is pre-image resistant if and only if $\IsogPath$ is hard. However, if $\End(E_0)$ is known, one can find collisions in polynomial time~\cite{KLPT14,EHLMP18}. Therefore, it was proposed to sample the starting curve randomly. Let $\algorithmname{SampleSS}(p)$ be an algorithm sampling a uniformly random supersingular elliptic curve over $\F_{p^2}$. 
We define the advantage of a collision-finding algorithm $\mathscr A$ for the CGL family of hash functions as
$$\mathrm{Adv}_{\CGL}^{\mathscr A}(p) = 
\Pr\left[
\begin{array}{cc|cc}
m \neq m' \text{ and } &&& E \gets \algorithmname{SampleSS}(p)\\
\CGL_E(m) = \CGL_E(m') &&& (m,m') \gets \mathscr A(E)
\end{array}
\right].$$
It was heuristically argued in~\cite{EHLMP18} that the collision resistance of this construction is equivalent to $\EndRing$. The flaws of the heuristics are discussed in Section~\ref{subsec:techoverview}. With our main theorem, we can now prove this resistance.

\begin{theorem}[Collision resistance of the CGL hash function]\label{thm:CGLcollres}
For any algorithm $\mathscr A$, there is an algorithm to solve $\EndRing$ in expected polynomial time in $\log(p)$, in $\mathrm{Adv}_{\CGL}^{\mathscr A}(p)^{-1}$ and in the expected running time of $\mathscr A$.
\end{theorem}

\begin{proof}
Since $\EndRing$ is equivalent to $\OneEnd$ (Theorem~\ref{thm:main}), it is sufficient to prove that $\mathscr A$ can be used to solve $\OneEnd$. 
First, let us prove that a successful collision for $\CGL_E$ gives a non-scalar endomorphism of $E$. Let $\varphi,\psi \colon E \to E'$ be two distinct non-backtracking walks, i.e., isogenies of cyclic kernel of order~$\ell^a$ and $\ell^{b}$ respectively. 
If $\hat\varphi \circ \psi$ is scalar, the degrees imply that $a+b$ is even and $\hat\varphi \circ \psi = [\ell^{\frac{a+b}{2}}]$. Without loss of generality, suppose $b\geq a$. From the defining property of the dual isogeny, we deduce that $\hat\psi = [\ell^{\frac{b-a}{2}}]\hat\varphi$. Taking the dual again, we get $\psi = [\ell^{\frac{b-a}{2}}]\varphi$. If $b > a$, then $\{0_E\} \neq E[\ell^{\frac{b-a}{2}}] \subseteq \ker \psi$, contradicting the cyclicity of $\ker \psi$. Therefore $b = a$, and we conclude that $\psi = \varphi$, a contradiction. So $\hat\varphi \circ \psi$ is non-scalar.

Now, given a curve $E$, we can solve $\OneEnd$ as follows:
\begin{enumerate}
\item First take a random walk $\eta\colon E \to E'$, so that $E'$ has statistical distance $\varepsilon = O(1/p)$ from uniform (Proposition~\ref{prop:rand-walk-standard});
\item Then call $\mathscr A(E')$, which gives a non-scalar endomorphism $\alpha$ of $E'$ with probability at least $\mathrm{Adv}_{\CGL}^{\mathscr A}(p) - \varepsilon$,
\item Return $\hat \eta \circ \alpha \circ \eta$. 
\end{enumerate}
The algorithm is successful after an expected $(\mathrm{Adv}_{\CGL}^{\mathscr A}(p) - \varepsilon)^{-1}$ number of attempts.
This works within the claimed running time if $\mathrm{Adv}_{\CGL}^{\mathscr A}(p) > 2\varepsilon$. Otherwise, we have $(\mathrm{Adv}_{\CGL}^{\mathscr A}(p))^{-1} = \Omega(p)$, and one can indeed solve $\EndRing$ in time polynomial in $p$ (see~\cite[Theorem~75]{Kohel96} for the first such algorithm, in time $\tilde O(p)$, or Theorem~\ref{thm:solvingEndRing} below for time $\tilde O(p^{1/2})$).
\qed\end{proof}

\subsection{Soundness of the SQIsign identification scheme}\label{subsec:SQIsign}

SQIsign is a digital signature scheme proposed in~\cite{DFKLPW20}. SQIsign, and its variant SQIsignHD~\cite{sqisignhd}, offer the most compact public keys and signatures of all known post-quantum constructions. Note that the results from this section apply equally to SQIsign and SQIsignHD.
This digital signature scheme is constructed as an identification protocol, turned into a signature by the Fiat--Shamir transform. The protocol proves knowledge of a witness for a problem that closely resembles $\OneEnd$. While~\cite{DFKLPW20} heuristically argues that the protocol is sound if $\EndRing$ is hard, our main theorem allows us to prove it.

Let $\algorithmname{SQIsign.param}$ be the SQIsign public parameter generation procedure, which on input a security level $k$, outputs data $\mathsf{pp}$ which encodes, among other things, a prime number $p = \Theta(2^{2k})$.
Let $\algorithmname{SQIsign.keygen}$ be the SQIsign key generation procedure, which on input $\mathsf{pp}$, outputs a pair $(\mathsf{pk},\mathsf{sk})$. The public key $\mathsf{pk}$ is a supersingular elliptic curve over $\F_{p^2}$, and $\mathsf{sk}$ is its endomorphism ring.

Let $\mathscr V$ be an honest verifier for the SQIsign identification protocol.
For any (malicious) prover $\mathscr P^*$ and parameters $\mathsf{pp}$, run the following experiment: first, sample a key pair $(\mathsf{pk},\mathsf{sk}) \gets \algorithmname{SQIsign.keygen}(\mathsf{pp})$, and give $\mathsf{pk}$ to $\mathscr P^*$. Then, run the SQIsign identification protocol between $\mathscr P^*$ and $\mathscr V$ with input $\mathsf{pk}$. 
Let $\pi^{\mathscr P^*}(\mathsf{pp})$ be the probability that $\mathscr V$ outputs $\top$ at the end of the protocol.
We define the \emph{soundness advantage} $\mathrm{Adv}_{\mathrm{SQIsound}}^{\mathscr P^*}(\mathsf{pp}) = \pi^{\mathscr P^*}(\mathsf{pp}) - 1/c$, where $c = \Theta(2^k)$ is the size of the challenge space.

In other words, $\pi^{\mathscr P^*}(\mathsf{pp})$ is the probability that $\mathscr P^*$ successfully fools an honest verifier, for a random key. Since there is a simple malicious prover achieving $\pi^{\mathscr P^*}(\mathsf{pp}) = 1/c$ (by guessing the challenge at the start of the protocol), the advantage $\mathrm{Adv}_{\mathrm{SQIsound}}^{\mathscr P^*}(\mathsf{pp})$ measures how much better $\mathscr P^*$ performs.

\begin{theorem}[Soundness of SQIsign]\label{thm:SQIsignsound}
Let $\mathscr P^*$ be a malicious prover.
Consider public parameters $\mathsf{pp}$, encoding the prime $p$.
There is an algorithm to solve $\EndRing$ for curves over $\F_{p^2}$ in expected polynomial time in $\log(p)$, in $\mathrm{Adv}_{\mathrm{SQIsound}}^{\mathscr P^*}(\mathsf{pp})^{-1}$ and in the expected running time of $\mathscr P^*$.
\end{theorem}

\begin{proof}
Let $r$ denote the expected running time of $\mathscr P^*$.
From~\cite[Theorem~1]{DFKLPW20}, there is an algorithm of expected running time
$$r' = O\left(\frac{r}{\mathrm{Adv}_{\mathrm{SQIsound}}^{\mathscr P^*}(\mathsf{pp})}\right)$$
for the \emph{supersingular smooth endomorphism problem}, with solutions of length $O(\log p)$. The supersingular smooth endomorphism problem is defined as $\OneEnd$ with the additional contraint that the output has smooth degree. Therefore, the same algorithm solves $\OneEnd$. The result follows from the equivalence between $\OneEnd$ and $\EndRing$ (Theorem~\ref{thm:main}).
\qed\end{proof}

The same theorem is true with the quantity $\pi^{\mathscr P^*}(\mathsf{pp})^{-1}$ in place of $\mathrm{Adv}_{\mathrm{SQIsound}}^{\mathscr P^*}(\mathsf{pp})^{-1}$, which may be more natural. In the proof, we would get
$$r' = O\left(\frac{r}{\pi^{\mathscr P^*}(\mathsf{pp}) - 1/c}\right).$$
Note that $r'$ is not necessarily polynomial in $\pi^{\mathscr P^*}(\mathsf{pp})^{-1}$, as $\pi^{\mathscr P^*}(\mathsf{pp})$ could be arbitrarily close to $1/c$. 
We can consider two cases: first, if $\pi^{\mathscr P^*}(\mathsf{pp}) > 2/c$, we have $\pi^{\mathscr P^*}(\mathsf{pp}) - 1/c > \pi^{\mathscr P^*}(\mathsf{pp})/2$ and we can conclude as above.
Second, if $\pi^{\mathscr P^*}(\mathsf{pp}) \leq 2/c$, then $\pi^{\mathscr P^*}(\mathsf{pp})^{-1} \geq c/2 = \Omega(2^k) = \Omega(p^{1/2})$, so we can conclude from the fact that there exists an algorithm for $\EndRing$ in expected time $\tilde O(p^{1/2})$ (see Theorem~\ref{thm:solvingEndRing}).

\subsection{The endomorphism ring problem is equivalent to the isogeny problem} \label{subsec:Isogeny-EndRing}
It is known that the problem $\EndRing$ is equivalent to the $\ell$-isogeny path problem (assuming the generalised Riemann hypothesis~\cite{Wes21}). 
The same technique shows that $\EndRing$ is equivalent to the problem of finding isogenies of \emph{smooth} degree. 
Lifting this restriction yields the more general $\Isogeny$ problem.

\begin{prob}[$\Isogeny$]
Given a prime $p$ and two supersingular elliptic curves $E$ and $E'$ over $\F_{p^2}$, find an isogeny from $E$ to $E'$ in efficient representation.
\end{prob}

Given a function $\lambda \colon \Z_{>0} \to \Z_{>0}$, the $\Isogeny_\lambda$ problem denotes the $\Isogeny$ problem where the solution $\varphi$ is required to satisfy $\log (\deg \varphi) \leq \lambda(\log p)$ (the length of the output is bounded by a function of the length of the input).
 
From previous literature, it is easy to see that $\Isogeny$ reduces to $\EndRing$.

\begin{proposition}[$\Isogeny$ reduces to $\EndRing$]
Assuming the generalised Riemann hypothesis, the problem $\Isogeny_{\lambda}$ reduces to $\EndRing$ in probabilistic polynomial time (with respect to the length of the instance), for some function $\lambda(\log p) = O(\log p)$.
\end{proposition}

\begin{proof}
$\Isogeny$ immediately reduces to $\IsogPath$.
It is already known that the $\ell$-isogeny path problem (with paths of length $O(\log p)$) is equivalent to $\EndRing$~\cite{Wes21}, so $\Isogeny_\lambda$ reduces to $\EndRing$.
\qed\end{proof}

The converse reduction is trickier. As a solution to $\Isogeny$ is not guaranteed to have smooth degree, previous techniques have failed to prove that it is equivalent to $\EndRing$. Theorem~\ref{thm:main} unlocks this equivalence. Better yet, contrary to previous results of this form, Theorem~\ref{thm:EndRing-Isogeny} below is unconditional.
In particular, it implies that $\EndRing$ reduces to the $\ell$-isogeny path problem independently of the generalised Riemann hypothesis.

\begin{theorem}[$\EndRing$ reduces to $\Isogeny$]\label{thm:EndRing-Isogeny}
Given an oracle for $\Isogeny_\lambda$, there is an algorithm for $\EndRing$ that runs in expected polynomial time in $\log(p)$ and $\lambda(\log p)$.
\end{theorem}

\begin{algorithm}[h]
  \caption{Solving $\OneEnd$ given an $\Isogeny$ oracle.}\label{alg:OneEnd-to-Isogeny}
  \begin{algorithmic}[1]
    \REQUIRE {A supersingular elliptic curve $E/\F_{p^2}$, a parameter $\varepsilon >0$, an oracle $\mathscr O_\Isogeny$ solving the $\Isogeny_\lambda$ problem.}
    \ENSURE {An endomorphism $\alpha \in \End(E)\setminus \Z$ in efficient representation.}
    \STATE $P \gets$ an arbitrary nonzero point in $E[2]$
    \STATE $n \gets \lceil 2\log_3(p) - 4\log_3(\varepsilon)\rceil$
    \WHILE {\texttt true}
    	\STATE $\varphi \gets $ a non-backtracking random walk $\varphi \colon E \to E'$ of length $n$ in the $3$-isogeny graph
    	\STATE $\nu \gets$ the isogeny $\nu \colon E' \to E''$ of kernel $\langle\varphi(P)\rangle$
    	\STATE $\psi \gets \mathscr O_\Isogeny(E'',E)$, an isogeny $\psi\colon E'' \to E$
    	\STATE $\alpha \gets (\psi \circ \nu \circ \varphi)/2^e \in \End(E)$ for the largest possible $e$\label{step:div-by-pow-2}
    	\IF {$2 \mid \deg(\alpha)$}\label{step:check-div-by-pow-2}
    		\RETURN $\alpha$\label{step:returning-end}
    	\ENDIF
    \ENDWHILE
  \end{algorithmic}
\end{algorithm}

\begin{proof}
Since $\EndRing$ is equivalent to $\OneEnd$ (Theorem~\ref{thm:main}), let us prove that $\OneEnd$ reduces to $\Isogeny$. Suppose we have an oracle $\Oracle_\Isogeny$ for $\Isogeny_\lambda$. Let $E$ be a supersingular curve for which we want to solve $\OneEnd$. Consider a parameter $\varepsilon$.
The reduction is described in Algorithm~\ref{alg:OneEnd-to-Isogeny}.
Step~\ref{step:div-by-pow-2} and Step~\ref{step:check-div-by-pow-2} ensure that $\alpha$ is not a scalar (indeed, they ensure that upon return, at Step~\ref{step:returning-end}, we have $2 \nmid \alpha$ yet $2 \mid \deg(\alpha)$), so is a valid solution to $\OneEnd$.

Let us show that the expected number of iterations of the while-loop is $O(1)$.  Let $f \in \Z$ maximal such that $E''[2^f] \subseteq \ker(\psi)$, and let $\psi' = \psi/2^f$. If $\deg(\psi')$ is odd, then $\alpha$ is non-scalar (its degree is divisible by $2$ but not by $4$) and the loop terminates at this iteration. Now, suppose $\deg(\psi')$ is even and write $\ker(\psi') \cap E''[2] = G_\psi$, a group of order $2$.
The loop in the reduction terminates in the event that $\ker \hat\nu \neq G_\psi$. In the rest of the proof, we bound the probability of this event at each iteration.

Let $P$ be the probability distribution of the pair $(E'',\hat \nu)$, and $Q$ the probability distribution of the pair $(E'',\eta)$ where $\eta$ is uniformly random (among the three $2$-isogenies from $E''$).
Note that by construction, the value $Q(E'',\eta)$ does not depend on $\eta$, and we also write it $\tilde Q(E'')$.
Consider the function $\tau$ defined in~\cite[Lemma~14]{SECUER}.
We have 
\begin{align*}
\tau(p,2,3,k) &= \frac 1 4 (p-1)^{1/2} \left(1 + \sqrt{3}\right)\left(k + \frac 1 2\right) 3^{-k/2} \leq p^{1/2} 3^{-k/4}.
\end{align*}
From~\cite[Lemma~14]{SECUER}, if $\tau(p,2,3,k) \leq \varepsilon$, then the statistical distance $\|P-Q\|_1/2$ is at most $\varepsilon$. This condition is satisfied if the $3$-walk $\varphi$ has length at least 
\begin{align*}
n(p,2,3,\varepsilon) &= \min\{k \mid \tau(p,2,3,k) \leq \varepsilon\}\\
&\leq \min\{k \mid p^{1/2} 3^{-k/4} \leq \varepsilon\}
= 2\log_3(p) - 4\log_3(\varepsilon).
\end{align*}
We deduce that indeed $\|P-Q\|_1 < \varepsilon$, since $\varphi$ has length $\lceil 2\log_3(p) - 4\log_3(\varepsilon) \rceil$.

We now obtain the following bound:
\begin{align*}
\Pr&\left[\ker \hat\nu = G_\psi\right] 
= \sum_{(E'',\hat \nu)}P(E'',\hat \nu)\Pr[\ker \hat\nu = G_\psi \mid (E'',\hat \nu)]\\
&\leq \sum_{(E'',\hat \nu)}(Q(E'',\hat \nu) + \max_{\eta}|P(E'',\eta) - Q(E'',\eta)|)\Pr[\ker \hat\nu = G_\psi \mid E'']\\
&\leq \sum_{E''}\tilde Q(E'') + \sum_{E''}\max_{\eta}|P(E'',\eta) - Q(E'',\eta)| \leq \frac 1 3 + \varepsilon.\\
\end{align*}
The second line uses that for any fixed $E''$, the distribution of $\psi$ is independent of~$\nu$.
In conclusion, at each iteration, the event $\ker \hat\nu \neq G_\psi$ (leading to termination) happens with probability at least $2/3-\varepsilon$. With $\varepsilon < 1/3$, the expected number of iterations is at most $(2/3-\varepsilon)^{-1} \leq 3 = O(1)$.
\qed\end{proof}

\subsection{An unconditional algorithm for $\EndRing$ in time $\tilde O(p^{1/2})$}
As the foundational problem of isogeny-based cryptography, understanding the hardness of $\EndRing$ is critical.
The fastest known algorithms have complexity in $\tilde O(p^{1/2})$, but rely on unproven assumptions such as the generalised Riemann hypothesis. 
With our new results, we can now prove that $\EndRing$ can be solved in time $\tilde O(p^{1/2})$ \emph{unconditionally}.
In contrast, the previous fastest unconditional algorithm had complexity $\tilde O(p)$ and only returned a full-rank subring of the endomorphism ring~\cite[Theorem~75]{Kohel96}.

The first method to reach complexity $\tilde O(p^{1/2})$ under the generalised Riemann hypothesis consists in reducing $\EndRing$ to $\IsogPath$ (via~\cite{Wes21}), and solving $\IsogPath$ by a generic graph path-finding algorithm. Unconditionally, we can follow the same strategy, but using our new reduction from $\EndRing$ to $\IsogPath$ (Theorem~\ref{thm:EndRing-Isogeny}). Let us start by recalling the following folklore solution to $\IsogPath$.

\label{subsec:solving-EndRing}
\begin{proposition}\label{prop:solvingisopath}
Algorithm~\ref{alg:solvingisopath} solves the $\IsogPath$ problem in expected time
  $\poly(\ell,\log p)p^{1/2}$ and returns paths of length $O(\log p)$.
\end{proposition}

\begin{algorithm}[h]
  \caption{Solving $\IsogPath$.}\label{alg:solvingisopath}
  \begin{algorithmic}[1]
    \REQUIRE {Two supersingular elliptic curves $E_0/\F_{p^2}$ and $E_1/\F_{p^2}$, a parameter $n$.}
    \ENSURE {An $\ell$-isogeny path $E_0\to E_1$.}
    \STATE $T \gets \emptyset$ an empty hash table
    \WHILE {$\# T < p^{1/2}$}
    	\STATE $\varphi \gets $ a random walk $\varphi \colon E_0 \to E$ of length $n$ in the $\ell$-isogeny graph
    	\IF {$j(E)$ is not the key of any entry in $T$}
    		\STATE Record $\varphi$ in $T$, with key $j(E)$
    	\ENDIF
    \ENDWHILE
    \WHILE {\texttt true}
    	\STATE $\psi \gets $ a random walk $\psi \colon E_1 \to E$ of length $n$ in the $\ell$-isogeny graph
    	\IF {$j(E)$ is the key of a recorded entry $\varphi$ in $T$}
    		\RETURN $\hat\varphi \circ \psi \colon E_0 \to E_1$
    	\ENDIF
    \ENDWHILE
  \end{algorithmic}
\end{algorithm}

\begin{proof}
Algorithm~\ref{alg:solvingisopath} is a standard bi-directional pathfinding algorithm.
Choose the parameter $n$ as in Proposition~\ref{prop:rand-walk-standard}, so that random isogeny paths of length $O(\log p)$ reach a target at statistical distance $O(1/p)$ from uniform.
The $\ell$-isogeny graph has $O(p)$ vertices, and each sampled curve is close to uniform, so the table is complete after $O(p^{1/2})$ iterations of the first loop. By the same token, thanks to the birthday paradox, the second loop finds a matching entry after an expected $O(p^{1/2})$ number of attempts. The factor $\poly(\ell,\log p)$ accounts for the cost of sampling an isogeny path and checking that a candidate is in the table.
\qed\end{proof}

\begin{theorem}\label{thm:solvingEndRing}
There is an algorithm solving $\EndRing$ in expected time $\tilde O\left( p^{1/2}\right)$.
\end{theorem}
\begin{proof}
This follows from the fact that there is an algorithm of complexity $\tilde O\left( p^{1/2}\right)$ for the $2$-isogeny path problem 
(Proposition~\ref{prop:solvingisopath}), and $\EndRing$ reduces 
to polynomially many instances of the $\ell$-isogeny path problem (Theorem~\ref{thm:EndRing-Isogeny}).
\qed\end{proof}

\bibliographystyle{alpha}
\bibliography{bib}

\newpage

\appendix

\section{Illustration of our equidistribution theorem}\label{appen:illustrations}

In this appendix, we give some more examples and comments on the use of
Theorem~\ref{thm:equidistrE} and Proposition~\ref{prop:companionequidistr}.

\subsection{Distributions and functions}\label{appen:dist-func}

Recall the definition of our adjacency operator on~$L^2(\cG_\cF)$:
  \[
    A_\ell F(x) = \sum_{x\to y} F(y),
  \]
where the sum runs over edges of degree~$\ell$ leaving~$x$. This sum always
has~$\ell+1$ terms.
In some situations it is more natural to use a different operator~$B_\ell$:
\[
  B_\ell F(x) = \sum_{x\leftarrow y} F(y)
\]
where the sum runs over edges of degree~$\ell$ arriving at~$x$. This sum may
have fewer than~$\ell+1$ terms, due to automorphisms.

In order to relate the two operators, it is convenient to compute the adjoint
of~$A_\ell$. Let~$\mu$ be the measure on~$\cG_\cF$ (recall~$\mu(x) =
\frac{1}{\#\Aut(x)}$).
We have

\begin{eqnarray*}
  \langle A_\ell F, G\rangle
  &=& \sum_x A_\ell F(x) \overline{G(x)} \mu(x) \\
  &=& \sum_x \sum_{x\to y}F(y) \overline{G(x)} \mu(x) \\
  &=& \sum_y \sum_{x\leftarrow y}F(y) \overline{G(x)} \mu(x) \\
  &=& \sum_y F(y)\left(\frac{1}{\mu(y)}\sum_{x\leftarrow y} \overline{G(x)} \mu(x) \right)\mu(y) \\
  &=& \langle F, A_\ell^* G\rangle,
\end{eqnarray*}
where
\[
  A_\ell^* G(x) = \frac{1}{\mu(x)}\sum_{x\leftarrow y}G(y)\mu(y).
\]
We therefore introduce the ``diagonal'' operator~$M$ on~$L^2(\cG_\cF)$ defined by
\[
  MF(x) = F(x)\mu(x),
\]
so that we have
\[
  A_\ell^* = M^{-1}B_\ell M
  \text{ i.e. }
  B_\ell = MA_\ell^* M^{-1}.
\]
Since~$A_\ell^*$ has the same orthogonal eigenvectors as~$A_\ell$, with complex
conjugate eigenvalues, this gives the spectral decomposition of~$B_\ell$.

For instance, the action of one step of a degree~$\ell$ random walk on a distribution
on~$\cG_\cF$ is given by~$\frac{B_\ell}{\ell+1}$. Therefore, in the case
where~$\cG_\cF^1$ is connected, $L^2_{\deg}(\cG_\cF)$ is one-dimensional generated
by the constant function~$\triv$ equal to~$1$ everywhere, so we obtain that for
every~$F$, the sequence~$(\frac{A_\ell^*}{\ell+1})^k F$ quickly converges to a
constant function, and therefore for every distribution~$f$, the
distribution~$(\frac{B_\ell}{\ell+1})^k f$ obtained after~$k$ steps of the
random walk converges to the distribution~$M\triv=\mu$, i.e. to the stationary
distribution.

\subsection{Explicit orthogonal projection onto~$L^2_{\deg}$}\label{appen:proj}

Another useful tool is the explicit decomposition of functions according
to~$L^2(\cG_\cF) = L^2_{\deg}(\cG_\cF) \oplus L^2_0(\cG_\cF)$.
We first introduce some notation: for two vertices~$x,y$ of~$\cG_\cF$,
write~$x\sim y$ if they are in the same connected component of~$\cG_\cF^1$ (this
is an equivalence relation). Moreover, for every vertex~$x$, let
\[
  W(x) = \sum_{y\sim x}\mu(y).
\]
If~$x\sim y$, then~$W(x)=W(y)$. We now define an operator~$P$ on~$L^2(\cG_\cF)$:
\[
  PF(x) = \frac{1}{W(x)}\sum_{y\sim x}F(y)\mu(y).
\]
The function~$PF$ is clearly in~$L^2_{\deg}(\cG_\cF)$, and
if~$F\in L^2_{\deg}(\cG_\cF)$  then~$PF = F$; therefore~$P$ is a projector
onto~$L^2_{\deg}(\cG_\cF)$. In order to prove that~$P$ is the desired orthogonal
projector, it is enough to show that it is self-adjoint. Let~$F,G\in
L^2(\cG_\cF)$, then
\begin{eqnarray*}
  \langle PF, G \rangle
  &=& \sum_{x}PF(x)\overline{G(x)}\mu(x) \\
  &=& \sum_{x}\frac{1}{W(x)}\sum_{y\sim x}F(y)\mu(y)\overline{G(x)}\mu(x) \\
  &=& \sum_{y}\sum_{x\sim y}\frac{1}{W(x)}F(y)\mu(y)\overline{G(x)}\mu(x) \\
  &=& \sum_{y}\sum_{x\sim y}\frac{1}{W(y)}F(y)\mu(y)\overline{G(x)}\mu(x) \\
  &=& \sum_{y}F(y)\left(\frac{1}{W(y)}\sum_{x\sim y}\overline{G(x)}\mu(x)\right)\mu(y) \\
  &=& \langle F, PG\rangle.
\end{eqnarray*}
So~$P$, being self-adjoint, is the orthogonal projection
onto~$L^2_{\deg}(\cG_\cF)$.
Note that one can also see this formula as the sum of the orthogonal projections
onto the indicator functions of the connected components of~$\cG_\cF^1$.

\subsection{Equidistribution over an entire connected component of~$\cG_\cF$}

Assume that~$\cG_\cF$ is connected (one can always reduce to this case by
considering connected components). The existence of the map~$\Deg$ and the
resulting action of~$A_\ell$ on~$L^2_{\deg}(\cG_\cF)$ often force the
$\ell$-part of~$\cG_\cF$ to be disconnected and/or multipartite, preventing
degree~$\ell$ random walks from properly equidistributing. One may be tempted to
think that there is a fundamental obstruction to the equidistribution of random
walks on the whole of~$\cG_\cF$. We will see here that this is not the case: one
can easily obtain full equidistribution, simply by using several primes for the
random walk (in other words, in general the disconnectedness of~$\cG_\cF$ is the
only fundamental obstruction to equidistribution).

Pick a bound~$X$, and assume that~$\Sigma$ contains every prime~$\ell<X$ that
does not divide~$pN$. Let~$\nb(X)$ denote the number of such primes, so that we
have~$\nb(X)\approx\frac{X}{\log X}$. Define the operator~$\Delta$
on~$L^2(\cG_\cF)$ by
\[
  \Delta = \frac{1}{\nb(X)}\sum_{\ell< X} \frac{A_\ell}{\ell+1}.
\]
The interpretation of~$\Delta$ is that one step of the corresponding random
walk consists in choosing a prime~$\ell<X$ uniformly at random, and then using
one step of the degree~$\ell$ random walk.

Since the~$A_\ell$ are normal operators that pairwise commute, $\Delta$ is also
a normal operator, stabilises~$L^2_{\deg}(\cG_\cF)$ and~$L^2_0(\cG_\cF)$, and is
diagonalisable in the same orthogonal basis as the~$A_\ell$.
We bound its eigenvalues and operator norm.
\begin{itemize}
  \item On~$L^2_0(\cG_\cF)$, the operator norm of~$\Delta$ is bounded by
    \[
      \frac{1}{\nb(X)} \sum_{\ell<X}\frac{2\sqrt{\ell}}{\ell+1}
      \approx \frac{1}{\sqrt{X}}.
    \]
  \item On~$L^2_{\deg}(\cG_\cF)$, there is one eigenvector for each character~$\chi$
    of~$(\Z/N\Z)^\times$ that vanishes on~$\deg(H)$ (with the notations of
    Proposition~\ref{prop:companionequidistr}), with eigenvalue~$1$ if~$\chi$ is the
    trivial character, and otherwise
    \[
      \frac{1}{\nb(X)}\sum_{\ell < X}\chi(\ell) \approx \frac{1}{\sqrt{X}}.
    \]
\end{itemize}
For~$X$ large enough, all the eigenvalues are therefore small, except for the
eigenvalue~$1$ corresponding to the
constant function. Moreover, these approximation can be turned into good bounds
under the generalised Riemann hypothesis.
Therefore, for a moderate value of~$X$, the random walk corresponding
to~$\Delta$ will quickly equidistribute over the whole of~$\cG_\cF$.

\subsection{Example: graphs attached to endomorphisms
modulo~$\ell$}\label{appen:graphs}

It is often convenient to use a functor slightly different
from~$\End/N$, namely the functor~$\cF$ (for~$\Sigma$ the set of all primes not
dividing~$N$) defined by
\begin{itemize}
  \item $\cF(E) = \End(E)/N\End(E)$;
  \item $\cF(\varphi) \colon \alpha \mapsto (\deg
    \varphi)^{-1}\varphi\alpha\hat\varphi$ (which makes sense:
    $\deg\varphi$ is invertible~$\bmod N$).
\end{itemize}
The main reason to prefer this functor is that for every isogeny~$\varphi\colon E
\to E'$, the map~$\cF(\varphi)\colon
\End(E)/N\End(E) \to \End(E')/N\End(E')$ is a ring homomorphism, and therefore
preserves the trace, degree, dual, minimal polynomial and level.
This functor clearly satisfies the~$\modN$-congruence property.
We will describe the connected components of the graphs~$\cG_\cF$,
and~$\cG_\cF^1$, and the graph~$\cG_{\deg}$.

For simplicity, we will assume~$N = \ell$ is an odd prime different from~$p$,
but the other cases can be treated similarly.
We recall the classification of conjugacy classes in~$M_2(\F_\ell)$ and their
centraliser in~$\GL_2(\F_\ell)$:
\begin{enumerate}[(1)]
  \item Homotheties
    \[
      \begin{pmatrix}
        a & 0 \\ 0 & a
      \end{pmatrix}
      \text{ for } a\in\F_\ell.
    \]
    The centraliser~$H$ of such a matrix is~$\GL_2(\F_\ell)$, and~$\det(H) =
    \F_\ell^\times$.
  \item Diagonalisable matrices, with representatives
    \[
      \begin{pmatrix}
        a & 0 \\ 0 & b
      \end{pmatrix}
      \text{ for }a\neq b \in \F_\ell.
    \]
    The centraliser~$H$ of such a matrix is
    \[
      \begin{pmatrix}
        * & 0 \\ 0 & *
      \end{pmatrix}
      \cong \F_\ell^\times \times \F_\ell^\times,
    \]
    and~$\det(H) = \F_\ell^\times$. Two matrices in such a conjugacy class are
    also~$\SL_2(\F_\ell)$-conjugate.
  \item Semisimple, non-diagonalisable matrices, with representatives
    \[
      \Psi(\lambda)
      \text{ for } \lambda\in\F_{\ell^2},
    \]
    where~$\Psi\colon \F_{\ell^2} \to M_2(\F_\ell)$ is the ring homomorphism
    given by the action on an~$\F_\ell$-basis of~$\F_{\ell^2}$.
    The centraliser of such a matrix is~$H = \Psi(\F_{\ell^2}^\times) \cong
    \F_{\ell^2}^\times$ and~$\det(H) = \F_\ell^\times$.
    Two matrices in such a conjugacy class are also~$\SL_2(\F_\ell)$-conjugate.
  \item\label{item:nonsemisimple} Non-semisimple matrices, with representatives
    \[
      \begin{pmatrix}
        a & 1 \\ 0 & a
      \end{pmatrix}
      \text{ for } a\in\F_\ell.
    \]
    The centraliser of such a matrix is
    \[
      H =
      \left\{
      \begin{pmatrix}
        c & d \\ 0 & c
      \end{pmatrix}
      \mid c\in\F_\ell^\times, d\in\F_\ell
      \right\},
    \]
    and~$\det(H) = (\F_\ell^\times)^2$. Each such conjugacy class splits into
    two $\SL_2(\F_\ell)$-conjugacy class, with representatives
    \[
      \begin{pmatrix}
        a & \eps \\ 0 & a
      \end{pmatrix}
      \text{ for } a\in\F_\ell \text{ and
      }\eps\in\F_\ell^\times/(\F_\ell^\times)^2;
    \]
    under conjugacy by an element~$g\in\GL_2(\F_\ell)$, the entry~$\eps$ of the
    representative gets multiplied by~$\det(g)$ modulo~$(\F_\ell^\times)^2$.
\end{enumerate}
In particular, conjugacy classes are completely characterised by the pair
(characteristic polynomial, minimal polynomial).

For every~$E\in\catSS(p)$, choose an
isomorphism~$\psi_E\colon \End(E)/\ell\End(E) \cong M_2(\F_\ell)$.
By Proposition~\ref{prop:companionequidistr}~(\ref{item:action}), two
vertices~$(E,\alpha)$ and~$(E,\beta)$ above the same curve~$E$ are connected if
and only if~$\alpha$ and~$\beta$ are conjugate in~$\End(E)/\ell\End(E)$, if and only
if~$\psi_E(\alpha)$ and~$\psi_E(\beta)$ are conjugate.
Moreover, since any two curves are connected and
the~$\cF(\varphi)$ are ring isomorphisms, the
connected components of~$\cG_\cF$ are exactly the
\[
  \{ (E,\alpha) : \psi_E(\alpha)\in C\}
\]
where~$C$ is a conjugacy class in~$M_2(\F_\ell)$.
Fix~$\eta\in \F_\ell^\times\setminus(\F_\ell^\times)^2$.
We have~$\deg(H) = \F_\ell^\times$ for all types,
except~(\ref{item:nonsemisimple}) where~$\deg(H) = (\F_\ell^\times)^2$.
So the graph~$\cG_{\deg}$ consists in one isolated vertex for each~$C$, except
for type~(\ref{item:nonsemisimple}) which gives a connected component with two
vertices~$\{1,\eta\}$, connected by edges labelled by the elements
of~$\F_\ell^\times\setminus(\F_\ell^\times)^2$.

We now consider~$\cG_\cF^1$.
By Proposition~\ref{prop:companionequidistr}~(\ref{item:action}), two
vertices~$(E,\alpha)$ and~$(E,\beta)$ above the same curve~$E$ are connected
in~$\cG_\cF^1$ if and only if~$\alpha$ and~$\beta$ are conjugate
in~$\End(E)/\ell\End(E)$ by an element of 
degree~$1\bmod\ell$, if and only if~$\psi_E(\alpha)$ and~$\psi_E(\beta)$ are
$\SL_2(\F_\ell)$-conjugate.

For any ring~$R$ and~$g\in R^\times$, let~$\Ad(g)$ denote the endomorphism
of~$R$ given by~$r\mapsto grg^{-1}$.
Let~$E\in\catSS(p)$. By Proposition~\ref{prop:companionequidistr}~(\ref{item:connect}), there
exists~$\varphi\colon E_0\to E$ of degree~$1\bmod \ell$.
The composition~$\psi_E \circ \cF(\varphi)\circ \psi_{E_0}^{-1}$ is an
automorphism of~$M_2(\F_\ell)$, hence of the form~$\Ad(g)$ for
some~$g\in\GL_2(\F_\ell)$. Changing~$\psi_E$ if necessary (by an interior
automorphism), we may assume that~$\det g\in(\F_\ell^\times)^2$.
In this case, this property holds for all~$\varphi'\colon E_0\to E$ of
degree~$1\bmod \ell$:
\begin{eqnarray*}
  \psi_E \circ \cF(\varphi')\circ \psi_{E_0}^{-1}
  &=& \psi_E \circ \cF(\varphi)\circ \cF(\varphi^{-1}\varphi')\circ \psi_{E_0}^{-1} \\
  &=& \psi_E \circ \cF(\varphi)\circ \Ad(\varphi^{-1}\varphi')\circ \psi_{E_0}^{-1} \\
  &=& \psi_E \circ \cF(\varphi)\circ \psi_{E_0}^{-1}\circ \Ad(\psi_{E_0}(\varphi^{-1}\varphi')) \\
  &=& \Ad(g) \circ \Ad(\psi_{E_0}(\varphi^{-1}\varphi')) \\
  &=& \Ad(g\psi_{E_0}(\varphi^{-1}\varphi')),
\end{eqnarray*}
and~$\det\psi_{E_0}(\varphi^{-1}\varphi') = \deg(\varphi^{-1}\varphi') = 1\bmod \ell$.

Now let~$(E_0,\alpha)$ and~$(E,\beta)$ be vertices. Since~$E$ and~$E_0$ are
connected by an isogeny~$\varphi$ of degree~$1\bmod\ell$, the
vertices~$(E_0,\alpha)$ and~$(E,\beta)$ are connected if and only
if~$(E,\cF(\varphi)(\alpha))$ and~$(E,\beta)$ are connected,
if and only if~$\psi_E(\cF(\varphi)(\alpha))$ and~$\psi_E(\beta)$
are~$\SL_2(\F_\ell)$-conjugate, if and only if~$\psi_{E_0}(\alpha)$
and~$\psi_E(\beta)$ are~$\SL_2(\F_\ell)$-conjugate by our assumption
on~$\psi_E$.

In other words, the connected components of~$\cG_\cF^1$ are the
\[
  \{ (E,\alpha) : \psi_E(\alpha)\in C\}
\]
for~$C$ a conjugacy class in~$M_2(\F_\ell)$ not of
type~(\ref{item:nonsemisimple}), and the
\[
  \{ (E,\alpha) : \psi_E(\alpha)\in C_1\}
  \text{ and }
  \{ (E,\alpha) : \psi_E(\alpha)\in C_\eta\}
\]
for~$C$ a conjugacy class in~$M_2(\F_\ell)$ of
type~(\ref{item:nonsemisimple}).
(The point of our assumption was to make sure that components~$C_1$ and~$C_\eta$
are not swapped by an isogeny~$\varphi$ of degree~$1\bmod \ell$.)

\subsection{Example: endomorphism transported by a random walk}

Let us examine the following situation: let~$E_0\in\catSS(p)$
and~$\alpha_0\in\End(E_0)$, let~$\varphi\colon E_0 \to E$ be the
result of a $k$-steps random walk of $2$-isogenies, and let~$\alpha =
\varphi\alpha_0\hat\varphi$. What is the distribution of~$(E,\alpha\bmod N)$
as~$k\to\infty$?
Again, for simplicity we treat the prime case~$N=\ell\neq p$.

We first determine the behaviour of the random walk using the functor~$\cF$ from
Section~\ref{appen:graphs}. We choose~$\psi_E\colon \End(E)/\ell\End(E) \cong
M_2(\F_\ell)$ with the same compatibility condition as in that section.
Let~$C$ be the conjugacy class of the matrix~$\psi_{E_0}(\alpha\bmod\ell)$.

Let~$f$ be a distribution on the vertices of~$\cG_\cF$. As in
Section~\ref{appen:dist-func}, since the effect of one
step of random walk is given by~$f \mapsto \frac{B_2}{3} f$, it is
convenient to encode~$f$ into a function~$F\in L^2(\cG_\cF)$ by the
formula~$F(E,\beta) = f(E,\beta)\mu(E,\beta)$, so that the action is given by~$F
\mapsto \frac{A_2^*}{3}F$. The initial distribution~$f_0$ is defined by
\[
  f_0(E_0,\alpha_0\bmod\ell) = 1
\]
and~$f_0$ is~$0$ everywhere else, so the corresponding initial function~$F_0$ is
defined by
\[
  F_0(E_0,\alpha_0\bmod\ell) = \mu(E_0,\alpha_0\bmod\ell)^{-1}
\]
and~$F_0$ is~$0$ everywhere else. Using the projection formula from
Section~\ref{appen:proj}, we see that~$PF_0$ is the indicator function of the
connected component of~$(E_0,\alpha_0\bmod\ell)$, scaled
by~$W(E_0,\alpha_0\bmod\ell)$.
We use the corresponding function on~$\cG_{\deg}$ to determine the action
of~$\frac{A_2^*}{3}$.

\begin{itemize}
  \item $C$ is not of type~(\ref{item:nonsemisimple}): the
    vertex~$\Deg(E_0,\alpha_0\bmod \ell)$ is an isolated vertex
    in~$\cG_{\deg}$, and the action of~$\frac{A_2^*}{3}$
    is trivial.
  \item $C$ is of type~(\ref{item:nonsemisimple}): the connected component
    of~$\Deg(E_0,\alpha_0\bmod \ell)$ in~$\cG_{\deg}$ consists of two vertices.
    The action of~$\frac{A_2^*}{3}$ is trivial if~$2\in(\F_\ell^\times)^2$, and
    swaps the two vertices otherwise.
\end{itemize}
Coming back to functions on~$\cG_\cF$, we obtain the following.
\begin{itemize}
  \item If~$C$ is not of type~(\ref{item:nonsemisimple}), or if~$2$ is a
    square modulo~$\ell$, or if~$k$ is even, then we have~$(\frac{A_2^*}{3})^kPF_0 =
    PF_0$, i.e. $(\frac{B_2}{3})^kf_0$ is close to the distribution supported on
    the connected component of~$(E_0,\alpha_0\bmod \ell)$ and where the
    probability of~$(E,\beta)$ is proportional to~$\mu(E,\beta)$.
  \item Otherwise, $(\frac{A_2^*}{3})^kPF_0$ is the scaled indicator function
    of the connected component of~$\cG_\cF^1$ corresponding to~$C_{\eta\eps}$,
    where~$\psi_{E_0}(\alpha_0\bmod \ell)\in C_{\eps}$, i.e.
    $(\frac{B_2}{3})^kf_0$ is close to the distribution supported on that
    connected component and where the
    probability of~$(E,\beta)$ is proportional to~$\mu(E,\beta)$.
\end{itemize}

Finally, the actual distribution of~$\alpha\bmod\ell$ is obtained by taking the
distributions described above, and multiplying the corresponding random
endomorphism~$\bmod~\ell$ by~$2^k$ (which usually changes the conjugacy class).
The statistical distance to the distribution obtained by projection
onto~$L^2_{\deg}(\cG_\cF)$ can be estimated using the
eigenvalue bounds of Theorem~\ref{thm:equidistrE}.

%
%
%
%
%

\subsection{Distribution of isogenies produced by random walks}

Let~$\ell$ be a prime and~$E_0\in\catSS(p)$. A natural question is: what is the distribution of the
isogenies~$\varphi\colon E_0 \to E$ produced by a long random $\ell$-isogeny
walk starting from~$E_0$?
Our equidistribution theorem gives nontrivial information about this, in the
following form.
Let~$N\ge 2$ be an integer not divisible by $\ell$.
We are going to look at the distibution of~$\varphi \bmod N \in
\Hom(E_0,E)/N\Hom(E_0,E)$.

Let~$\Sigma$ be the set of all primes not dividing~$N$.
We introduce the functor $\cF = (\Hom(E_0,-)/N)^\times \colon \catSS_\Sigma(p) \to \Sets$:
\begin{itemize}
  \item $\cF(E) = \{\varphi \in \Hom(E_0,E)/N\Hom(E_0,E) \mid \deg(\varphi)\in
    (\Z/N\Z)^\times\}$;
  \item for~$\psi\in\Hom_{\Sigma}(E,E')$, the map~$\cF(\psi)\colon \cF(E) \to \cF(E')$
    is~$\varphi \mapsto \psi \circ \varphi$.
\end{itemize}
The functor~$\cF$ clearly satisfies the $\modN$-congruence property.

The action of the group~$G = (\End(E_0)/N\End(E_0))^\times$ on the set~$\cF(E_0) =
(\End(E_0)/N\End(E_0))^\times$ is the left regular action, and therefore
has a unique orbit with trivial stabiliser~$H = 1$ (we choose~$x = \id$ as the base
point).
Therefore~$\cG_{\deg}$ is the Cayley graph of~$(\Z/N\Z)^\times$, and the
map~$\Deg\colon \cG_\cF \to \cG_{\deg}$ sends~$(E,\varphi)$ to~$\deg\varphi
\bmod N$.

We obtain the following proposition.
\begin{proposition}\label{prop:distr-isog}
  Let~$k\ge 0$. Let~$\varphi\colon E_0 \to E_\varphi$ be the random isogeny obtained by
  a $k$-step $\ell$-isogeny random walk from~$E_0$. Let~$\nu$ be the
  distribution on pairs~$(E,\psi)$ where~$\psi\in\Hom(E_0,E)/N\Hom(E_0,E)$ has
  degree~$\ell^k\bmod N$, up to isomorphism, given by probability proportional
  to~$\frac{1}{\#\Aut(E,\psi)}$.
  Then the statistical distance between the distribution
  of~$(E_\varphi,\varphi\bmod N)$
  and~$\nu$ is at most
  \[
    \frac{1}{2\sqrt{6}}\left(\frac{2\sqrt{\ell}}{\ell+1}\right)^k N^2 \sqrt{p}.
  \]
\end{proposition}
\begin{proof}
  Let~$f\in L^2(\cG_\cF)$ be such that~$f(E_0,\id) = \#\Aut(E_0,\id)$
  and~$f(E,\psi) = 0$ for all other vertices~$(E,\psi)$.
  Then for all~$(E,\psi)\in\cG_\cF$ we have
  \[
    \frac{1}{\#\Aut(E,\psi)}
    \left(\frac{A_\ell^*}{\ell+1}\right)^kf(E,\psi)
    = \Pr[(E_\varphi,\varphi\bmod N) = (E,\psi)]
  \]
  (see Section~\ref{appen:dist-func}).
  Let~$f = f_0+f_1$ with~$f\in L^2_0(\cG_\cF)$ and~$f_1\in L^2_{\deg}(\cG_\cF)$
  be the orthogonal decomposition of~$f$.
  Then~$f_1$ is proportional to the function that takes the value~$1$ on
  all~$(E,\psi)$ with~$\deg\psi = 1 \in (\Z/N\Z)^\times$ and~$0$ elsewhere, and
  \[
    \frac{1}{\#\Aut(E,\psi)}
    \left(\frac{A_\ell^*}{\ell+1}\right)^kf_1(E,\psi)
    = \Pr{}_\nu[(E,\psi)].
  \]
  The statistical distance in the statement is therefore
  \begin{eqnarray*}
    && \sum_{(E,\psi)} \frac{1}{\#\Aut(E,\psi)}
      \left|\left(\frac{A_\ell^*}{\ell+1}\right)^k f_0(E,\psi)\right| \\
    &\le&\left(\sum_{(E,\psi)}\frac{1}{\#\Aut(E,\psi)}\right)^{1/2} 
      \left\|\left(\frac{A_\ell^*}{\ell+1}\right)^k f_0 \right\|
      \text{ (Cauchy--Schwarz)}\\
    &\le& \left(N^4\frac{p-1}{24}\right)^{1/2}
      \left(\frac{2\sqrt{\ell}}{\ell+1}\right)^k
      \|f\| \\
    &\le& N^2\sqrt{\frac{p}{24}}
      \left(\frac{2\sqrt{\ell}}{\ell+1}\right)^k,
  \end{eqnarray*}
  as claimed.
\qed\end{proof}

\begin{corollary}\label{cor:distr-isog}
  Keep the notations of Proposition~\ref{prop:distr-isog}.
  There exists a bound~$n = O(\log_\ell(pN) - \log_\ell(\eps))$ such that for
  all~$E$ and all~$k\ge n$, conditional on~$E_\varphi = E$, the distribution
  of~$\varphi\bmod N$ is $\eps$-close to uniform among isogenies of
  degree~$\ell^k\bmod N$.
\end{corollary}

\subsection{Computation of~$\End(E)$: the obvious~$\tilde O(p^{1/2})$ algorithm works}

Theorem~\ref{thm:solvingEndRing} states that one can compute~$\End(E)$ in
expected time~$\tilde O(p^{1/2})$ unconditionally. However, if we unravel the
reductions leading to this theorem, the resulting algorithm seems needlessly
complicated. Here, we unconditionally show that the obvious collision-based
algorithm to compute endomorphisms also yields~$\End(E)$ in expected
time~$\tilde O(p^{1/2})$.

\begin{algorithm}[h]
  \caption{Finding an endomorphism by collisions}\label{alg:collision-end}
  \begin{algorithmic}[1]
    \REQUIRE {A supersingular elliptic curve $E_0/\F_{p^2}$, a parameter $k$.}
    \ENSURE {An endomorphism of~$E_0$.}
    \STATE $T \gets \emptyset$ an empty hash table
    \WHILE {\texttt true}
    	\STATE $\psi \gets $ a random walk $\psi \colon E_0 \to E$ of length $k$ in the $\ell$-isogeny graph
    	\IF {$j(E)$ is the key of a recorded entry $\varphi$ in $T$}
            \RETURN $\hat\varphi \circ \psi \in\End(E_0)$
        \ELSE
    		\STATE Record $\varphi$ in $T$, with key $j(E)$
    	\ENDIF
    \ENDWHILE
  \end{algorithmic}
\end{algorithm}

An easy consequence of Corollary~\ref{cor:distr-isog} is the following.
\begin{proposition}\label{prop:distr-collision-end}
  There exists a bound~$n = O(\log_\ell(pN))$ such that for all choices of
  parameter~$k\ge n$,
  Algorithm~\ref{alg:collision-end} runs in expected time~$\poly(\ell,\log
  p,k)p^{1/2}$, and
  the endomorphisms produced are $O(1/p)$-close to uniform modulo~$N$ among
  endomorphisms of degree~$\ell^{2k}\in(\Z/N\Z)^\times$.
\end{proposition}

From Proposition~\ref{prop:distr-collision-end}, it is easy to see (by an analysis similar to
Section~\ref{sec:conj-invar} and Section~\ref{sec:the-main-reduction} but
easier) that for polynomial choices of~$k$ with~$\ell\in\{2,3\}$, the endomorphisms output by
Algorithm~\ref{alg:collision-end} generate~$\End(E_0)$ after polynomially many
calls.
A variant also provides a reduction from~$\EndRing$ to~$\Isogeny$ alternative to the
proof of Theorem~\ref{thm:EndRing-Isogeny}.

\end{document}

%% file: macros.tex
\usepackage{amsmath}
\usepackage{amsfonts}
\usepackage{amssymb}
\usepackage{amscd}
\usepackage{hyperref}
\usepackage{algorithmic, algorithm}
\usepackage[all]{xy}
\usepackage{stackengine}

\usepackage{color}
\usepackage[dvipsnames]{xcolor}
\usepackage{graphicx}
\usepackage[capitalize]{cleveref}

\usepackage{mathrsfs} 

\newcommand{\cO}{\mathcal{O}}

\DeclareMathOperator{\eps}{\varepsilon}
\DeclareMathOperator{\Q}{\mathbf{Q}}
\DeclareMathOperator{\R}{\mathbf{R}}
\DeclareMathOperator{\C}{\mathbf{C}}
\DeclareMathOperator{\Z}{\mathbf{Z}}
\DeclareMathOperator{\F}{\mathbf{F}}

\renewcommand{\P}{\mathbf{P}}

\DeclareMathOperator{\Tr}{Tr}

\DeclareMathOperator{\rank}{rank}

\DeclareMathOperator{\Vol}{Vol}
\DeclareMathOperator{\Aut}{Aut}

\DeclareMathOperator{\disc}{disc}

\DeclareMathOperator{\GL}{GL}

\DeclareMathOperator{\SL}{SL}

\DeclareMathOperator{\Hom}{Hom}

\DeclareMathOperator{\id}{id}
\DeclareMathOperator{\poly}{poly}

\newtheorem{prob}[theorem]{Problem} 

\newcounter{tasknumber}
\newcommand{\task}[2][]{%
  \addtocounter{tasknumber}{1}%
  \begin{center}%
  \framebox[1.1\width]{\begin{minipage}{0.9\textwidth}%
  \textbf{Task \arabic{tasknumber}} \textit{\if!#1(unassigned)!\else (#1)\fi}: {#2}%
  \end{minipage}}%
  \end{center}%
}

\newcounter{assumptionnumber}
\newcommand{\assumption}[2][]{%
  \addtocounter{assumptionnumber}{1}%
  \begin{center}%
  \framebox[1.1\width]{\begin{minipage}{0.9\textwidth}%
  \textbf{Assumption \arabic{assumptionnumber}} \textit{\if!#1!\else (#1)\fi}: {#2}%
  \end{minipage}}%
  \end{center}%
}

\newcommand{\authnote}[2][]{\noindent {\if!#1!  {\bf TODO} \else {\small \bf #1} \fi: #2}}

%% file: paper.bbl
\newcommand{\etalchar}[1]{$^{#1}$}
\begin{thebibliography}{HLMW23}

\bibitem[Arp23]{Arpin}
Sarah Arpin.
\newblock Adding level structure to supersingular elliptic curve isogeny
  graphs.
\newblock Preprint arXiv:2203.03531, 2023.
\newblock \url{https://arxiv.org/abs/2203.03531}.

\bibitem[BCC{\etalchar{+}}23]{SECUER}
Andrea Basso, Giulio Codogni, Deirdre Connolly, Luca~De Feo, Tako~Boris
  Fouotsa, Guido~Maria Lido, Travis Morrison, Lorenz Panny, Sikhar Patranabis,
  and Benjamin Wesolowski.
\newblock Supersingular curves you can trust.
\newblock In Carmit Hazay and Martijn Stam, editors, {\em Advances in
  Cryptology - {EUROCRYPT} 2023}, volume 14005 of {\em Lecture Notes in
  Computer Science}, pages 405--437. Springer, 2023.

\bibitem[CD23]{CastryckDecruAttack}
Wouter Castryck and Thomas Decru.
\newblock An efficient key recovery attack on {SIDH}.
\newblock In {\em Advances in cryptology---{EUROCRYPT} 2023. {P}art {V}},
  volume 14008 of {\em Lecture Notes in Comput. Sci.}, pages 423--447.
  Springer, Cham, [2023] \copyright 2023.

\bibitem[CL23]{otherequidistr}
Giulio Codogni and Guido Lido.
\newblock Spectral theory of isogeny graphs.
\newblock Preprint arXiv:2308.13913, 2023.
\newblock \url{https://arxiv.org/abs/2308.13913}.

\bibitem[CLG09]{CGL09}
Denis~X. Charles, Kristin~E. Lauter, and Eyal~Z. Goren.
\newblock Cryptographic hash functions from expander graphs.
\newblock {\em Journal of Cryptology}, 22(1):93--113, Jan 2009.

\bibitem[Del73]{Delignebounds}
Pierre Deligne.
\newblock La conjecture de {Weil}. {I}.
\newblock {\em Publ. Math., Inst. Hautes {\'E}tud. Sci.}, 43:273--307, 1973.

\bibitem[Deu41]{Deuring41}
Max Deuring.
\newblock Die typen der multiplikatorenringe elliptischer funktionenk{\"o}rper.
\newblock {\em Abhandlungen aus dem Mathematischen Seminar der Universit{\"a}t
  Hamburg}, 14(1):197--272, 1941.

\bibitem[DKL{\etalchar{+}}20]{DFKLPW20}
Luca {De Feo}, David Kohel, Antonin Leroux, Christophe Petit, and Benjamin
  Wesolowski.
\newblock Sqisign: Compact post-quantum signatures from quaternions and
  isogenies.
\newblock In Shiho Moriai and Huaxiong Wang, editors, {\em Advances in
  Cryptology - {ASIACRYPT} 2020 - 26th International Conference on the Theory
  and Application of Cryptology and Information Security}, volume 12491 of {\em
  Lecture Notes in Computer Science}, pages 64--93. Springer, 2020.

\bibitem[DLRW23]{sqisignhd}
Pierrick Dartois, Antonin Leroux, Damien Robert, and Benjamin Wesolowski.
\newblock {SQISignHD}: New dimensions in cryptography.
\newblock {IACR} Cryptology ePrint Archive, Report 2023/436, 2023.
\newblock \url{https://eprint.iacr.org/2023/436}.

\bibitem[DV13]{DembeleVoight}
Lassina Demb{\'e}l{\'e} and John Voight.
\newblock Explicit methods for {Hilbert} modular forms.
\newblock In {\em Elliptic curves, Hilbert modular forms and Galois
  deformations.}, pages 135--198. Basel: Birkh{\"a}user/Springer, 2013.

\bibitem[EHL{\etalchar{+}}18]{EHLMP18}
Kirsten Eisentr{\"a}ger, Sean Hallgren, Kristin Lauter, Travis Morrison, and
  Christophe Petit.
\newblock Supersingular isogeny graphs and endomorphism rings: Reductions and
  solutions.
\newblock In Jesper~Buus Nielsen and Vincent Rijmen, editors, {\em Advances in
  Cryptology -- EUROCRYPT 2018}, pages 329--368, Cham, 2018. Springer
  International Publishing.

\bibitem[FIK{\etalchar{+}}23]{EndRingGRH}
Jenny Fuselier, Annamaria Iezzi, Mark Kozek, Travis Morrison, and
  Changningphaabi Namoijam.
\newblock Computing supersingular endomorphism rings using inseparable
  endomorphisms.
\newblock Preprint arXiv:2306.03051, 2023.
\newblock \url{https://arxiv.org/abs/2306.03051}.

\bibitem[HLMW23]{HW23}
Arthur Herl{\'e}dan Le~Merdy and Benjamin Wesolowski.
\newblock The supersingular endomorphism ring problem given one endomorphism.
\newblock Preprint arXiv:2309.11912, 2023.
\newblock \url{https://arxiv.org/abs/2309.11912}.

\bibitem[IR93]{maxord1}
G{\'a}bor Ivanyos and Lajos R{\'o}nyai.
\newblock Finding maximal orders in semisimple algebras over
  {{\(\mathbb{Q}\)}}.
\newblock {\em Comput. Complexity}, 3(3):245--261, 1993.

\bibitem[JL70]{JL}
H.~Jacquet and R.~P. Langlands.
\newblock {\em Automorphic forms on {GL}(2)}, volume 114 of {\em Lect. Notes
  Math.}
\newblock Springer, Cham, 1970.

\bibitem[KLPT14]{KLPT14}
David Kohel, Kristin Lauter, Christophe Petit, and Jean-Pierre Tignol.
\newblock On the quaternion $\ell$-isogeny path problem.
\newblock {\em LMS Journal of Computation and Mathematics}, 17(A):418--432,
  2014.

\bibitem[Koh96]{Kohel96}
David Kohel.
\newblock {\em Endomorphism rings of elliptic curves over finite fields}.
\newblock PhD thesis, University of California, Berkeley, 1996.

\bibitem[LM23]{LeiMueller}
Antonio Lei and Katharina M{\"u}ller.
\newblock On towers of isogeny graphs with full level structure.
\newblock Preprint arXiv:2309.00524, 2023.
\newblock \url{https://arxiv.org/abs/2309.00524}.

\bibitem[Mes86]{Mestre86}
Jean-Francois Mestre.
\newblock La m{\'e}thode des graphes. exemples et applications.
\newblock In {\em Proceedings of the international conference on class numbers
  and fundamental units of algebraic number fields (Katata)}, pages 217--242,
  1986.

\bibitem[ML98]{MacLane}
Saunders Mac~Lane.
\newblock {\em Categories for the working mathematician}, volume~5 of {\em
  Graduate Texts in Mathematics}.
\newblock Springer-Verlag, New York, second edition, 1998.

\bibitem[MMP{\etalchar{+}}23]{MMPPW23}
Luciano Maino, Chloe Martindale, Lorenz Panny, Giacomo Pope, and Benjamin
  Wesolowski.
\newblock A direct key recovery attack on {SIDH}.
\newblock In Carmit Hazay and Martijn Stam, editors, {\em Advances in
  Cryptology - {EUROCRYPT} 2023. {P}art {V}}, volume 14008 of {\em Lecture
  Notes in Computer Science}, pages 448--471. Springer, 2023.

\bibitem[Piz90]{Pizer90}
Arnold~K Pizer.
\newblock Ramanujan graphs and {H}ecke operators.
\newblock {\em Bulletin of the American Mathematical Society}, 23(1):127--137,
  1990.

\bibitem[Rob22]{RobertApplications}
Damien Robert.
\newblock Some applications of higher dimensional isogenies to elliptic curves
  (overview of results).
\newblock Cryptology ePrint Archive, Paper 2022/1704, 2022.
\newblock \url{https://eprint.iacr.org/2022/1704}.

\bibitem[Rob23]{RobertAttack}
Damien Robert.
\newblock Breaking {SIDH} in polynomial time.
\newblock In {\em Advances in cryptology---{EUROCRYPT} 2023. {P}art {V}},
  volume 14008 of {\em Lecture Notes in Comput. Sci.}, pages 472--503.
  Springer, Cham, [2023] \copyright 2023.

\bibitem[Sil86]{Silverman-Arithmetic}
Joseph~H. Silverman.
\newblock {\em The Arithmetic of Elliptic Curves}, volume 106 of {\em Gradute
  Texts in Mathematics}.
\newblock Springer-Verlag, 1986.

\bibitem[Voi13]{maxord2}
John Voight.
\newblock Identifying the matrix ring: algorithms for quaternion algebras and
  quadratic forms.
\newblock In {\em Quadratic and higher degree forms}, pages 255--298. New York,
  NY: Springer, 2013.

\bibitem[Voi21]{VoightBook}
John Voight.
\newblock {\em Quaternion Algebras}.
\newblock Springer International Publishing, 2021.
\newblock Graduate Texts in Mathematics, No. 288.

\bibitem[Wes22a]{Wes22}
Benjamin Wesolowski.
\newblock Orientations and the supersingular endomorphism ring problem.
\newblock In Orr Dunkelman and Stefan Dziembowski, editors, {\em Advances in
  Cryptology -- {EUROCRYPT} 2022}, volume 13277 of {\em Lecture Notes in
  Computer Science}, pages 345--371. Springer, 2022.

\bibitem[Wes22b]{Wes21}
Benjamin Wesolowski.
\newblock The supersingular isogeny path and endomorphism ring problems are
  equivalent.
\newblock In {\em FOCS 2021-62nd Annual IEEE Symposium on Foundations of
  Computer Science}, 2022.

\end{thebibliography}
